\documentclass[12pt]{amsart}
\usepackage{graphicx}
\usepackage{tikz}

\usepackage{graphicx}
\usepackage{epsfig}

\newtheorem{thm}{Theorem}[section]
\newtheorem{prop}[thm]{Proposition}

\newtheorem{lemma}[thm]{Lemma}

\newtheorem{remark}[thm]{Remark}
\newtheorem{example}[thm]{Example}

\usepackage{amsmath}
\usepackage{amsxtra}
\usepackage{amscd}
\usepackage{amsthm}
\usepackage{amsfonts}
\usepackage{amssymb}
\usepackage{eucal}

\usepackage{float}
\usepackage{mathtools}
\usepackage{stix}
\usepackage{setspace}
\usepackage{adjustbox}
\usepackage{subcaption}
\usepackage[toc,page]{appendix}
\usepackage{hyperref}
\hypersetup{colorlinks,linkcolor={red},citecolor={olive},urlcolor={red}}

\newcommand{\bmb}{\left( \begin{array}{rr}}
\newcommand{\enm}{\end{array}\right)}

\newcommand{\Z}{{\mathbb Z}}

\newcommand{\R}{{\mathbb R}}

\newcommand{\al}{{\alpha}}

\numberwithin{equation}{section}

\begin{document}

\title{Limit shapes for Domain-Wall (colored) vertex models}
\author{Philippe Di Francesco and David Keating} 
\address{Department of Mathematics,
University of Illinois at Urbana-Champaign,
1409 West Green Street, Urbana, IL 61801, USA
}

\begin{abstract}
We study partition functions with domain-wall like boundary conditions for path models issued from colored vertex models. These models display an arctic phenomenon, as attested by numerical simulations. We show that the colored vertex model is equivalent to a certain single-color ``colorblind" vertex model. In a special case of the weights for the colorblind touching paths, we derive the arctic curve using a bijective sliding map to non-intersecting paths, for which arctic curves were previously derived using the tangent method. The resulting arctic curves are only piecewise analytic, as in the known non-free fermion cases of Six vertex model with domain-wall boundaries and its relatives. We also prove a shear phenomenon, that some portions of the arctic curve are sheared versions of the analytic continuation of other portions, as already observed in the uniformly weighted Six and Twenty vertex models.
\end{abstract}

\maketitle
\date{\today}
\tableofcontents

\section{Introduction}

In this paper we consider some simple path models borrowed from recent developments on colored vertex models, and study their arctic phenomenon.
The arctic phenomenon occurs in a variety of two-dimensional statistical models, and was first spotted in relation with domino tilings of the Aztec diamond \cite{JPS},
and later generalized to rhombus tilings of simple domains \cite{cohn1998shape}, and more generally dimer models. Looking at typical tiling configurations of large scaled domains, one observes frozen, crystal-like regions near the boundaries of the domain and liquid, disordered regions away form the boundaries, with a phase separation becoming sharper as the size increases. The latter converges to a curve coined the arctic curve. Exact solutions for dimer models using Kasteleyn matrix technology have allowed a thorough investigation of arctic curves (see \cite{KOS,KO1,KO2}). However all these models are so-called ``free fermion" models, that can be described with objects with repulsive interaction, such as non-intersecting paths. For those, the general theory predicts that the arctic curve is analytic. 

In contrast, integrable lattice models from statistical physics often allow for situations where the fermionic descriptions include non-trivial interactions. The most famous such model is the Six Vertex (6V) model on a square with domain-wall boundary conditions, which can be described in terms of osculating paths, namely non-intersecting paths with contact interactions, and includes the case of Alternating Sign Matrices (ASMs), whose limit shape/arctic curve is non-analytic, made of 4 pieces of ellipse with curvature discontinuities at the contact points. The arctic phenomenon for such models was first derived non-rigorously using in particular the tangent method \cite{COSPO}, later proved in \cite{Aggar}. The tangent method consists of slightly modifying the setting of the model, say described by (possibly interacting) paths, by simply moving away from the domain the endpoint of the last path supposed to describe the phase boundary, and using that new path as a revelator of the arctic curve formed by the remaining ones. Indeed, away from the other paths, the most likely trajectory of a path is a line, which is expected to be tangent to the arctic curve at the point at which it leaves it. By determining the parametric family of these tangent lines obtained by moving the endpoint, one recovers the arctic curve as their envelope. Moreover, the tangent lines are determined by a boundary one-point function, in which the last path is constrained to end at a different point on the boundary of the domain. Such functions are easily calculable by utilizing the integrability of the model, and lead to predictions for the arctic curves in excellent agreement with numerics. The tangent method has been validated in free fermion models as well \cite{corteel2021arctic,DFLAP,DFGUI,DFG3}, including the possibility of non-local weights for the paths, such as that involving the area under the path. In the latter case, it was shown that the tangent lines become curved geodesics due to the background created by the area weights, but the method still applies and the arctic curve, previously computed in \cite{BuKni,Petrov}, is still their envelope \cite{DFG2}.

More predictions were made for non-free fermion models, such as the 20V model, related to the 6V model, using the tangent method, all leading to non-analytic arctic curves \cite{BDFG,DF21V,DF23V}.  
In these models, as well as the 6V, in the case of uniform weights, the arctic curve is made of several analytic pieces, some of which are obtained by a shear transformation applied to the analytic continuation of another portion. In particular, in the case of the 6V model with uniform weights (equivalent to ASMs), each piece of ellipse can be obtained as a shear of its neighbors, and a similar property holds in the case of the uniform weight 20V model as well \cite{BDFG,DF21V,DF23V}. In the latter model this shear phenomenon was explained  by using explicit transformations in its description in terms of osculating paths. Similar shear phenomena are seen in models coming from colored paths related to interacting tilings \cite{corteel2022colored,guse2025colored}.

In the present paper, we study another one-parameter ($t$) family of non-free fermion (colored) lattice paths, related to integrable colored vertex models \cite{aggarwal2023colored,corteel2022vertex}. We show that the colored model is equivalent to a colorblind version with specific weights, recently used to generate known families of multivariate symmetric polynomials \cite{garbali2020modified}. We study both the colored and colorblind models with domain-wall type boundary conditions, and we find exact formulas for the partition function in the special cases $t=0$ and $t=1$.

For suitable domain-wall type boundary conditions, we expect an arctic phenomenon. We apply the tangent method to the special case $t=0$ of the colorblind model where path edges are allowed to touch only along one direction, and with more general boundary conditions allowing for any fixed path exit points on the top boundary. We show by construction that paths can be taken apart (and rendered free-fermionic) under a sliding map, which conversely allows to determine their limit shape/arctic curve. Portions of the latter appear naturally in two forms: (1) simple translates of portions of arctic curve of the free fermion model (obtained by the sliding map) or (2) shear transforms of such portions.
We then introduce area weights, for which we determine the limit shapes by use of the curved tangent method of \cite{DFG2}. Remarkably, the above-mentioned shear phenomenon still holds.

Our analytic predictions are confirmed by numerical simulations using random sampling, which we also use to analyze more general colored vertex weights.

The paper is organized as follows. In Section 2, we introduce the colored model and the domain-wall boundary conditions, and introduce an equivalent colorblind version with suitable weights whose partition function is computed for $t=0,1$. When $t=1$ the partition function for the colorblind model is computed by first computing the partition function for the colored model whose paths are independent from one another for this choice of weight. The main property distinguishing the model with $t=0$ from general values of $t$, is that for general values both horizontal and vertical edges can be visited simultaneously by multiple paths, while for $t=0$ only vertical edges can be multiple. In the latter case, we generalize the domain-wall boundary conditions by allowing for arbitrary fixed exit points on the top boundary, a situation already considered in \cite{DFGUI} for free-fermionic non-intersecting paths. We introduce a ``sliding map" between the configurations of the $t=0$ ``touching path" model and those of the non-intersecting one, allowing for the computation of the partition functions. 

Section 3 is devoted to the application of the tangent method to the touching path model. We obtain exact predictions for the arctic curve in terms of an arbitrary given distribution function of endpoints along the top boundary, including the possibility of macroscopic gaps (absence of exit points on some segments) or frozen exits (finite fraction of paths exiting at the same point). We observe a shear phenomenon, explicitly due to the sliding map. This analysis is repeated in Section 4 for touching paths receiving an extra ``area" weight, namely a weight $q$ per unit square between the path and the West vertical boundary. We find that the shear phenomenon still holds in this case.

In Section 5, we give numerical evidence by simulating large configurations of the models, for various distributions of exit points. We first describe our sampling algorithm, and then show simulation results. We extend the study to colored models with $t\neq 0$, and discuss the results.

\noindent{\bf Acknowledgments.} PDF is partially supported by the Morris and Gertrude Fine endowment and the Simons Foundation travel grant MPS-TSM00002262. DK is supported by the NSF RTG grant DMS-1937241.

\section{Colored vertex model: colorblind weights}

\subsection{The colored model: partition functions}

Here we define the colored vertex model and colored domain-wall boundary conditions.  Consider a square lattice vertex model with $n$ colors of paths. We label the colors 1 to $n$. For $s_N,s_S,s_E,s_W\subset \{1,2,\ldots,n\}=:[n]$ such that $s_S \cap s_W = s_N \cap s_E = \emptyset$ and $s_S \cup s_W = s_N \cup s_E$ the vertex weights are given by
\begin{equation}\label{eq:coloredweights}
    \tilde w(s_W,s_S,s_N,s_E; x, t) := x^{|s_E|} t^{d(s_N,s_E)}
\end{equation}
where $d(s_N,s_E):= \#\{i,j | i<j,\, i\in s_E,\, j \in s_E \cup s_N \}$ and the weights are 0 otherwise. In words,
\begin{itemize}
    \item $s_W, s_S$ describe which colors of paths are entering the vertex from the west or south, respectively,
    \item $s_N,s_E$ describe which colors of path are exiting the vertex to the north or east, respectively,
    \item we get a factor of $x$ in the weight for each path exiting east,
    \item and when a color exits to the east we get a factor of $t$ in the weight for each larger color present in the vertex.
\end{itemize}
 Note that, for it to have non-zero weight, at most one path of a each color can appear in a vertex. These vertex weights were first introduced in \cite{aggarwal2023colored,corteel2022vertex} although here, unlike in those works, we restrict to having only one path of each color.

Consider the model on an $n\times n$ square grid with colored domain-wall type boundary conditions we now describe. Fix a permutation $\sigma\in S_n$, the set of permutations of $[n]$. At every horizontal edge on the west boundary a path enters the domain and we color them 1 to $n$ from top to bottom. At each vertical edge at the north boundary a path exits the domain and we color them $\sigma^{-1}_1$ to $\sigma^{-1}_n$ from left to right. All other boundary edges have no paths entering or exiting. See Figure \ref{fig:coloredDW} for an example.

\begin{figure}
    \centering
    \[
        \begin{tikzpicture}
        \draw[thin] (1,0)--(4,0)--(4,3)--(1,3)--(1,0);
        \draw[red, thick] (0,0)--(1,0); \draw[red, thick] (2,3)--(2,4);
        \draw[orange, thick] (0,1)--(1,1);
        \draw[orange, thick] (4,3)--(4,4);
        \draw[green, thick] (0,2)--(1,2);
        \draw[green, thick] (1,3)--(1,4);
        \draw[blue, thick] (0,3)--(1,3);
        \draw[blue, thick] (3,3)--(3,4);
        \end{tikzpicture}
        \qquad
        \begin{tikzpicture}
        \draw[red, thick] (0,0)--(1.1,0)--(1.1,2.1)--(2,2.1)--(2,4);
        \draw[orange, thick] (0,1)--(1,1)--(1,2)--(4,2)--(4,4);
        \draw[green, thick] (0,2)--(0.9,2)--(0.9,4);
        \draw[blue, thick] (0,3)--(3,3)--(3,4);
        \end{tikzpicture}
    \]
    \caption{Example colored domain-wall boundary condition and a possible configuration for $n=4$.  The colors are order {\color{blue} blue} $<$ {\color{green} green} $<$ {\color{orange} orange} $<$ {\color{red} red}. Here $\sigma=(3,1,4,2)$. The coloring on the top boundary is given by $\sigma^{-1}=(2,4,1,3)$. The weight of the configuration is $x_3^4x_4^2t^4$.}
    \label{fig:coloredDW}
\end{figure}
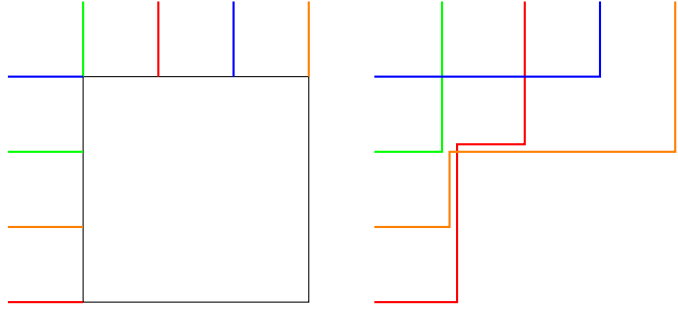

For fixed permutation $\sigma$, we denote the partition function of the model with fixed boundary coloring by
\begin{equation}\label{eq:Zsig}
    \tilde Z^{\sigma}_n(x_1,\ldots,x_n|t) := \sum_{{\rm path}\, {\rm configs.}\atop} \prod_{i=1}^n \prod_{j=1}^n \tilde w(s_W(i,j),s_S(i,j),s_N(i,j),s_E(i,j);x_i,t)
\end{equation}
where vertices are labeled by pairs $(i,j)\in [1,n]\times [1,n]$, enumerated from bottom to top and left to right, and the vertices in the $i$-th row use parameter $x_i$.

We will also consider boundary conditions in which we do not fix the coloring of the paths exiting at the north boundary in the domain and instead allow it to be free. Let $\tilde Z_n(x_1,\ldots,x_n|t)$ be the partition function for the model with free boundary coloring. We have
\begin{equation}\label{eq:Zfree}
    \tilde Z_n(x_1,\ldots,x_n|t) := \sum_{\sigma\in S_n} \tilde Z^{\sigma}_n(x_1,\ldots,x_n|t).
\end{equation}

Figure \ref{fig:coloredEx} shows configurations approximately sampled from the Gibbs measure in the free-color model for several choices of $t$. Further discussion of the simulations is given in Section \ref{sec:simulations}.

\begin{figure}
    \centering
    \resizebox{0.9\textwidth}{!}{
    \begin{tabular}{cccc}
    \includegraphics[width=0.3\linewidth]{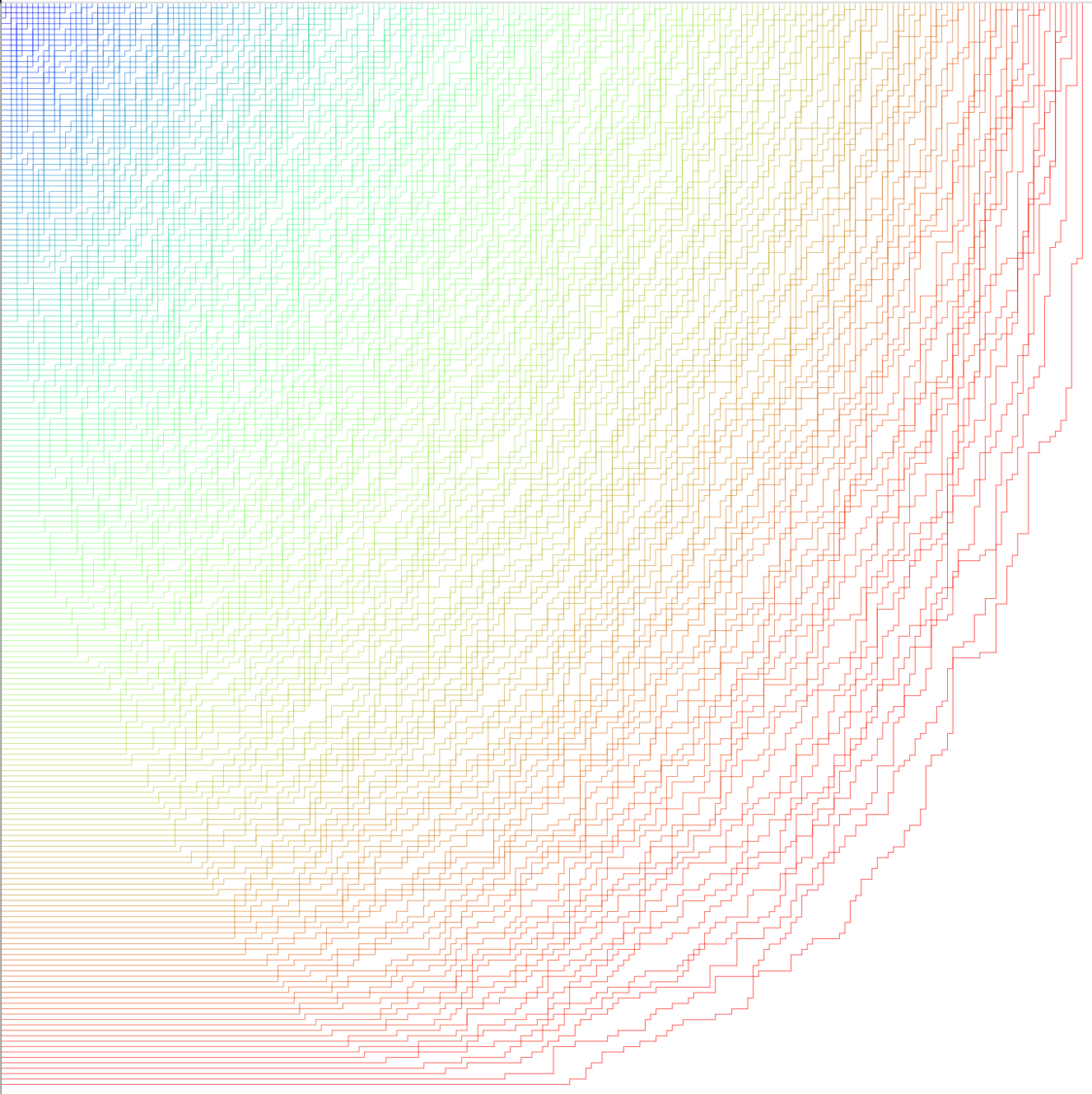} &
    \includegraphics[width=0.3\linewidth]{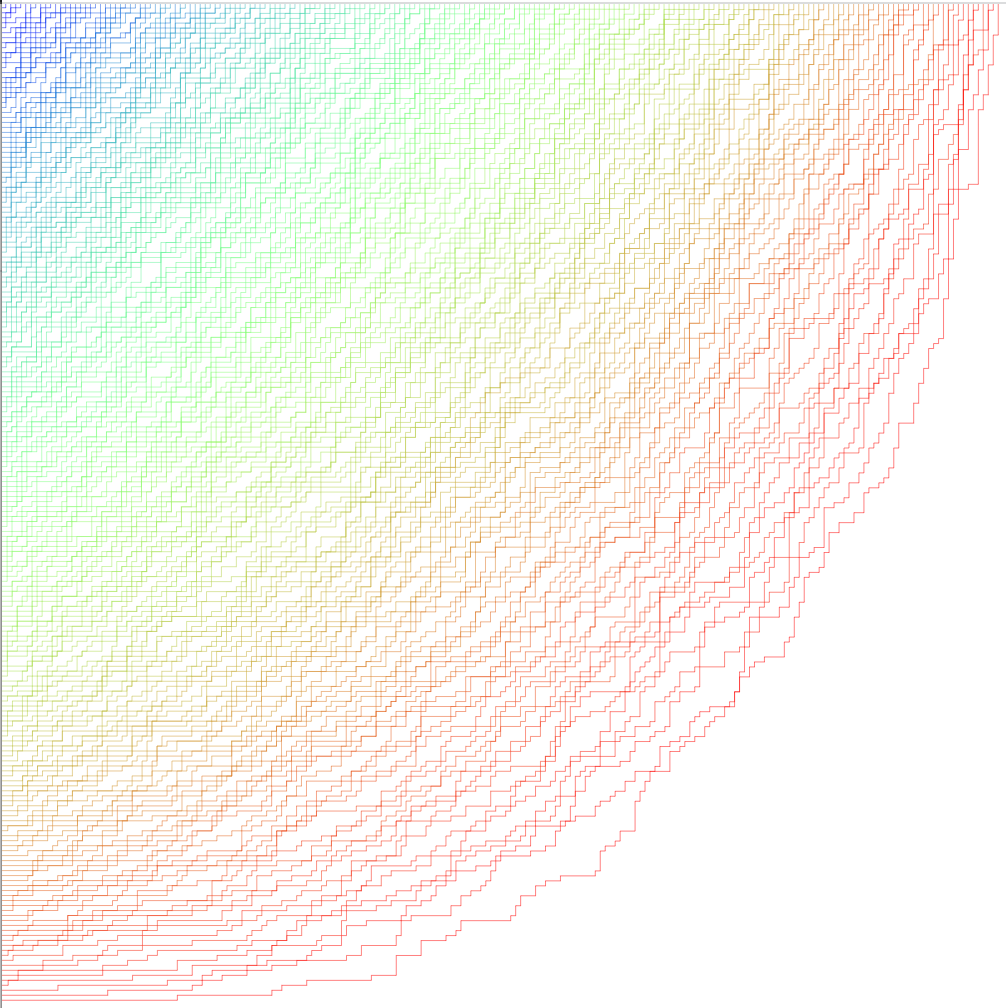} &
    \includegraphics[width=0.3\linewidth]{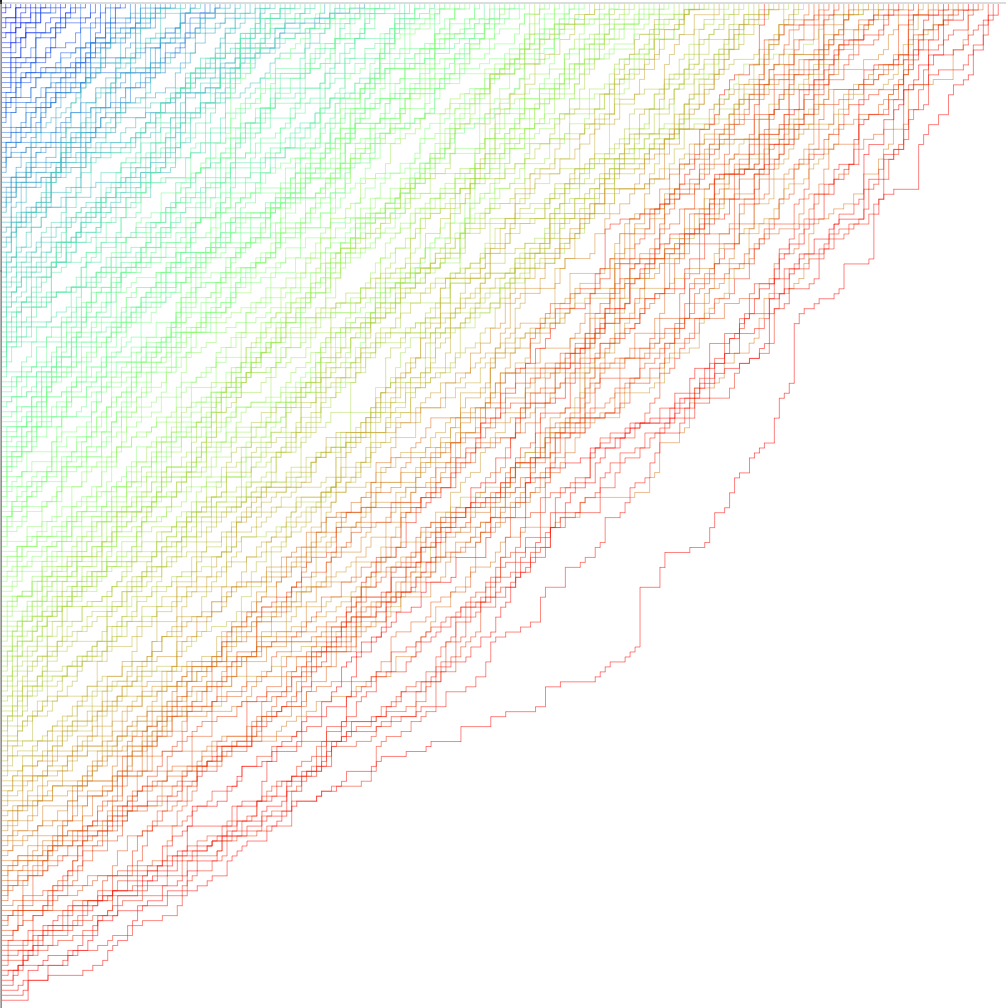} &
    \includegraphics[width=0.3\linewidth]{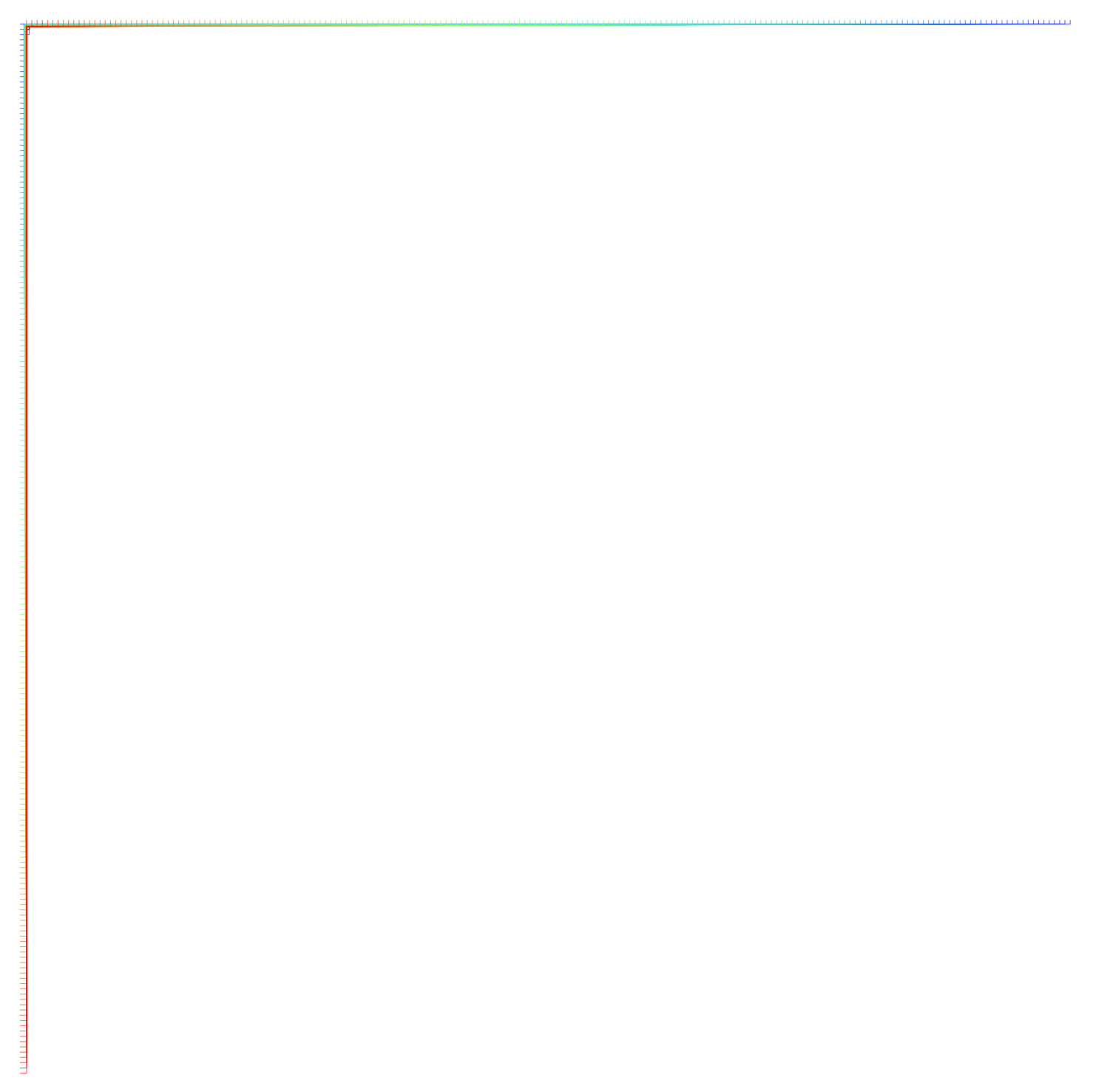}
    \\
    $t=0$ & $t=0.5$ & $t=1$ & $t=2$
    \end{tabular}
    }
    \caption{Simulations for $n=200$ with $x_1=\ldots=x_n=1$.}
    \label{fig:coloredEx}
\end{figure}

We note that for $t=0$ and $t=1$, the partition functions take a simpler form.
\begin{prop}
    For $t=0$, we have
    \[
    \tilde Z_n(x_1,\ldots,x_n|0)=\tilde Z^{\operatorname{id}}_n(x_1,\ldots,x_n|0)
    \]
    where $\operatorname{id}\in S_n$ is the identity permutation.
\end{prop}
\begin{proof}
    When $t=0$ the model no longer allows a color to exit a vertex to the east if larger colors are present in the vertex. In particular, since the paths enter on the west boundary ordered 1 to $n$ from top-to-bottom, they must exit on the north boundary ordered 1 to $n$ from left-to-right as otherwise two colors would cross with the smaller color exiting to the east.
\end{proof}

\begin{prop}\label{prop:t1prod}
    For $t=1$, we have
    \[
    \tilde Z^{\sigma}_n(x_1,\ldots,x_n|1) = \prod_{i=1}^n h_{\sigma_i-1}(x_{n-i+1}, x_{n-i+2},\ldots, x_n)
    \]
    where $h_m$ is the $m$-th complete homogeneous symmetric polynomial.
\end{prop}
\begin{proof}
Observe that when $t=1$ the different colored paths no longer interact and behave independently from one another. Let $Z_{i,j}(x_1,\ldots,x_n)$ be the partition function for a single path entering an $n\times n$ grid on the west boundary at the $i$-th row from the top and exiting on the top boundary at the $j$-th column from the left. The partition function $Z^{\sigma}_n$ factors as
\[
\tilde Z^{\sigma}_n(x_1,\ldots,x_n|1) = \prod_{i=1}^n Z_{i,\sigma_i}(x_1,\ldots,x_n).
\]
We are left to compute $Z_{i,j}$. Observe that since the path starts in the $i$-th row from the top and can only travel up and right, $Z_{i,j}$ depends only on $x_{n-i+1}, x_{n-i+2},\ldots, x_n$. Note that $h_{i-1}(x_{n-j+1}, \ldots, x_n)$ counts the number of SSYT of a single row of length $i-1$ with entries in $\{n-j+1,\ldots,n\}$. One can construct a bijection between path configurations and tableaux by letting the height at which the path exits each vertex to the east map to the filling in each cell of the tableaux. 
\end{proof}

\subsection{The colorblind model: partition functions}

\begin{figure}
\begin{center}
\includegraphics[width=14cm]{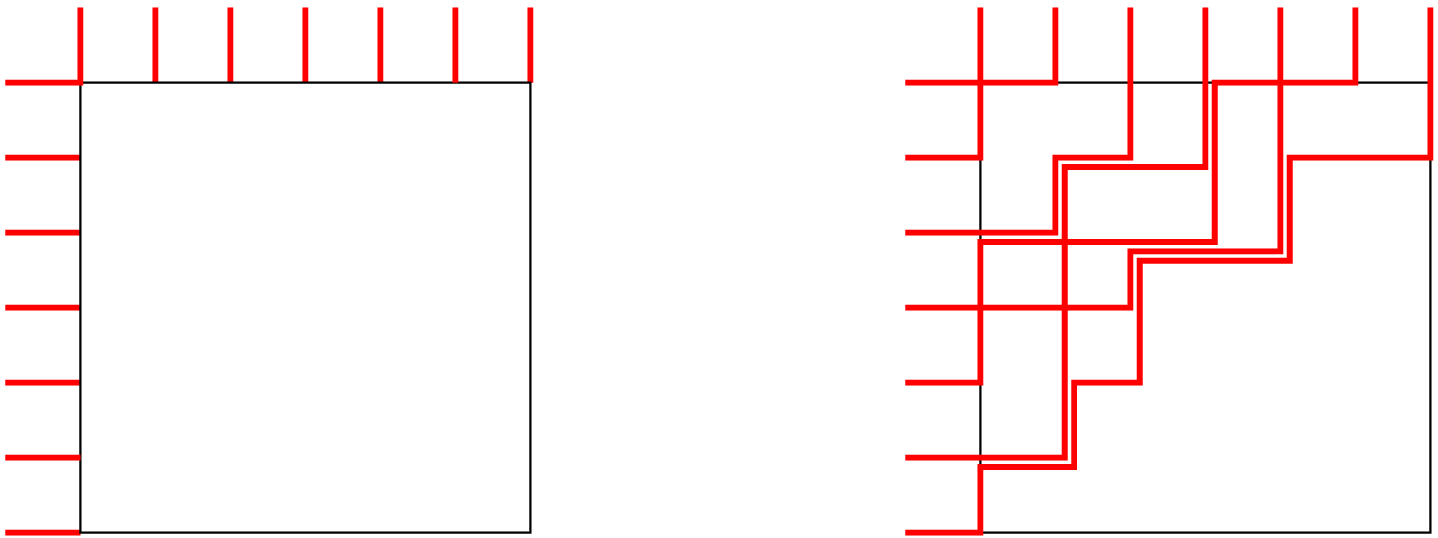}
\end{center}
\caption{\small $n\times n$ square grid with Domain Wall Boundaries (left): paths start on the W edge and end on the N edge. A sample configuration (right).}
\label{fig:colex}
\end{figure}

Now we introduce the colorblind vertex model. We consider the square lattice vertex model of non-intersecting touching paths. Each edge is occupied by a number of parallel path edges $n\in \Z_+$.
The vertex weight is
\begin{equation} \label{weights} w(n_W,n_S,n_N,n_E;x,t):= {n_N+n_E \choose n_E}_{\!\!\! t} \, t^{n_E(n_E-1)/2} \, x^{n_E} \, \delta_{n_W+n_S,n_N+n_E} 
\end{equation}
where $n_W,n_S,n_N,n_E \in \Z_+$ give the number of paths on the west, south, north, and east edge of the vertex, respectively. 

\begin{remark}
    These weights were studied in \cite{garbali2020modified}. There the authors also introduce a multicolored version different from the multicolored weights given in \eqref{eq:coloredweights} depending on an additional parameter $z$. It would be interesting to understand if the multicolored weights in \cite{garbali2020modified} could be obtained from the weights \eqref{eq:coloredweights} from a similar color merging procedure as in Lemma \ref{lem:coloredtouncolored}, although we do not pursue that here.
\end{remark}

We now consider the model on a $n\times n$ square grid with domain-wall boundary conditions, namely such that a single path originates on each W horizontal edge and terminates on each N vertical edge (in red thick lines on Figure \ref{fig:colex} left). A sample configuration for $n=7$ is represented in Figure \ref{fig:colex} (right).
The partition function reads
$$ Z_n(x_1,...,x_n\vert t):= \sum_{{\rm path}\, {\rm configs.}\atop} \prod_{i=1}^n \prod_{j=1}^n w(n_W(i,j),n_S(i,j),n_N(i,j),n_E(i,j);x_i,t) $$
where vertices are labeled by pairs $(i,j)\in [1,n]\times [1,n]$, enumerated from bottom to top and left to right.

To begin we relate this model to the colored model with weights given in \eqref{eq:coloredweights}\footnote{As far as the authors are aware, the connection between the weights in \cite{aggarwal2023colored,corteel2022vertex} and the weights in \cite{garbali2020modified} given in this lemma was not previously known.}.
\begin{lemma}\label{lem:coloredtouncolored}
    Fix $s_S, s_W\subset [n]$ with $|s_S|=n_S$, $|s_W|=n_W$, and $s_S \cap s_W =\emptyset$. We have
\[
w(n_S,n_W,n_N,n_E;x,t) = \sum_{\substack{s_N, s_E \in [n] \\ |s_N|=n_N, |s_E|=n_E}} \tilde w(s_S,s_W,s_N,s_E; x, t).
\]
\end{lemma}
\begin{proof}
    Note that in the colored vertex model only the relative order of the colors present matters to the weight, so we may assume $s_S\cup s_W=\{1,2,\ldots, n_S+n_W\}$. The path of color $i$ can exit the vertex to the north contributing a weight of $1$ to the vertex or it can exit to the east and contribute a weight of $x t^{n_S+n_W-i}$, regardless of the behavior of the other paths. We see that 
    \[
    \begin{aligned}
    \sum_{s_N, s_E \in [n]} \tilde w(s_S,s_W,s_N,s_E; x, t) &\;= \prod_{i=0}^{n_S+n_W-1}(1+xt^{i}) \\
    &\;= \sum_{i=0}^{n_S+s_W} \binom{n_S+n_W}{i}_{\!\! t} t^{i(i-1)/2}x^i
    \end{aligned}
    \]
    where the second equality follows from $q$-binomial theorem (with $q$ replaced by $t$). Restricting to the terms such that $|s_E|=n_E$ (that is, $i=n_E$ in the sum) and recalling that $n_S+n_W=n_N+n_E$ gives the lemma.
\end{proof}

Using Lemma \ref{lem:coloredtouncolored}, we immediately get the following equality between the partition function between the colored and colorblind models.
\begin{thm}\label{thm:coloredtouncolored}
\[
\tilde Z_n(x_1,\ldots,x_n|t) = Z_n(x_1,\ldots,x_n|t).
\]
\end{thm}
\begin{proof}
    Fix a path configuration in the colorblind model. Now consider all configurations of paths in the colored model such that the number of paths occupying each edge is the same as that of the colorblind configuration. Repeatedly applying Lemma \ref{lem:coloredtouncolored} shows that the weight in the two models agree. The theorem follows by summing over all the configurations.
\end{proof}
Figure \ref{fig:colortouncoloredEx} gives examples of the one-to-many weight-preserving mapping from a colorblind configuration to a collection of colored configurations used in the proof of Theorem \ref{thm:coloredtouncolored}.

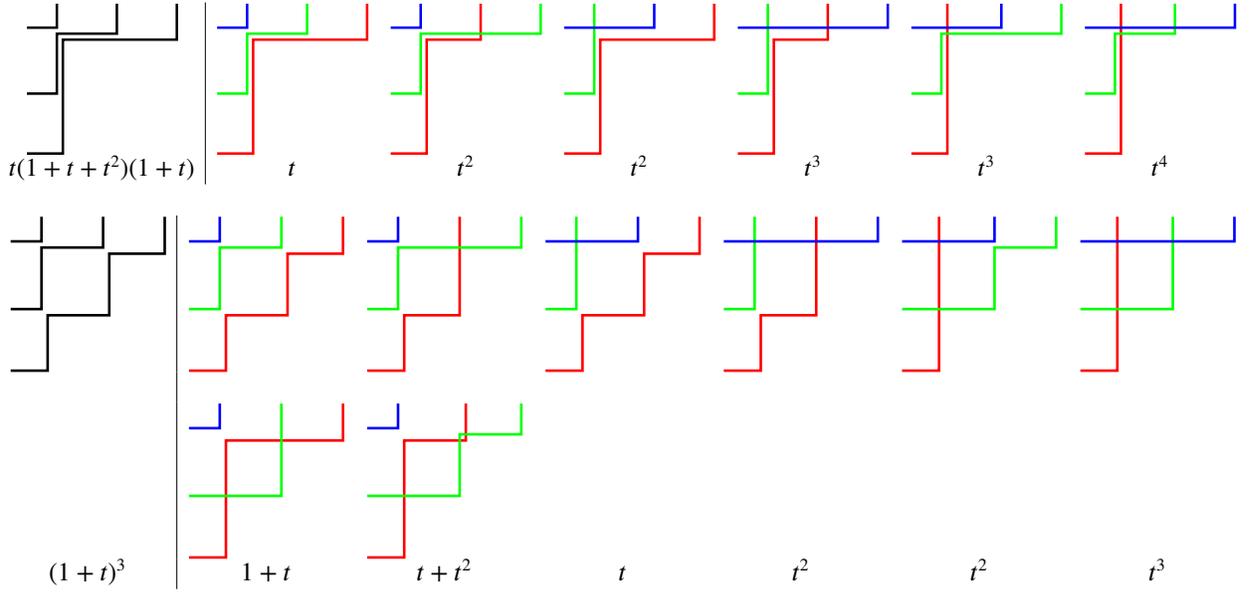
\begin{figure}
    \centering
\[
\resizebox{\textwidth}{!}{
\begin{tabular}{c|cccccc}
\begin{tikzpicture}[baseline=(current bounding box.center)]
\draw[very thick] (0,0.5)--(0.6,0.5)--(0.6,2.4)--(2.5,2.4)--(2.5,3);
\draw[very thick] (0,1.5)--(0.5,1.5)--(0.5,2.5)--(1.5,2.5)--(1.5,3);
\draw[very thick] (0,2.6)--(0.5,2.6)--(0.5,3);
\end{tikzpicture}
&
\begin{tikzpicture}[baseline=(current bounding box.center)]
\draw[very thick, red] (0,0.5)--(0.6,0.5)--(0.6,2.4)--(2.5,2.4)--(2.5,3);
\draw[very thick, green] (0,1.5)--(0.5,1.5)--(0.5,2.5)--(1.5,2.5)--(1.5,3);
\draw[very thick, blue] (0,2.6)--(0.5,2.6)--(0.5,3);
\end{tikzpicture}
&
\begin{tikzpicture}[baseline=(current bounding box.center)]
\draw[very thick, red] (0,0.5)--(0.6,0.5)--(0.6,2.4)--(1.5,2.4)--(1.5,3);
\draw[very thick, green] (0,1.5)--(0.5,1.5)--(0.5,2.5)--(2.5,2.5)--(2.5,3);
\draw[very thick, blue] (0,2.6)--(0.5,2.6)--(0.5,3);
\end{tikzpicture}
&
\begin{tikzpicture}[baseline=(current bounding box.center)]
\draw[very thick, red] (0,0.5)--(0.6,0.5)--(0.6,2.4)--(2.5,2.4)--(2.5,3);
\draw[very thick, green] (0,1.5)--(0.5,1.5)--(0.5,3);
\draw[very thick, blue] (0,2.6)--(1.5,2.6)--(1.5,3);
\end{tikzpicture}
&
\begin{tikzpicture}[baseline=(current bounding box.center)]
\draw[very thick, red] (0,0.5)--(0.6,0.5)--(0.6,2.4)--(1.5,2.4)--(1.5,3);
\draw[very thick, green] (0,1.5)--(0.5,1.5)--(0.5,3);
\draw[very thick, blue] (0,2.6)--(2.5,2.6)--(2.5,3);
\end{tikzpicture}
&
\begin{tikzpicture}[baseline=(current bounding box.center)]
\draw[very thick, red] (0,0.5)--(0.6,0.5)--(0.6,3);
\draw[very thick, green] (0,1.5)--(0.5,1.5)--(0.5,2.5)--(2.5,2.5)--(2.5,3);
\draw[very thick, blue] (0,2.6)--(1.5,2.6)--(1.5,3);
\end{tikzpicture}
&
\begin{tikzpicture}[baseline=(current bounding box.center)]
\draw[very thick, red] (0,0.5)--(0.6,0.5)--(0.6,3);
\draw[very thick, green] (0,1.5)--(0.5,1.5)--(0.5,2.5)--(1.5,2.5)--(1.5,3);
\draw[very thick, blue] (0,2.6)--(2.5,2.6)--(2.5,3);
\end{tikzpicture}
\\
$t(1+t+t^2)(1+t)$ & $t$ & $t^2$ & $t^2$ & $t^3$ & $t^3$ & $t^4$
\end{tabular}
}
\]

\[
\resizebox{\textwidth}{!}{
\begin{tabular}{c|cccccc}
\begin{tikzpicture}[baseline=(current bounding box.center)]
\draw[very thick] (0,0.5)--(0.6,0.5)--(0.6,1.4)--(1.6,1.4)--(1.6,2.4)--(2.5,2.4)--(2.5,3);
\draw[very thick] (0,1.5)--(0.5,1.5)--(0.5,2.5)--(1.5,2.5)--(1.5,3);
\draw[very thick] (0,2.6)--(0.5,2.6)--(0.5,3);
\end{tikzpicture}
&
\begin{tikzpicture}[baseline=(current bounding box.center)]
\draw[very thick, red] (0,0.5)--(0.6,0.5)--(0.6,1.4)--(1.6,1.4)--(1.6,2.4)--(2.5,2.4)--(2.5,3);
\draw[very thick, green] (0,1.5)--(0.5,1.5)--(0.5,2.5)--(1.5,2.5)--(1.5,3);
\draw[very thick, blue] (0,2.6)--(0.5,2.6)--(0.5,3);
\end{tikzpicture}
&
\begin{tikzpicture}[baseline=(current bounding box.center)]
\draw[very thick, red] (0,0.5)--(0.6,0.5)--(0.6,1.4)--(1.5,1.4)--(1.5,3);
\draw[very thick, green] (0,1.5)--(0.5,1.5)--(0.5,2.5)--(2.5,2.5)--(2.5,3);
\draw[very thick, blue] (0,2.6)--(0.5,2.6)--(0.5,3);
\end{tikzpicture}
&
\begin{tikzpicture}[baseline=(current bounding box.center)]
\draw[very thick, red] (0,0.5)--(0.6,0.5)--(0.6,1.4)--(1.6,1.4)--(1.6,2.4)--(2.5,2.4)--(2.5,3);
\draw[very thick, green] (0,1.5)--(0.5,1.5)--(0.5,3);
\draw[very thick, blue] (0,2.6)--(1.5,2.6)--(1.5,3);
\end{tikzpicture}
&
\begin{tikzpicture}[baseline=(current bounding box.center)]
\draw[very thick, red] (0,0.5)--(0.6,0.5)--(0.6,1.4)--(1.5,1.4)--(1.5,3);
\draw[very thick, green] (0,1.5)--(0.5,1.5)--(0.5,3);
\draw[very thick, blue] (0,2.6)--(2.5,2.6)--(2.5,3);
\end{tikzpicture}
&
\begin{tikzpicture}[baseline=(current bounding box.center)]
\draw[very thick, red] (0,0.5)--(0.6,0.5)--(0.6,3);
\draw[very thick, green] (0,1.5)--(1.5,1.5)--(1.5,2.5)--(2.5,2.5)--(2.5,3);
\draw[very thick, blue] (0,2.6)--(1.5,2.6)--(1.5,3);
\end{tikzpicture}
&
\begin{tikzpicture}[baseline=(current bounding box.center)]
\draw[very thick, red] (0,0.5)--(0.6,0.5)--(0.6,3);
\draw[very thick, green] (0,1.5)--(1.5,1.5)--(1.5,3);
\draw[very thick, blue] (0,2.6)--(2.5,2.6)--(2.5,3);
\end{tikzpicture}
\\
\\
& 
\begin{tikzpicture}[baseline=(current bounding box.center)]
\draw[very thick, red] (0,0.5)--(0.6,0.5)--(0.6,2.4)--(2.5,2.4)--(2.5,3);
\draw[very thick, green] (0,1.5)--(1.5,1.5)--(1.5,3);
\draw[very thick, blue] (0,2.6)--(0.5,2.6)--(0.5,3);
\end{tikzpicture}
&
\begin{tikzpicture}[baseline=(current bounding box.center)]
\draw[very thick, red] (0,0.5)--(0.6,0.5)--(0.6,2.4)--(1.6,2.4)--(1.6,3);
\draw[very thick, green] (0,1.5)--(1.5,1.5)--(1.5,2.5)--(2.5,2.5)--(2.5,3);
\draw[very thick, blue] (0,2.6)--(0.5,2.6)--(0.5,3);
\end{tikzpicture}
\\
$(1+t)^3$ & $1+t$ & $t+t^2$ & $t$ & $t^2$ & $t^2$ & $t^3$
\end{tabular}
}
\]
    \caption{Two example of the one-to-many mapping from a colorblind configuration to a collection of colored configurations such that the total weight is preserved. Here $n=3$, $x_1=x_2=x_3=1$, and we order the colors {\color{blue} blue} $<$ {\color{green} green} $<$ {\color{red} red}.}
    \label{fig:colortouncoloredEx}
\end{figure}

\begin{thm}
We have, for $t=1$,
$$ Z_n(x_1,...,x_n\vert 1)={\rm Per}_{1\leq i,j \leq n} \left( h_{j-1}(x_{n-i+1},x_{n-i+2},...,x_n) \right) ,$$
where ${\rm Per}$ stands for the permanent (i.e. determinant without signs) and $h_m(x_1,...,x_k)$ denotes the complete symmetric function with generating series
$$\prod_{i=1}^k \frac{1}{1-u x_i}= \sum_{m\geq 0} u^m h_m(x_1,...,x_k) .$$
\end{thm}
\begin{proof}
    Starting from Theorem \ref{thm:coloredtouncolored} we have
    \[
    \begin{aligned}
    Z_n(x_1,...,x_n\vert 1) &\;= \tilde Z_n(x_1,...,x_n\vert 1) \\
    &\;= \sum_{\sigma\in S_n}\tilde Z_n^{\sigma}(x_1,...,x_n\vert 1) \\
    &\;= \sum_{\sigma\in S_n} \prod_{i=1}^n h_{\sigma_i-1}(x_{n-i+1}, x_{n-i+2},\ldots, x_n) \\
    &\;= {\rm Per}_{1\leq i,j \leq n} \left( h_{j-1}(x_{n-i+1},x_{n-i+2},...,x_n) \right)
    \end{aligned}
    \]
    where the third equality follows from Proposition \ref{prop:t1prod}.
\end{proof}

\begin{figure}
\begin{center}
\includegraphics[width=14cm]{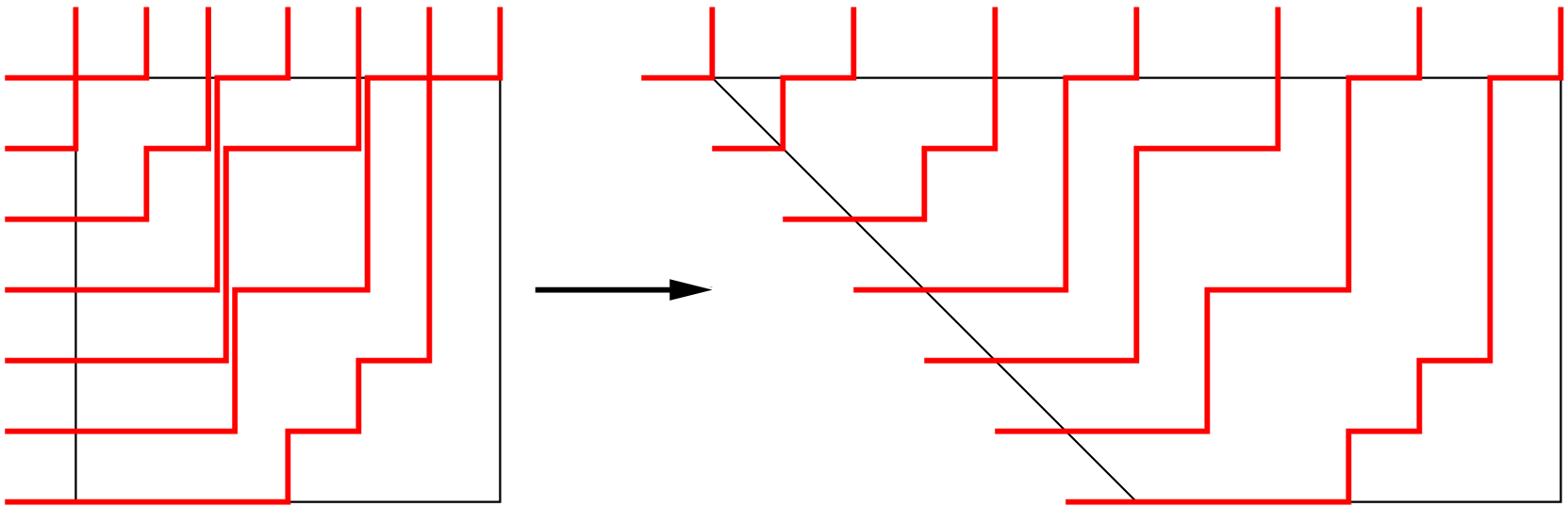}
\end{center}
\caption{\small For $t=0$, the sliding map is a bijection from the touching path configurations (with only vertical multiply occupied edges) 
to strictly non-intersecting paths.}
\label{fig:slide}
\end{figure}

\begin{thm}
We have, for $t=0$,
$$ Z_n(x_1,...,x_n\vert 0)=\prod_{1\leq i< j\leq n} (x_i+x_j) .$$
\end{thm}
\begin{proof}
Note first that as $t=0$, only the configurations with horizontal edges occupied by {\it at most} one path can contribute to the partition function, due to the factor $t^{n_E(n_E-1)/2}$. However, multiple vertical edges can be taken apart by sliding the paths horizontally to the right as follows:
the $i$-th path from above is slid by $i-1$ steps to the right (see Figure \ref{fig:slide} for an illustration).
This provides a bijection between the DWBC configurations at $t=0$ and the families of $n$ strictly non-intersecting lattice paths (NILP) starting at the points $(i,-i)$ and ending at $(2i,0)$, $i=0,1,...,n-1$. Moreover, the path weights are simply a weight $x_i$ per horizontal step at row $i$ from the bottom.
Such non-intersecting paths are easily enumerated, using the Gessel-Viennot theorem, by the following determinant. Let $z_{i,j}$ be the partition function of  single
paths starting at $(i,-i)$ and ending at $(2j,0)$. We have for $0\leq i,j \leq n-1$,
$$ z_{i,j}= \sum_{n_1,n_{2},...,n_{i} \geq 0\atop \sum_a n_a = 2j-i}  \prod_{a=i}^{n} x_{n+1-a}^{n_{n+1-a}} = 
\left.\prod_{a=n-i+1}^n \frac{1}{1-u x_a}  \right\vert_{u^{2j-i}} =h_{2j-i}(x_{n-i+1},x_{n-i+2},...,x_n).$$ 
Applying the Gessel-Viennot formula leads to the partition function
\begin{align*}
Z_n(x_1,...,x_n\vert 0)&=\det_{0\leq i,j\leq n-1}\left( h_{2j-i}(x_{n-i+1},x_{n-i+2},...,x_n)\right)=\det_{0\leq i,j\leq n-1}\left( h_{2j-i}(x_{1},x_{2},...,x_n)\right)\\
&=s_{\lambda_n}(x_1,...,x_n)=\prod_{1\leq i<j \leq n}(x_i+x_j) 
\end{align*}
where we have first used row manipulations and the identity
$$h_k(x_1,...,x_n)=\sum_{j=0}^k h_j(x_1,...,x_{i}) h_{k-j}(x_{i+1},x_{i+2},...,x_n), $$
and then
used the Jacobi-Trudi identity to identify the partition function with the $GL_n$ character (Schur function) of the staircase diagram $\lambda_n=(n,n-1,...,1)$ 
representation.
\end{proof}

\begin{remark}\label{rem:t0equiv}
Note that for $t=0$, there is a unique way to color each colorblind configuration to get a colored configuration with non-zero weight. This gives a weight-preserving bijection at the level of configurations between the colored and colorblind models. Thus it is equivalent to talk about either the colored or colorblind model when $t=0$.
\end{remark}

\subsection{Generalizations for $t=0$}\label{secgenpath}

\begin{figure}
\begin{center}
\includegraphics[width=13cm]{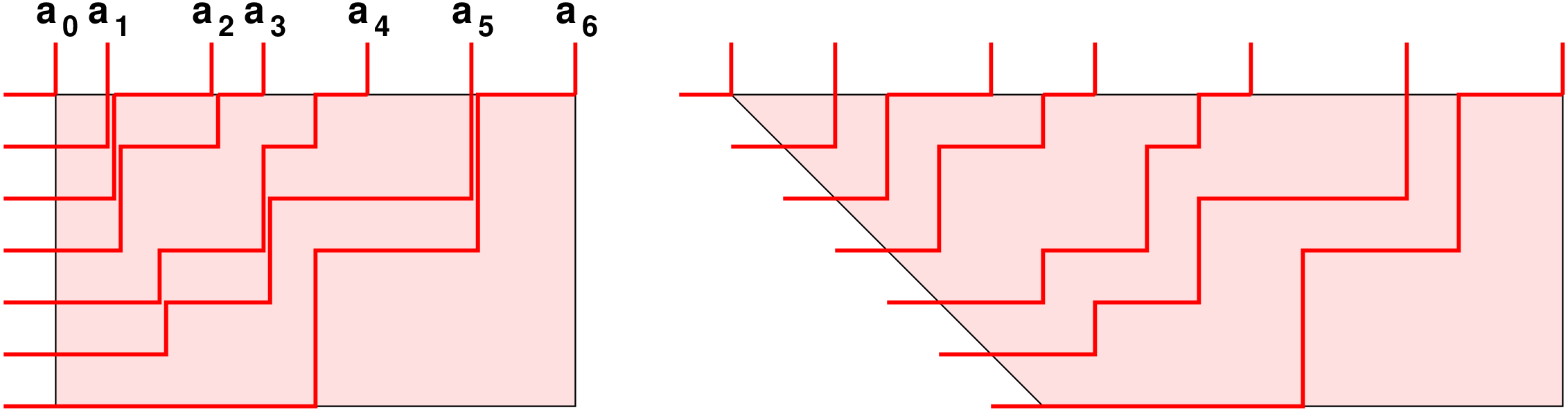}
\end{center}
\caption{\small Top: Sample configuration of the touching path model on an $(n+1)\times (a_n+1)$ square grid with Domain Wall Boundaries, with $n=6$: paths start on the W edge and end on the N edge at points $(a_i,0)$. Bottom: the NILP obtained by the sliding map bijection: the $i$-th path starts at $(i,-i)$ and ends at $(a_i+i,0)$.}
\label{fig:newcolex}
\end{figure}

Based on the above bijection with NILP using the sliding map, the partition functions for a whole family of generalizations of the 
DWBC above can be exactly computed.
Consider strictly increasing integer sequences $0=a_0<a_1<\cdots <a_{n}$. We consider the touching path model above with 
weights \eqref{weights} (and all $x_i=1$ for simplicity)
with generalized boundary conditions on a rectangular grid of size $(n+1)\times (a_{n}+1)$, where the path starting points are at 
positions $(0,-i)$, $i=0,1,...,n$
and the path endpoints are at positions $(a_i,0)$, $i=0,1,...,n$. We denote by $Z_n(a_0,a_1,...,a_{n})$ the corresponding 
partition function. Under the sliding map, the configurations are mapped to those of the NILP with starting points $(i,-i)$, $i=0,1,2,...,n$
and endpoints $a_i+i$, $i=0,1,...,n$, with partition function $Z_n^{NILP}(a_0,a_1+1,a_2+2,...,a_{n}+n)$ computed in 
\cite{DFGUI}. We deduce the following theorem.

\begin{thm}\label{pfthm}
We have
\begin{equation}
Z_n(a_0,a_1,...,a_{n})=Z_n^{NILP}(a_0,a_1+1,a_2+2,...,a_{n}+n)=\frac{\Delta(a_0,a_1+1,a_2+2,...,a_n+n)}{\Delta(0,1,2,...,n)},
\end{equation}
in terms of the Vandermonde determinant $\Delta(x_0,x_1,...,x_n)=\prod_{0\leq i<j\leq n} (x_j-x_i)$.
\end{thm}

\section{Limit shape: the case $t=0$}
In this section, we consider the model of  $t=0$ touching paths (touching paths for short) with generalized boundary conditions as in Section \ref{secgenpath} above, and compute the limit shape for large $n$ by use of the tangent method.

\subsection{General solution I: SE and SW portions of the arctic curve}

\begin{figure}
\begin{center}
\includegraphics[width=13cm]{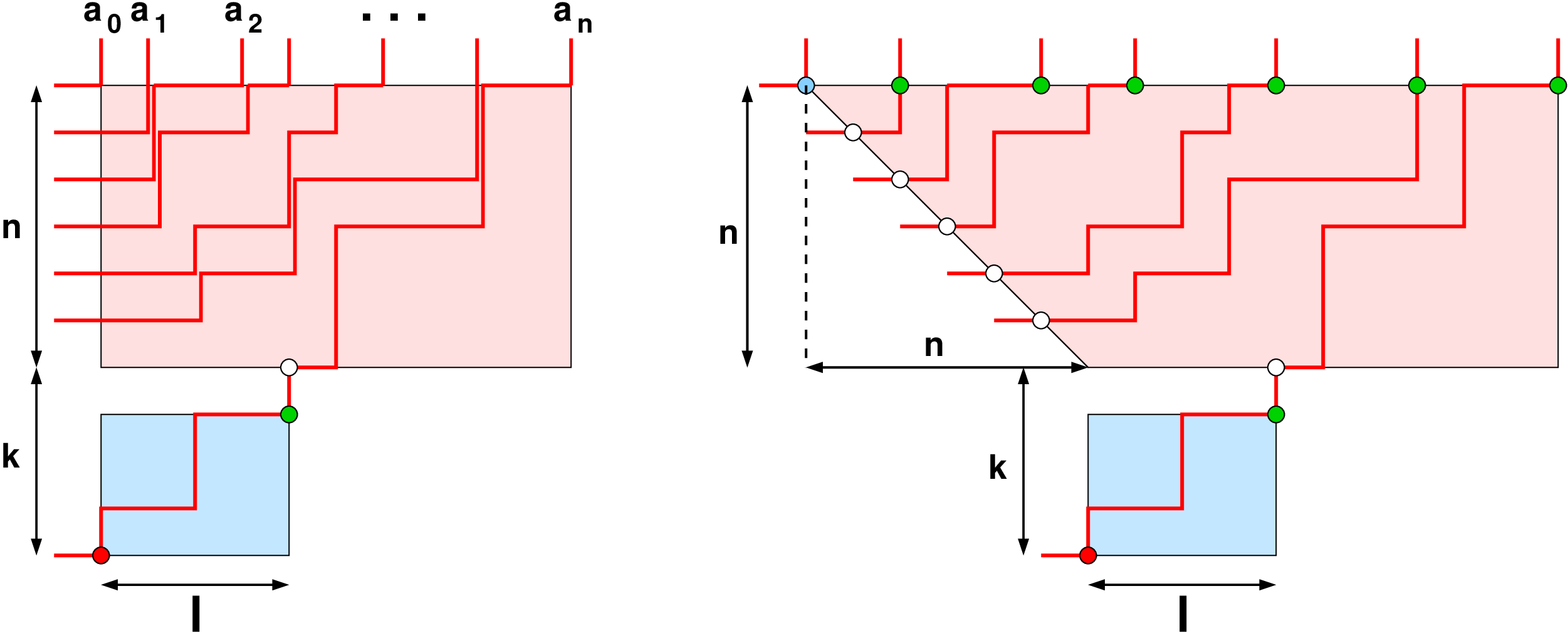}
\end{center}
\caption{\small  The tangent method for touching paths (left) and the bijective image for NILP by the sliding map (right).}
\label{fig:tangent}
\end{figure}

Recall that the tangent method, illustrated in Figure \ref{fig:tangent}, consists of moving away the starting point of the $n$-th path (say by a vertical translation by $(0,-k)$), so that the new total partition function $Z'_{n,k}(a_0,...,a_n)$ is decomposed into the sum
\begin{equation}\label{tgt}Z'_{n,k}(a_0,...,a_n)=\sum_{\ell} Y_{k,\ell} \, H_{n,\ell}(a_0,...,a_n) ,\end{equation}
where $Y_{k,\ell}$ is the partition function of a single path starting at point $(0,-n-k)$ and first touching the line $y=-n$ at the point $(\ell,-n)$ (blue rectangle in Figure \ref{fig:tangent}, left), and $H_{n,\ell}(a_0,...,a_n)$ is the partition function with starting points $(0,-i)$, $i=0,1,...,n-1$ and $(\ell,-n)$ and endpoints at $(a_i,0)$, $i=0,1,...,n$ 
(pink rectangle in Figure \ref{fig:tangent}, left). 

We now consider the large $n$ limit in which $k=n z$, $\ell = n w$, and the $a_i$ are distributed according to a function $\al(x)$, $x\in [0,1]$, with
$a_i\sim n \al(i/n)$.

From the sliding map bijection to NILP, equation \eqref{tgt} can be rephrased in NILP terms, with
$Y_{k,\ell}=Y_{k,\ell}^{NILP}$ and $H_{n,\ell}(a_0,...,a_n)=H_{n,\ell}^{NILP}(b_0,b_1,...,b_n)$, with $b_i=a_i+i$. The NILP versions correspond to the partition functions, respectively, in the blue and pink regions on the right of Figure \ref{fig:tangent}. As a result, the tangent method yields the exact same result for the SE portion of arctic curve of both models. More precisely, as the bottom path is translated by $(n,0)$, the SE portion of the arctic curve (limit rescaled by a factor $1/n$) for NILP is the $(x,y)\to (x+1,y)$ translate of that for the touching paths. Note also that the distribution of endpoints of the NILP is 
$\beta(x)=\al(x)+x$, for $x\in [0,1]$. 

\begin{figure}
\begin{center}
\includegraphics[width=15cm]{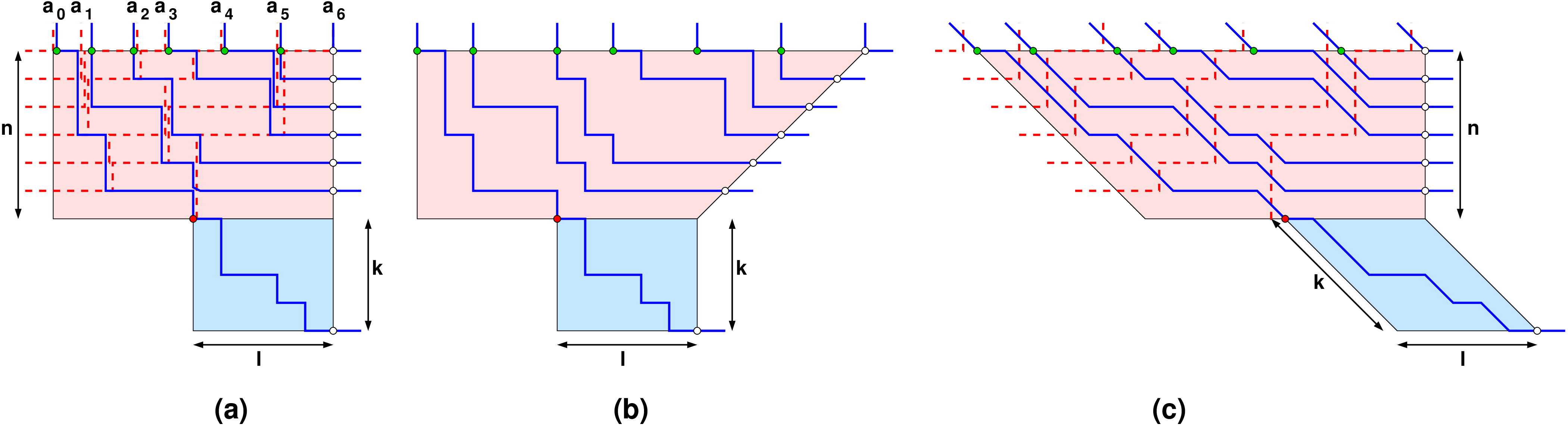}
\end{center}
\caption{\small  The tangent method for the SW branch of the arctic curve. (a) $t=0$ new touching paths (in solid blue line) dual to the original ones (dashed red lines) and (b) their bijective NILP image under the sliding map, reexpressed in (c) as the duals of the original NILP.}
\label{fig:othernewcolex}
\end{figure}

The SW portion of arctic curve is obtained analogously by considering an equivalent description of the configurations within the pink rectangle of Figure \ref{fig:tangent} (left),
using a different (dual) set of paths. Indeed, as illustrated in Figure \ref{fig:othernewcolex} (a), starting from a configuration of touching paths, let us draw new horizontal path steps on all the un-occupied horizontal edges, and erase all the former horizontal steps, and then reconnect
vertical steps to those while keeping the touching yet non-intersecting condition: we obtain a new family of touching paths going left and up (the blue paths in Figure \ref{fig:othernewcolex} (a)), and starting from the E border of the rectangle at points $(a_n,-i)$, $i=0, 1,...,n-1$, and at the same point $(\ell,-n)$ on the S boundary. The blue rectangle is now associated to a single path starting at
$(a_n,-n-k)$ and touching the line $y=-n$ for the first time at the point $(\ell,-n)$. Via the sliding map, the configuration is equivalent to that of NILP displayed in Figure \ref{fig:othernewcolex} (b), itself equivalent to the NILP represented in  Figure \ref{fig:othernewcolex} (c), in clear bijection with the original NILP (in red dashed lines).
Note, however, that the NILP boundary conditions of Figure \ref{fig:othernewcolex} (b) are the same as the original one, up to a vertical reflection, under which the endpoints become
$c_i= b_n-b_{n-i}=a_n-a_{n-i}+i$.
We deduce that the pink partition function in Figure \ref{fig:othernewcolex} (b) is nothing but $H_{n,\ell}^{NILP}(c_0,c_1,..,c_n)=H_{n,\ell}(d_0,d_1,...,d_n)$, with
$d_i=c_i-i=a_n-a_{n-i}$, while the blue partition function remains $Y_{k,\ell}^{NILP}$. As a consequence, the SW portion of arctic curve for the $t=0$ (dual) touching paths is identical to the vertical reflection of the SE portion for the original touching paths,
up to a reflection of the endpoints, i.e. a change of distribution $\al\to \delta$, where
$$ \delta(x)=\al(1)-\al(1-x) .$$

Both SE and SW portions of arctic curve were derived in \cite{DFGUI} for arbitrary endpoint distribution $\beta$.
\begin{thm}[\cite{DFGUI}]\label{NILPthm}
Let 
$$x(t)=e^{-\int_0^1 \frac{du}{t-\beta(u)} }, $$
then the arctic curve reads parametrically:
\begin{eqnarray*}X^{NILP}(t)&=t -\frac{x(t)}{x'(t)}\Big(1-x(t)\Big)\\
Y^{NILP}(t)&=-\frac{1}{x'(t)} \Big(1-x(t)\Big)^2
\end{eqnarray*}
where   $t\in [\beta(1),\infty)$ for the SE branch and $t\in (-\infty,0]$ for the SW branch.
\end{thm}
Note that, as expected for free fermion models, the curve is analytic, as the two portions are continuations of one-another.
This allowed us to conjecture that the arctic curve equation of Theorem \ref{NILPthm} is still valid over the whole range $t\in \R$.

This immediately translates to touching paths with arbitrary endpoint distribution $\alpha: [0,1]\to [0,\al(1)]$.
\begin{thm}\label{mainthm}
The arctic curve for touching paths reads parametrically as follows. Let 
$$x_0(t)=e^{-\int_0^1 \frac{du}{t-u-\al(u)} },\quad x_1(t)= e^{-\int_0^1 \frac{du}{t-u-\delta(u)} }=\frac{1}{x_0(\al(1)+1-t)},$$
with $\delta(u)=\al(1)-\al(1-u)$ as above.
For $t\in[\al(1)+1,\infty)$, the SE portion is parametrized by
\begin{align}X_{SE}(t) &=t-1 -\frac{x_0(t)}{x_0'(t)}\Big(1-x_0(t)\Big)\nonumber \\
Y_{SE}(t) &=-\frac{1}{x_0'(t)} \Big(1-x_0(t)\Big)^2.\label{SEtouch}
\end{align}
For $t\in[\al(1)+1,\infty)$, the SW portion is parametrized by
\begin{align}X_{SW}(t) &=\al(1)+1-t+\frac{x_1(t)}{x_1'(t)}\Big(1-x_1(t)\Big)\nonumber \\
Y_{SW}(t) &=-\frac{1}{x_1'(t)} \Big(1-x_1(t)\Big)^2 .\label{SWtouch}
\end{align}
\end{thm}

\begin{remark}\label{remshear}
Using the relation $x_1(t)=x_0(\al(1)+1-t)^{-1}$, we may rewrite, using the variable $\tilde t=\al(1)+1-t$:
\begin{align*}
X_{SW}(t)&=\tilde t-\frac{1}{x_0'(\tilde t)}\Big(1-x_0(\tilde t)\Big)=X_{SE}(\tilde t)+Y_{SE}(\tilde t) \\
Y_{SW}(t)&= -\frac{1}{x_0'(\tilde t)} \Big(1-x_0(\tilde t)\Big)^2 =Y_{SE}(\tilde t)
\end{align*}
with $\tilde t\in (-\infty,0]$. Comparing with the result of Theorem \ref{NILPthm}, we see that this is the image of the 
$SE$ branch of the arctic curve of the vertical reflection of the NILP model (i.e. the reflection of its $SW$ branch), under the shear transformation $(x,y)\to (x+y,y)$. 
A simple explanation for this fact comes from the tangent method: the leftmost (blue) path in Figure \ref{fig:othernewcolex} (c) is indeed mapped to the leftmost (blue) path in 
Figure \ref{fig:othernewcolex} (b) under the same shear transformation. This is an example of the ``shear phenomenon" for certain arctic curves, that a portion of the curve is the shear 
transform of the analytic continuation of another.
\end{remark}

\begin{remark}
The arctic curve is non-analytic at the point $\lim_{t\to\infty} (X_{SE}(t),Y_{SE}(t))=(\int_0^1\al(u)du,-1)$. The result of Theorem \ref{mainthm}
therefore only applies to the SE and SW branches of arctic curve. In the case of macroscopic gaps in the distribution, more work is required, but the same
mix of shear phenomenon and plain translation applies locally (see Section \ref{gapsec} below).
\end{remark}

We now show some examples of direct applications of Theorem \ref{mainthm}.

\begin{figure}
\begin{center}
\includegraphics[width=8cm]{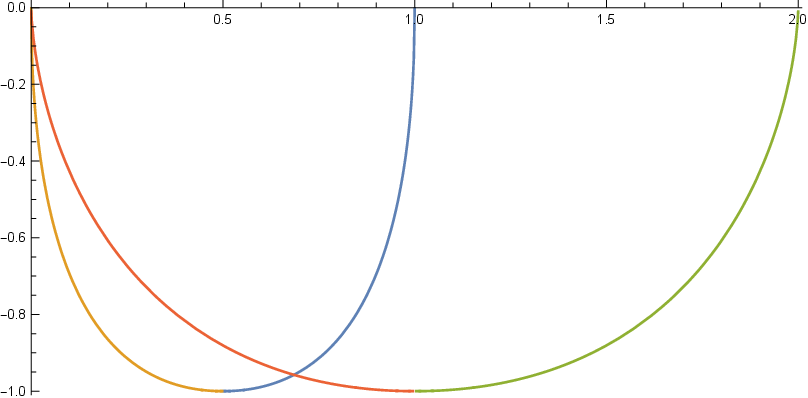}
\end{center}
\caption{\small  Arctic curves for the cases $p=1$ and $p=2$ of  Example \ref{twopex}.}
\label{fig:twop}
\end{figure}

\begin{example}[Uniform spacing]\label{twopex}
The case of regularly spaced points $a_i= p i$, for $p\in \Z_{>0}$
corresponds to a distribution  $\al(x)=p x=\delta(x)$, leading to
$$x_0(t)=x_1(t)= \left( 1-\frac{p+1}{t}\right)^{\frac{1}{p+1}}. $$
As a consequence, the arctic curve is symmetric w.r.t. the vertical line $x=\frac{p}{2}$.
The arctic curve reads, after a change of variables $t=(p+1)/(1-u^{p+1})$ for $u\in [0,1)$:
\begin{align*}
X_{SW}(u)&=\frac{p+1}{1-u^{p+1}}\left( 1- \frac{(p+1)(1-u)u^{p+1}}{1-u^{p+1}}\right)-1= p-X_{SE}(u) \\
Y_{SW}(u)&=-\left(\frac{(p+1)(1-u)}{1-u^{p+1}}\right)^2 u^p=Y_{SE}(u)
\end{align*}
In particular, for $p=1$, the arctic curve is made of two pieces of parabolas:
\begin{align*} 
(P1):& (y+2x)^2-4(y+1)=0 \qquad ( x\in [0,\frac12],\ y\leq 0)\\
(P2):& (y-2x+2)^2-4(y+1)=0\qquad ( x\in [\frac12,1],\ y\leq 0)
\end{align*}
while for $p=2$, it is made of two pieces of quartics:
\begin{align*} 
(Q1):& {\scriptstyle (3 x^2-3 x y+y^2-9x+2y)^2-4y(3 x^2-3 x y+y^2-9x+2y)+4y^2(y+1)=0}\qquad ( x\in [0,1],\ y\leq 0)\\
(Q2):& {\scriptstyle (3 x^2+3 x y+y^2-3x-4y-6)^2)^2-4y(3 x^2+3 x y+y^2-3x-4y-6)+4y^2(y+1)=0} \qquad ( x\in [1,2],\ y\leq 0)
\end{align*}
These two cases are depicted in Figure \ref{fig:twop}.
\end{example}

\subsection{General solution II: the case of a macroscopic gap}\label{gapsec}

\begin{figure}
\begin{center}
\includegraphics[width=10.cm]{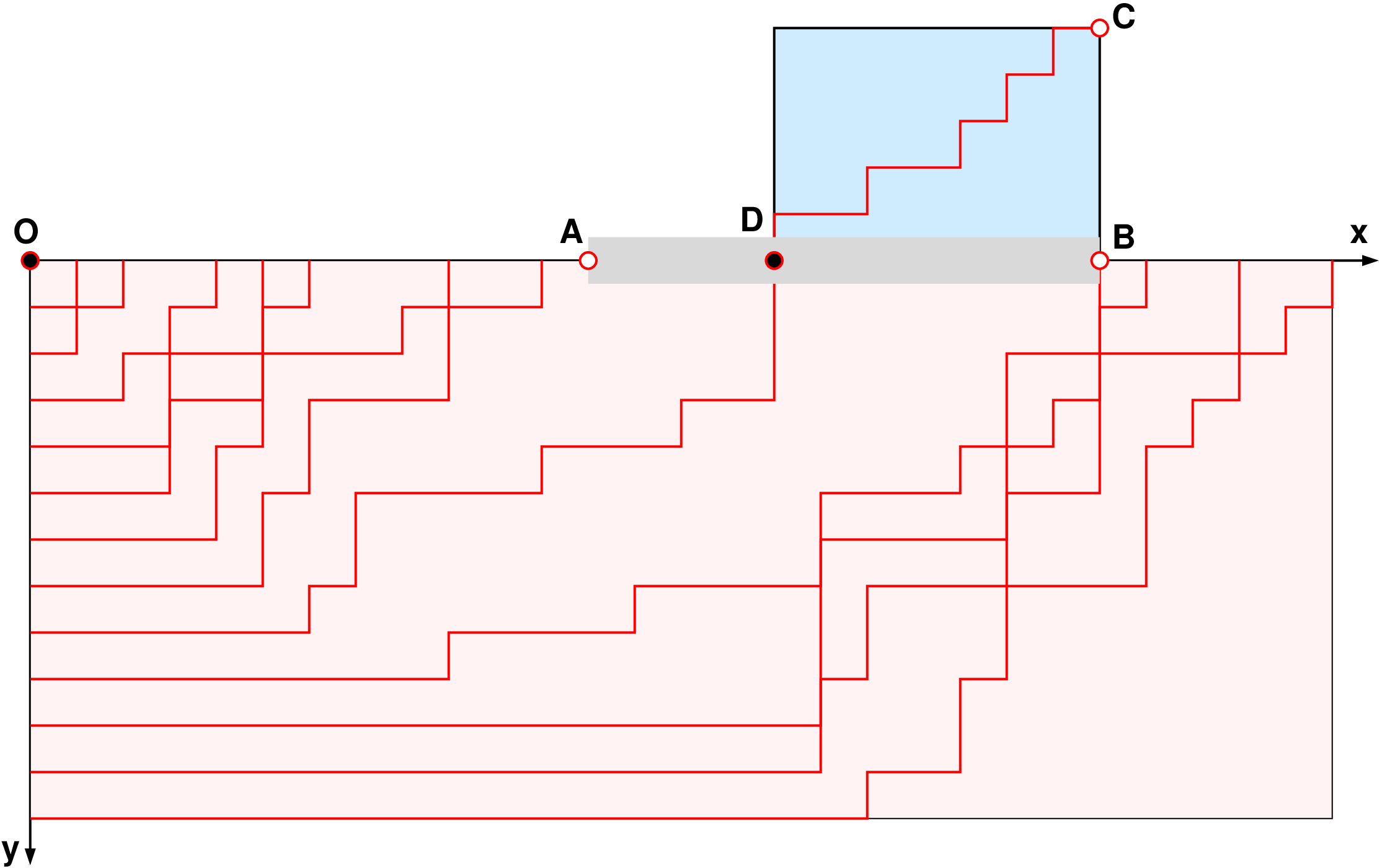}
\end{center}
\caption{\small  Tangent method for the case of a gapped  distribution of (red) touching path endpoints. 
The gap $[A,B]$ is shaded in gray. The left endpoint of the gap $A$ is moved to a position $C$, forcing an exit point $D$ from the original domain (in pink). The blue domain corresponds to the partition function of a single free path from $D$ to $C$.}
\label{fig:onegap}
\end{figure}

We now consider the general case of a gap in the distribution of endpoints as illustrated in Figure \ref{fig:onegap}, i.e. an absence of endpoints say between points  $a_k$ and $a_{k+1}$, with $a_{k+1}-a_k=m$. In the large $n$ scaling 
limit where $\kappa=k/n$, $\mu=m/n$, we have a distribution of endpoints
\begin{equation}\al(u)=\left\{ \begin{matrix}
\al_1(u) & u\in [0,\kappa] \\
\al_2(u) & u \in [\kappa,1]
\end{matrix} \right. , \ \ \al_2(\kappa)=\al_1(\kappa)+\mu.\label{distrigap}
\end{equation}

Let us concentrate on the portion of arctic curve created by the gap. We apply the tangent method as in Figure \ref{fig:onegap}, by moving the path endpoint from position $A=(a_k,0)$ to $C=(a_k+m,s)$
for some $s>0$. The path exits the original setting at some point $D=(a_k+r,0)$, with $r\in [0, m]$.
The new partition function (in pink in Figure \ref{fig:onegap}) divided by the one without the move reads, using Theorem \ref{pfthm},
\begin{align*}
H_{n,r}=&\;\frac{Z_{n,r}}{Z_{n,0}} \\
=&\; \exp\left\{- \sum_{i=0}^{k-1}{\rm Log}\Big(\frac{a_i+i -r-a_k-k}{a_i+i-a_k-k}\Big)-\sum_{i=k+1}^n {\rm Log}\Big(\frac{a_i+i -r-a_k-k}{a_i+i-a_k-k}\Big) \right\} \\
\sim&\; e^{n S_0(\rho)}
\end{align*}
with
\begin{align*}
S_0(\rho)&\;= \int_0^\kappa du\,  {\rm Log}\Big(\frac{\al_1(u)+u-\rho-\al_1(\kappa)-\kappa}{\al_1(u)+u-\al_1(\kappa)-\kappa}\Big) 
+\int_{\kappa}^1 du\,  {\rm Log}\Big(\frac{\al_2(u)+u-\rho-\al_1(\kappa)-\kappa}{\al_2(u)+u-\al_1(\kappa)-\kappa}\Big)  .
\end{align*}
After exiting at position $D=(a_k+r,0)$ the path continues until it reaches the endpoint $C=(a_k+m,s)$. With the scaling $r=\rho n,k=\kappa n,m=\mu n$ and $s=z n$,
the corresponding single path partition function reads
\begin{align*}
Y_{r,s}&= {m-r+s\choose s} \sim e^{n S_1(\rho)}
\end{align*}
with
\begin{align*}
S_1(\rho)&=(\mu-\rho+z)\,{\rm Log}(\mu-\rho+z)-(\mu-\rho)\,{\rm Log}(\mu-\rho)-z\,{\rm Log}(z). 
\end{align*}
The leading contribution to the total partition function $\sum_r \frac{Z_{n,r}}{Z_{n,0}}Y_{r,s}$ is attained at the saddle-point of the total action $S(\rho)=S_0(\rho)+S_1(\rho)$,
namely the most likely value of $\rho$ is determined by $\partial_\rho S(\rho)=0$:
$$\int_0^\kappa  \frac{du}{\rho+\al_1(\kappa)+\kappa-\al_1(u)-u} 
+\int_{\kappa}^1 \frac{du}{\rho+\al_1(\kappa)+\kappa-\al_2(u)-u}+{\rm Log}\Big( \frac{\mu-\rho}{\mu-\rho+z}\Big) =0 .$$
Introducing the notation
$$ t= \kappa+\rho+\al_1(\kappa),\ \ x_0(t)=e^{-\int_0^1 \frac{du}{t-u-\al(u)}}=e^{-\int_0^\kappa\frac{du}{t-u-\al_1(u)}-\int_\kappa^1\frac{du}{t-u-\al_2(u)}} =\frac{\mu-\rho}{\mu-\rho+z}, $$
the tangent to the arctic curve is the line through $(\al_1(\kappa)+\rho,0)=(t-\kappa,0)$ and $(\al_1(\kappa)+\mu,z)$, with equation
$\frac{\mu-\rho}{z}Y- X+t-\kappa=0$. Noting that $\frac{\mu-\rho}{z}=x_0(t)/(1-x_0(t))$, the envelope of this family of lines
is the parametric curve given by
\begin{equation}
\frac{x_0(t)}{1-x_0(t)} Y -X +t-\kappa =0,\ \ \frac{x_0'(t)}{(1-x_0(t))^2} Y+1=0 \ \ \Rightarrow\ \ \left\{ \begin{matrix} X_{SW}^{\rm gap}&=t-\kappa -\frac{x_0(t)(1-x_0(t))}{x_0'(t)} \\
Y_{SW}^{\rm gap}&=-\frac{(1-x_0(t))^2}{x_0'(t)}\hfill \end{matrix}\right. . \label{SWgapeq}
\end{equation}
We note that this is a simple translation by $(-\kappa,0)$ of the SW gap portion of the arctic curve for the corresponding NILP.  The range of $t$ is determined by imposing $z>0$, i.e. $x_0(t)<1$.
This gives the segment $t\in [\kappa+\al_1(\kappa),t^*]$, where $t^*<\kappa+\mu+\al_1(\kappa)=\kappa+\al_2(\kappa)$ is the unique point
$t\in [\kappa+\al_1(\kappa),\kappa+\al_1(\kappa)+\mu]$ such that $x_0(t)=1$.
The existence and uniqueness of such a point is easily seen by noting that $x_0(t)=1$ iff
$$F(t)= \int_0^\kappa\frac{du}{t-u-\al_1(u)}+\int_\kappa^1\frac{du}{t-u-\al_2(u)} =0,$$
and $F(t)$ is strictly decreasing from $+\infty$ to $-\infty$ on the segment $[\kappa+\al_1(\kappa),\kappa+\al_1(\kappa)+\mu]$.

\begin{figure}
\begin{center}
\includegraphics[width=10.cm]{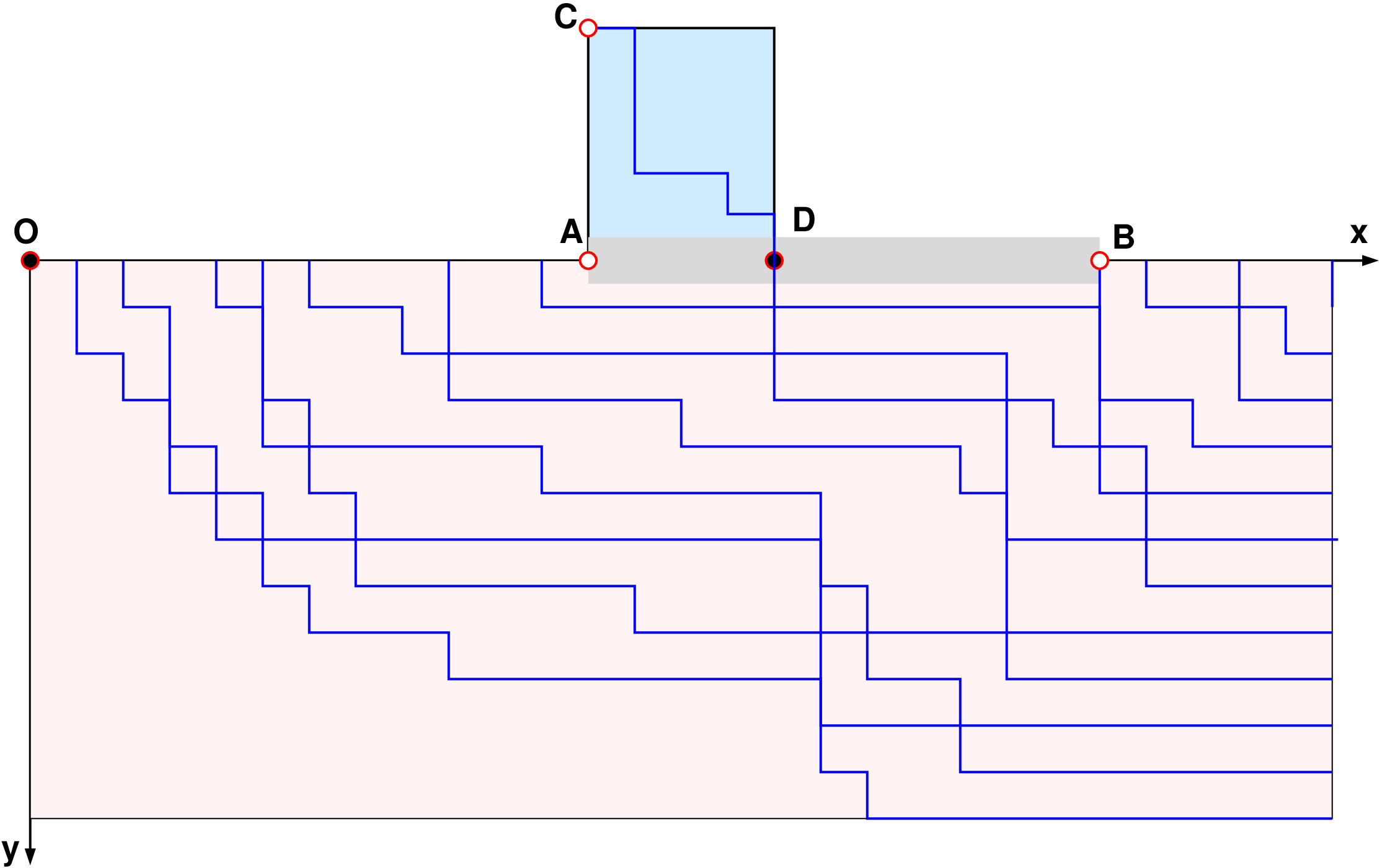}
\end{center}
\caption{\small  Tangent method for the case of a gapped  distribution of (blue) dual touching path endpoints. 
The left endpoint of the gap $A$ is moved to a position $C$, forcing an exit point $D$ from the original domain (in pink). The blue domain corresponds to the partition function of a single free dual path from $D$ to $C$. The reflection of this picture maps to one similar to Figure \ref{fig:onegap}, upon 
reflection of endpoint positions (i.e. $\al(u)\to \delta(u)$) and of $A,B,C,D$.}
\label{fig:onegapdual}
\end{figure}

Similarly, we may obtain the $SE$ gap portion of the arctic curve by considering, as before, the dual paths (with same vertical steps, but flipped horizontal ones) and then reflecting the picture w.r.t. a vertical line, to apply again the tangent method. 
As mentioned before, the reflection of endpoint positions changes $x_0(t)\to x_1(t)= 1/x_0(1+\al(1)-t)$, with $\al(1)=\al_2(1)$, as well as $\kappa\to 1-\kappa$ and results in 
\begin{equation} \left\{ \begin{matrix} X_{SE}^{\rm gap}&=\al(1)-t+1-\kappa +\frac{x_1(t)(1-x_1(t))}{x_1'(t)} \\
Y_{SE}^{\rm gap}&=-\frac{(1-x_1(t))^2}{x_1'(t)} \hfill\end{matrix}\right.\qquad (t\in[t^*,\kappa+\mu+\al_1(\kappa)]).
\label{SEgapeq}
\end{equation}
Note that in both cases $t$ is the intercept between the tangent line and the gap segment.

\begin{remark}
Comparing to the corresponding NILP arctic curve, for which we have the following gap portions:
\begin{align*}
&\left\{ \begin{matrix} X_{SW}^{\rm NILP,gap}(t)&=t -\frac{x_0(t)(1-x_0(t))}{x_0'(t)} \\
Y_{SW}^{\rm NILP,gap}(t)&=-\frac{(1-x_0(t))^2}{x_0'(t)}\hfill \end{matrix}\right. \  (t\in [\kappa+\al_1(\kappa),t^*]),\\
&\left\{ \begin{matrix} X_{SE}^{\rm NILP,gap}(t)&=t -\frac{x_0(t)(1-x_0(t))}{x_0'(t)} \\
Y_{SE}^{\rm NILP,gap}(t)&=-\frac{(1-x_0(t))^2}{x_0'(t)}\hfill \end{matrix}\right. \  (t \in [t^*,\kappa+\al_1(\kappa)+\mu]).
\end{align*}
We see that the touching path $SW$ gap branch is just the translation of the NILP $SW$ gap branch by $(-\kappa,0)$, while the 
touching path $SE$ gap branch is a shear transform of that of the NILP, also translated by the same amount:
\begin{align*}
(X_{SW}^{\rm gap}(t),Y_{SW}^{\rm gap}(t))&=(X_{SW}^{\rm NILP,gap}(t)-\kappa,Y_{SW}^{\rm gap}(t)),\\
 (X_{SE}^{\rm gap}(t),Y_{SE}^{\rm gap}(t))&=
(X_{SE}^{\rm NILP,gap}(\tilde t)+Y_{SE}^{\rm NILP,gap}(\tilde t)-\kappa,Y_{SE}^{\rm NILP,gap}(\tilde t))
\end{align*}
where we have used again the relation $x_1(t)=1/x_0(\tilde t)$, $\tilde t =\al(1)+1-t$.
The reason for these two different connections is rooted in the two corresponding different applications of the tangent method, see also Remark \ref{remshear}.

The first (SE gap branch) uses the original paths: moving the end of the rightmost path at the left of the gap can be done in both touching path and NILP situations, and the setting for the tangent method is essentially the same, except for a global sliding of the gap to the right by the quantity $k=\kappa n$, at the same time as the new path endpoint and exit point. To produce the SE gap curve for the touching paths, the NILP SE gap curve must therefore be translated back by an opposite quantity, which explains the above translation by $(-\kappa,0)$.

The second (SW gap branch) uses the dual paths (in which horizontal steps are flipped, but vertical steps preserved and recombined) in both touching and NILP situations. As seen above, the tangent method applied to these paths is essentially the same in both touching and NILP situations, however
the leftmost path (exiting inside the gap) is now mapped from the NILP to the touching case via the shear transform $(x,y)\mapsto (x+y,y)$, composed with  the global translation of the gap by $(-\kappa,0)$.
\end{remark}

\begin{remark}\label{rmk:infgap}
The gap solution of this section allows to recover the result of Theorem \ref{mainthm} above as a particular case. 
Indeed, we may consider first the case when $\kappa=0$, i.e. we start with a gap of size $\mu$, followed by a
distribution $\al_2$ of endpoints. Applying the tangent method with the dual paths gives access to the SW branch of arctic
curve, which can be viewed as the gap's SE branch. More precisely, eq. \eqref{SEgapeq} reduces to \eqref{SWtouch}
for $\kappa=0$ and $\al_2=\al$. The ranges of parameter $t$ match if we send the gap size $\mu\to \infty$.

Next, we consider the case $\kappa=1$, i.e. we have a distribution $\al_1$ of endpoints, followed by a gap of size $\mu$.
Applying the tangent method with the usual paths gives access to the SE branch of arctic
curve, which can be viewed as the gap's SW branch. More precisely, eq. \eqref{SWgapeq} reduces to \eqref{SEtouch}
for $\kappa=1$ and $\al_1=\al$. As before, the ranges of parameter $t$ match if we send the gap size $\mu\to \infty$.
\end{remark}

\begin{figure}
\begin{center}
\includegraphics[width=11cm]{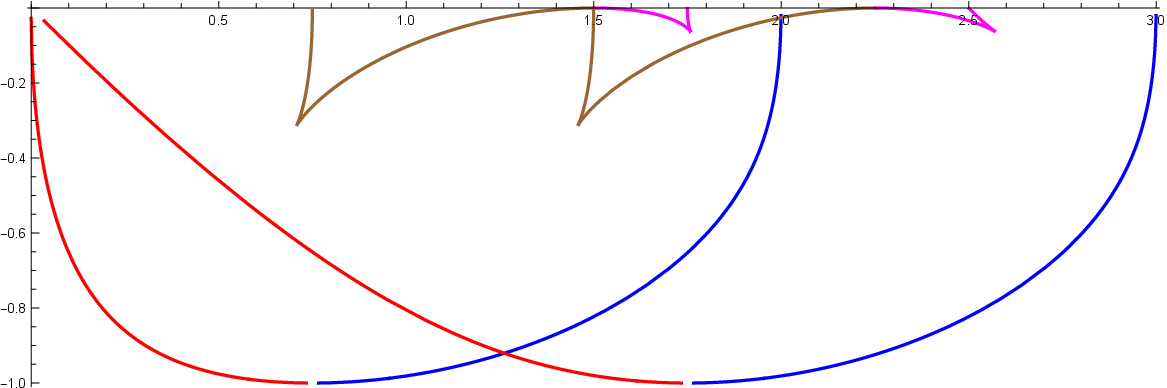}
\end{center}
\caption{\small  Arctic curves for the case of Example \ref{gappedex}, with $\kappa=\frac{3}{4}$ and $\mu=1$. Note the mapping of SE and SW branches between the NILP curve (right) and the touching path curve (left): while the left side of the gap is simply translated by $-\kappa=-3/4$ (mapping of SW gap branches), the right side is both sheared and translated by $-\kappa$ (mapping of SE gap branches).
The opposite holds for the outer branches (SE translated by $-1$ and SW simply sheared).}
\label{fig:gapped}
\end{figure}

\begin{example}[Effect of a gap]\label{gappedex}
 We consider a uniform distribution of endpoints on two segments $[0,k]$ and $[k+m,n+m]$, separated by a gap of size $m$, namely $a_i=i$, $i=0,1,...,k$ and $a_{k+i}=k+m+i$, $i=1,2,...,n-k$. With the scaling parameters $\kappa=k/n, \mu=m/n$, this corresponds to the large $n$ distribution of endpoints
$$ \al(u)=\left\{\begin{matrix}
u, & u \in [0,\kappa] \hfill \\
u+\mu, & u \in [\kappa,1] \hfill 
\end{matrix} \right. . $$
We deduce that
$$ x_0(t)= \sqrt{\frac{1-\frac{2\kappa}{t}}{1-\frac{2(1-\kappa)}{2+\mu-t}}}. $$
The left SW and right SE branches of arctic curve are given by Theorem \ref{mainthm}. We have
\begin{align*}&(SW):  \ \ \left\{\begin{matrix} X_{SW}(t) &= X^{\rm NILP}_{SW}(t)+Y^{\rm NILP}_{SW}(t) \hfill \\
Y_{SW}(t) &= Y^{\rm NILP}_{SW}(t)\hfill \end{matrix}
\right. \qquad t\in (-\infty,0],  \\
&(SE): \ \  \left\{\begin{matrix} X_{SE}(t) &= X^{\rm NILP}_{SW}(t)-1 \hfill \\
Y^{\rm gap}_{SE}(t) &= Y^{\rm NILP}_{SE}(t)\hfill \end{matrix}
\right. \qquad (t\in [2+\mu,\infty)). 
\end{align*}

For the gap branches, we have respectively
\begin{align*}&(SW \ {\rm gap}): \ \  \left\{\begin{matrix} X^{\rm gap}_{SW}(t) &= X^{\rm gap, NILP}_{SW}(t)-\kappa \hfill \\
Y^{\rm gap}_{SW}(t) &= Y^{\rm gap, NILP}_{SW}(t)\hfill \end{matrix}
\right. \qquad \qquad  \qquad t\in [2\kappa,\kappa(2+\mu)] \\
&(SE \ {\rm gap}):  \ \ \left\{\begin{matrix} X^{\rm gap}_{SE}(t) &= X^{\rm gap, NILP}_{SE}(t)+Y^{\rm gap, NILP}_{SE}(t)-\kappa \hfill \\
Y^{\rm gap}_{SE}(t) &= Y^{\rm gap, NILP}_{SE}(t)\hfill \end{matrix}
\right. \qquad t\in [\kappa(2+\mu),2\kappa+\mu]  
\end{align*}
where we have identified $t^*=\kappa(2+\mu)$ from the explicit expression for
$$
Y^{\rm NILP}(t)= -\frac{(1-x_0(t))^2}{x_0'(t)} 
= \frac{-4 (t-\kappa(2+\mu))^2}{t^2-2 t\kappa(2+\mu)+\kappa(2+\mu)(2\kappa+\mu)} \ 
\frac{\sqrt{1-\frac{2\kappa}{t}} 
\sqrt{1-\frac{2(1-\kappa)}{2+\mu-t}}}{\left(\sqrt{1-\frac{2\kappa}{t}} +
\sqrt{1-\frac{2(1-\kappa)}{2+\mu-t}}\right)^2} .
$$
This is illustrated in Figure \ref{fig:gapped}, where we have represented both NILP and touching path arctic curves. Note that the relationship between NILP and touching paths is ``reversed" for the gap, as the translation maps the SW branches and the shear composed with translation maps the SE branches.
\end{example}

The one-gap argument can be repeated for each gap individually  in the case of several gaps in the distribution of endpoints. 
The various shears/translations will depend on the positions of the gaps.

\begin{figure}
\begin{center}
\includegraphics[width=13cm]{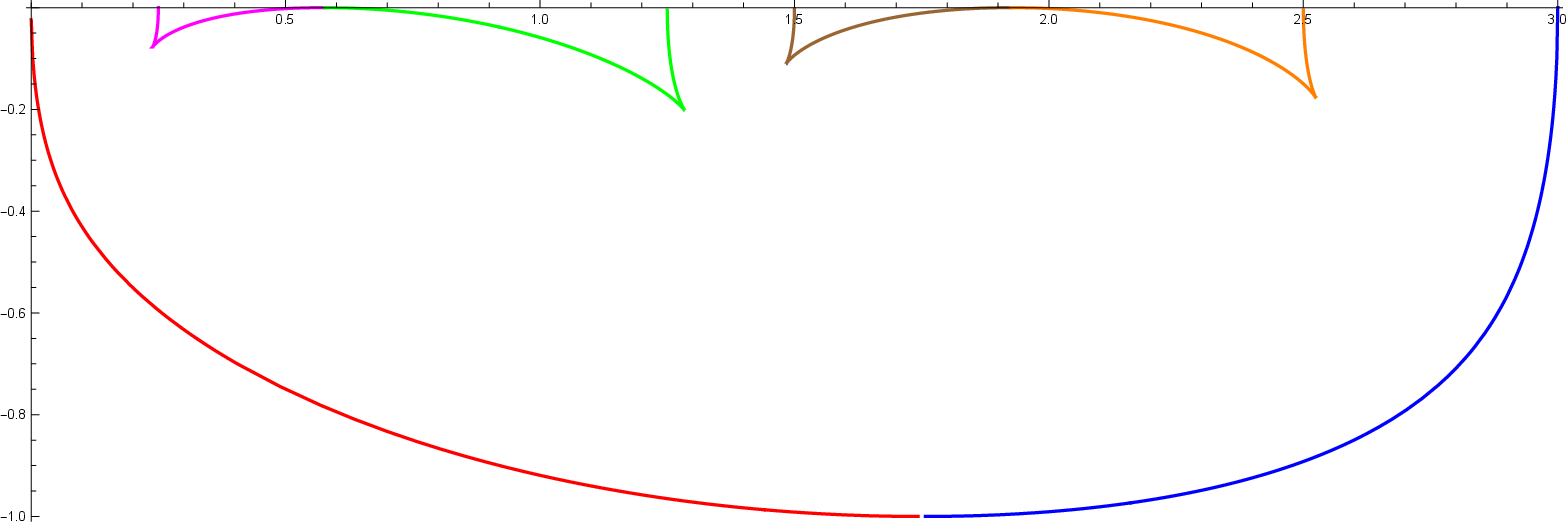}
\end{center}
\caption{\small  Arctic curve for the case of $r=2$ gaps of Example \ref{gappedex}, with $\kappa_1=\frac{1}{4},\kappa_2=\frac12$, and $\mu_1=\mu_2=1$.}
\label{fig:doublegap}
\end{figure}

\begin{example}[Effect of several gaps]\label{doublegapex}
We consider a distribution $\al$ which is  uniform on $r+1$ segments separated by $r$ gaps, of the form
$$ \al(u)= u +\mu_1+\mu_2+\cdots+\mu_i , \quad (u\in [\kappa_i,\kappa_{i+1}]), 1\leq i\leq r, $$
with $0\leq \kappa_1<\kappa_2 <\cdots <\kappa_{r+1}=1$ and $\mu_i>0$, and $\al(1)=\mu_1+\cdots+\mu_{r}+1$.
This leads to
$$ x_0(t)= \prod_{i=0}^r \sqrt{\frac{1+\frac{2 \kappa_{i+1}}{\mu_1+\cdots+\mu_i-t}}{1+\frac{2 \kappa_{i}}{\mu_1+\cdots+\mu_i-t}} },$$
with the convention that $\kappa_0=0$.
Repeating the analysis above for each single gap, we find that the arctic curve for the touching paths is entirely determined by that of the corresponding NILP, 
$(X(t),Y(t))=(X^{NILP}(t),Y^{NILP}(t))$ (from Theorem \ref{NILPthm}, with $\beta(u)=\al(u)+u$),  as follows:
\begin{itemize} 
\item{Left branch:}
$(X_0(t),Y_0(t))= (X(t)+Y(t),Y(t)),\ \  t\in (-\infty,0] $
\item{Left $i$-th gap:}
$(X_{2i-1}(t),Y_{2i-1}(t))= (X(t)-\kappa_i,Y(t)),\ \  t\in [2\kappa_i+\mu_1+\cdots +\mu_{i-1},t_i^*] $
\item{Right $i$-th gap:}
$(X_{2i}(t),Y_{2i}(t))= (X(t)+Y(t)-\kappa_i,Y(t)),\ \  t\in [t_i^*,2\kappa_i+\mu_1+\cdots +\mu_i] $
\item{Right branch:} 
$(X_{2r+1}(t),Y_{2r+1}(t))= (X(t)-1,Y(t)),\ \  t\in [1+\al(1),\infty) $
\end{itemize}
Here $t_i^*$ are the double zeros of the function $Y(t)$ (tangency points of the NILP arctic curve to the x axis). We represent in Figure \ref{fig:doublegap} the case of $r=2$ gaps with three uniform exit point segments, with $\kappa_1=\frac{1}{4},\kappa_2=\frac12,\kappa_3=1$, and $\mu_1=\mu_2=1$.
Here we find $t_1^*=\frac{13-\sqrt{41}}{8},t_2^*=\frac{13+\sqrt{41}}{8}$.
\end{example}

\subsection{General solution III: frozen regions}

In this section we handle the case of frozen regions. By frozen regions we mean the existence of points at which a macroscopic number of paths exit. Indeed, such multiply occupied exit points for the touching paths turn into 
segments where each integer point is an exit point in the corresponding NILP: such segments were coined as ``frozen" in \cite{DFGUI}. The previous arguments do not seem to allow to explore the top branches of the arctic curve for the touching paths.

Below we present some conjectures for these portions of the arctic curve via an example.

\begin{figure}
\begin{center}
\begin{minipage}{0.33\textwidth}
        \centering
        \includegraphics[width=4.5cm]{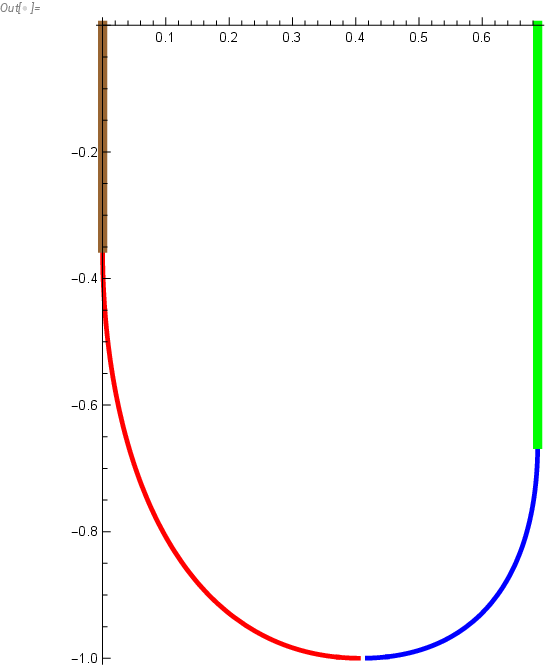} 
    \end{minipage}\hfill
     \begin{minipage}{0.66\textwidth}
        \centering
        \includegraphics[width=9.cm]{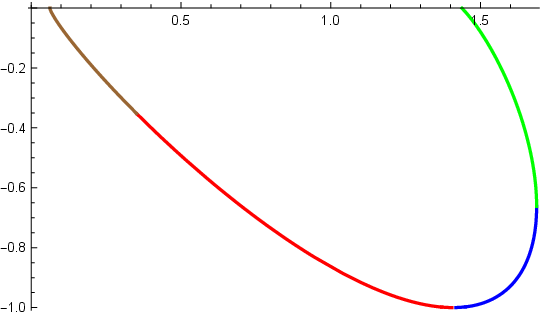}  
    \end{minipage}
\end{center}
\caption{\small Left: SW and SE branches of the arctic curve for the case of Example \ref{frozex}, with $\kappa=\frac{1}{16}$ and $\lambda=\frac{1}{4}$. Both SE and SW branches meet for $t\to \infty$ at the point
$(X_\infty,Y_\infty)=(209/512,-1)$.  The remaining NE and NW branches of the NILP configurations (in brown and green) are conjectured to map to
vertical segments.}
\label{fig:frozen}
\end{figure}

\begin{example}[Freezing on left and right]\label{frozex}
Assume that the first path exit point has a $k$-fold multiple edge while the last has an $\ell$-fold one, and each point in-between is an exit point, with respectively $k/n=\kappa$ and $\ell/n=\lambda$. This turns into two frozen regions of the NILP model, where a path exits at each integer point of  two segments $[1,k]$ and $[2n-k-2\ell+1,2n-k-\ell]$, while every other point of the segment
$[k+1,2n-k-2\ell]$ is an exit point as well. With the scaling limit $k=\kappa n,\ell=\lambda n$ as $n\to \infty$, the original distribution of endpoints is given by
$$ \al(x)=\left\{ \begin{matrix}0 & x\in [0,\kappa]\\
x-\kappa  & x\in [\kappa, 1-\lambda] \\
1-\kappa-\lambda & x\in [1-\lambda,1]\end{matrix}\right. \ \ \ \Rightarrow \ \ \ x_0(t)=\left(1-\frac{\kappa}{t}\right) \frac{t-2+\lambda+\kappa}{t-2+2\lambda-\kappa}
\sqrt{\frac{t-2+2\lambda+\kappa}{t-\kappa}} .$$
Note that the SW and SE branches of the arctic curve have finite endpoints $(X_0,Y_0)$ and $(X_1,Y_1)$ respectively, both obtained in the limit $t\to 2-\kappa-\lambda$,
while the meeting point $(X_\infty,Y_\infty)$ between the SE and SW branches of the arctic curve is obtained in the  $t\to \infty$ limit, with
\begin{align*}
(X_0,Y_0)&=(0,-(2-\kappa-\lambda)\sqrt{\frac{\lambda}{2-2\kappa-\lambda}} ), \ \ (X_\infty,Y_\infty)=(\frac12 (1-\kappa-\lambda)(1-\kappa+\lambda), -1),\\
 (X_1,Y_1)&=(1-\kappa-\lambda,-(2-\kappa-\lambda)\sqrt{\frac{\kappa}{2-\kappa-2\lambda}} ).\end{align*}
We display the SW and SE branches of the arctic curve for the case $\kappa=\frac{1}{16}$ and $\lambda=\frac{1}{4}$ in Figure \ref{fig:frozen} (left). Comparing with the NILP case (right), it is natural to expect that the NW and NE branches collapse onto vertical segments under the inverse of the sliding map (which we may understand as a ``squeezing" map),
whereas the SE branch is simply translated, and the SW branch is sheared as usual.
\end{example}

\section{The case of $q$-weighted paths}\label{sec:qweights}

\subsection{Partition function}

There is a natural q-deformation of the counting problem of previous sections for NILP, which includes an extra weight per path 
$p$ of the form $q^{Area(p)}$, where
$Area(p)$ denotes the area "under" the path $p$, obtained by locally weighting each vertical step $(i,j)-(i-1,j)$ by $q$ to the number of unit boxes
to its left until the W vertical boundary, namely $q^j$.  
Under the inverse sliding map to touching path configurations, the area weight of each individual path is unchanged, allowing for the enumeration of $q$-weighted 
touching paths, with partition function $Z_n(a_0,a_1,...,a_n\vert q)=Z_n^{NILP}(a_0,a_1+1,a_2+2,...,a_n+n\vert q)$
where the latter denotes the $q$-weighted NILP partition function,
given by the natural q-deformation of Theorem \ref{pfthm} \cite{DFG2}, leading to the following theorem. 

\begin{thm}\label{qpf}
The partition function for the $q$-weighted touching paths reads
$$Z_n(a_0,a_1,...,a_n\vert q)=Z_n^{NILP}(a_0,a_1+1,a_2+2,...,a_n+n\vert q)=q^{c_n}\frac{\Delta(q^{a_0},q^{a_1+1},q^{a_2+2},...,q^{a_n+n})}{\Delta(1,q,q^2,...,q^n)}, $$
where $c_n$ is independent of the $a_i$'s\footnote{The precise value of $c_n$ can be found in \cite{DFG2}, however in the following we will only use ratios of partition functions, independent of this normalization factor.}.
\end{thm}

\subsection{$q$-Tangent method: geodesics}

As we only used Theorem \ref{pfthm} as our starting point for the tangent method in the previous section, we
may now repeat verbatim the reasoning of Section \ref{gapsec} in the case of the q-weighted touching paths, using Theorem \ref{qpf} instead.
However, a subtlety arises: the tangent method relies on determining a family of "free path" curves tangent to the arctic curve. The
most likely free trajectories are lines for the unweighted case $q=1$, but are different curves for $q\neq 1$.
To determine these, consider the model of a single free path from say $(0,0)$ to $(\ell,p)$, with a weight $q^x$ for each vertical step at
abscissa $x$. The partition function reads
$$ Z_{\ell,p}= {p+\ell \choose p}_{\!\!\! q}, \quad {p+\ell \choose p}_{\!\!\! q}=\prod_{i=1}^p \frac{1-q^{i+\ell}}{1-q^i}.$$
To determine the most likely configuration of such paths, we cut them by a vertical line $x=c$ into two pieces, one from $(0,0)$
to say $(c,d)$ and one from $(c,d)$ to $(\ell,p)$. The most likely position $d=d^*(c)$ of the path on the $x=c$ line is obtained by maximizing
the corresponding contributions to $Z_{\ell,p}$: $Z_{c,d} Z_{\ell-c,p-d} q^{c(p-d)} $, where the last factor accounts for the area of 
the rectangle of size $c\times (p-d)$, not counted in $Z_{\ell-c,p-d}$. In the scaling limit where $(\ell,p,c,d)\sim n(\lambda,z,x,y)$, and 
$q\to q^{1/n}$,
this reads
\begin{align*} &{n (x+y) \choose n y}_{\!\!\!q^{1/n}} \, {n (\lambda+z-x-y) \choose n (z-y)}_{\!\!\!q^{1/n}} \,(q^{1/n})^{n^2 x(z-y)} \sim 
e^{n S(\lambda,z,x,y)} 
\end{align*}
with
\begin{align*}
&S(\lambda,z,x,y)= \int_{0}^y du\, {\rm Log}\left( \frac{q^{x+u}-1}{q^u -1}\right)
+\int_0^{z-y} du\, {\rm Log}\left( \frac{q^{\lambda-x+u}-1}{q^u -1}\right)+x(z-y) {\rm Log}(q).
\end{align*}

The largest contribution to the partition function is determined by the saddle-point solution $\partial_{y}S(\lambda,z,x,y)=0$, namely
\begin{equation}\frac{q^{x+y}-1}{q^y-1} \, \frac{q^{z-y}-1}{q^{\lambda+z-x-y}-1}\, q^{-x}=1\ \  \Rightarrow \ \ q^x \frac{q^z-1}{q^{\lambda+z}-1} +q^{z-y} \frac{q^\lambda-1}{q^{\lambda+z}-1} =1 .\label{geod}
\end{equation}
This is the equation of the geodesic (most likely trajectory) of a q-weighted path from $(0,0)$ to $(\lambda,z)$ (in rescaled variables).

\subsection{$q$-Tangent method: arctic curves}

We may now repeat the derivation of Section \ref{gapsec} in the q-deformed case. Recall, from Remark \ref{rmk:infgap}, that we can recover the SE and SW portions of the arctic curve by treating them as portions of infinitely large gaps. 

Starting with the same endpoint distribution \eqref{distrigap},
we move the rightmost path endpoint at the left of the gap by an amount $r$ into the gap, and compute the ratio of the partition functions
with moved point and with unmoved.
For simplicity, let us first assume that $q\geq 1$.
With obvious notations 
we obtain, taking the large $n$ scaling limit $q\to \bar q=q^{1/n}$, $r\to n \rho$, $a_i\to n \al(i/n)$, the ratio of partition functions reads,
using Theorem \ref{qpf},
\begin{align*}
H_{n,r}^{\bar q}=\frac{Z_{n,r}^{\bar q}}{Z_{n,0}^{\bar q}}&= \exp\left\{- \sum_{i=0}^{k-1}{\rm Log}\Big(\frac{\bar q^{a_i+i} -\bar q^{r+a_k+k}}{\bar q^{a_i+i} -\bar q^{a_k+k}}\Big)-\sum_{i=k+1}^n {\rm Log}\Big(\frac{\bar q^{a_i+i} -\bar q^{r+a_k+k}}{\bar q^{a_i+i} -\bar q^{a_k+k}}\Big) \right\} \ \sim e^{n S_0^q(\rho)} 
\end{align*}
with
\begin{align*}
S_0^q(\rho)&= \int_0^\kappa du\,  {\rm Log}\Big(\frac{q^{\al_1(u)+u}-q^{\rho+\al_1(\kappa)+\kappa}}{q^{\al_1(u)+u}-q^{\al_1(\kappa)+\kappa}}\Big) 
+\int_{\kappa}^1 du\,  {\rm Log}\Big(\frac{q^{\al_2(u)+u}-q^{\rho+\al_1(\kappa)+\kappa}}{q^{\al_2(u)+u}-q^{\al_1(\kappa)+\kappa}}\Big)  .
\end{align*}
Asymptotically for large $n$, after exiting at position $(a_k+r,0)\sim (n(\al_1(\kappa)+\rho),0)$ the path continues (along a geodesic) until it reaches the endpoint $(a_k+m,s)\sim n(\al_1(\kappa)+\mu,z)$, with the same scaling variables as in Section \ref{gapsec}. The contribution of the single path partition function reads (including the area of the rectangle of size $(a_k+r)\times p$ not counted by the q-binomial)
\begin{align*}
Y_{r,s}^{\bar q}&=  q^{p(a_k+r)}\, {m-r+s\choose s}_{\!\!\!\bar q} \sim e^{n S_1^q(\rho)}\\
S_1^q(\rho)&=z(\al_1(\kappa)+\rho) {\rm Log}(q)+ \int_0^{\mu-\rho} du\, {\rm Log}\left(\frac{q^{z+u}-1}{q^u-1} \right)  
\end{align*}
The leading contribution to the total partition function $\sum_r \frac{Z_{n,r}^{\bar q}}{Z_{n,0}^{\bar q}}Y_{r,s}^{\bar q}$ 
is attained at the saddle-point of the total action $S(\rho)=S_0(\rho)+S_1(\rho)$,
namely the most likely value of $\rho$ is determined by $\partial_\rho S(\rho)=0$:
$${\rm Log}(q)\,q^{\rho+\al_1(\kappa)+\kappa}\left\{\int_0^\kappa  \frac{du}{q^{\rho+\al_1(\kappa)+\kappa}-q^{\al_1(u)+u} }
+\int_{\kappa}^1 \frac{du}{q^{\rho+\al_1(\kappa)+\kappa}-q^{\al_2(u)+u}}\right\}+{\rm Log}\Big( q^z\frac{q^{\mu-\rho}-1}{q^{\mu-\rho+z}-1} \Big) =0 .$$
Introducing the functions and variables $x_0^q(t),t$ below, we have
$$ t= q^{\kappa+\rho+\al_1(\kappa)},\ \ x_0^q(t)=q^{-t \int_0^1 \frac{du}{t-q^{u+\al(u)}}}=q^{-t \left\{\int_0^\kappa\frac{du}{t-q^{u+\al_1(u)}}+\int_\kappa^1\frac{du}{t-q^{u+\al_2(u)}}\right\}} =q^z \frac{q^{\mu-\rho}-1}{q^{\mu-\rho+z}-1} , $$
hence
$$ q^{\mu -\rho}=\frac{1-q^{-z}x_0^q(t)}{1-x_0^q(t)} .$$

We now plug this into the geodesic equation \eqref{geod} with $\lambda=\mu-\rho$ to get
\begin{equation}(1-x_0^{q}(t)) \, q^x + x_0^{q}(t)\, q^{-y} =1 .\label{pregeod} \end{equation}
Note that we chose the origin $(x,y)=(0,0)$ at the exit point of the tangent path. In order to translate to the original coordinates, we must change 
$x\to x-\al_1(\kappa)-\rho=x-\theta+\kappa$, or $q^x\to q^{x+\kappa}/t$, with $t=q^{\theta}$.
Finally, the family of tangent geodesics to the arctic curves reads
\begin{equation}(1-x_0^{q}(t)) \, q^{\kappa+x} +t  x_0^{q}(t)\, q^{-y} =t .\label{tangeod}
\end{equation}

The envelope of this family is the SW portion of the gap arctic curve is given by
\begin{equation}\left\{ \begin{matrix} q^{X_{SW}^{q-gap}(t)}&= q^{-\kappa} \frac{t^2 (x_0^q)'(t)}{t(x_0^q)'(t) +x_0^q(t)\big(1-x_0^q(t)\big)}\hfill \\
q^{Y_{SW}^{q-gap}(t)}&= \frac{t(x_0^q)'(t) +x_0^q(t)\big(1-x_0^q(t)\big)}{t (x_0^q)'(t)+1-x_0^q(t)}\hfill \end{matrix}
\right. .\label{qgapSW}
\end{equation}
The range of the parameter $t$ is such that $\theta\in [\kappa+\al_1(\kappa),\theta^*]$, where $t^*=q^{\theta^*}$ is the unique root of $x_0^q(t^*)=1$,
with $\theta^*\in [\al_1(\kappa)+\kappa,\al_1(\kappa)+\kappa+\mu]$.

As before, the SE branch of the gap is obtained by considering the same setting but in terms of the dual paths, going from the E boundary to the N one. 
After reflecting the picture w.r.t. a vertical line,
the new distribution of endpoints is as before $\delta(u)=\al(1)-\al(1-u)$.
Recall that vertical steps at abscissa $i$ in original (touching) path
receive the weight $q^i$. For dual path vertical steps at abscissa $i$, after the reflection of the configuration the weight becomes 
$q^{a_n-i}$. The ratio of partition functions with displaced endpoint to that without becomes
$$ \tilde H_{n,r}=H_{n,r}^{\bar q}\Big\vert_{q\to q^{-1}\atop \al\to \delta} \sim e^{n \tilde S_0^q}. $$
Indeed all factors $q^{a_n}$ cancel between numerator and denominator, so the net effect is to change q-weights from $q\to q^{-1}$,
while the reflected distribution changes from $\al\to \delta$ (which also includes $\kappa\to 1-\kappa$ and $\rho\to \mu-\rho$ or $r\to m-r$).
The new function $x_0^q$ (now denoted $x_1^q$)  for the reflected configuration, with 
$\tilde t= t\vert_{q\to q^{-1},\al\to \delta}=t q^{-\al(1)-1}$, is given by
\begin{equation}\label{newx} 
x_1^q(t)= x_0^q(t)\Big\vert_{q\to q^{-1}\atop \al\to \delta} = q^{{\tilde t} \int_0^1 \frac{du}{\tilde t -q^{-\delta(u)-u} }}=\frac{1}{x_0^q(t)}.
\end{equation}
When applying the tangent method, we may reflect the free exiting path as well, upon also changing the q-weight from $q\to q^{-1}$,
up to an overall factor of ${\bar q}^{n z \al(1)}$ that does not affect the saddle-point equations. As a net result, we obtain the {\it same} 
geodesic equation \eqref{pregeod} in which we must perform the changes $x_0^q(t)\to x_1^q(t)=1/x_0^q(t)$, $q\to q^{-1}$, and $x\to -x$
to account for the vertical reflection. Finally, we must translate the origin of the frame in the same manner by taking $x\to x+\kappa-\theta$,
leading to the family of tangent geodesics
\begin{equation}(1-x_1^{q}(t)) \, q^{\kappa+x} +t  x_1^{q}(t)\, q^{y} =t .\label{tangeodnew}
\end{equation}
The envelope of this family is the SE portion of the gap arctic curve
\begin{equation} \left\{ \begin{matrix} q^{X_{SE}^{q-gap}(t)}&=q^{-\kappa}  \frac{t^2(x_1^q)'(t)}{t(x_1^q)'(t) +x_1^q(t)\big(1-x_1^q(t)\big)}\hfill \\
q^{Y_{SE}^{q-gap}(t)}&= \frac{t (x_1^q)'(t)+1-x_1^q(t)}{t(x_1^q)'(t) +x_1^q(t)\big(1-x_1^q(t)\big)}\hfill \end{matrix}
\right. .
\label{qgapSE}
\end{equation}
The range of the parameter $t$ is such that $\theta\in [\theta^*,\kappa+\al_1(\kappa)+\mu]$, where $t^*=q^{\theta^*}$ as above. 
Using the relation \eqref{newx}, and noting that the same will hold in the case of many separated gaps, we summarize the portions of arctic 
curve at the vicinity of any gap in the following theorem.

\begin{thm}
The existence of gaps of the form $[\al(\kappa),\al(\kappa)+\mu]$ 
in the distribution $\al$ gives rise to two extra portions of arctic curve to the E and W of each gap, with parametric equations
\begin{align}
&\left\{ \begin{matrix} q^{X_{SW}^{q-gap}(t)}&= q^{-\kappa} \frac{t^2 (x_0^q)'(t)}{t(x_0^q)'(t) +x_0^q(t)\big(1-x_0^q(t)\big)}\hfill \\
q^{Y_{SW}^{q-gap}(t)}&= \frac{t(x_0^q)'(t) +x_0^q(t)\big(1-x_0^q(t)\big)}{t (x_0^q)'(t)+1-x_0^q(t)}\hfill \end{matrix}
\right. \quad t\in [q^{\al(\kappa)+\kappa},t^*] ,\label{newqgapSW}\\
&\left\{ \begin{matrix} q^{X_{SE}^{q-gap}(t)}&=q^{-\kappa}  \frac{t^2(x_0^q)'(t)}{t(x_0^q)'(t) +1-x_0^q(t)}\hfill \\
q^{Y_{SE}^{q-gap}(t)}&= \frac{t(x_0^q)'(t) +x_0^q(t)\big(1-x_0^q(t)\big)}{t (x_0^q)'(t)+1-x_0^q(t)}\hfill \end{matrix}
\right.  \quad t\in [t^*,q^{\al(\kappa)+\kappa+\mu}] ,
\label{newqgapSE}
\end{align}
where
$$x_0^q(t)=q^{-t\int_0^1 \frac{du}{t-q^{\al(u)+u}}} ,$$
and $t^*$ is the double root of the equation $x_0^q(t)=1$ for $t\in [q^{\al(\kappa)+\kappa},q^{\al(\kappa)+\kappa+\mu}]$. 
\end{thm}

\begin{remark}
In the q-weighted NILP version after application of the sliding map, the arctic curve is analytic and reads:
$$\left\{ \begin{matrix} q^{X^{q-NILP}(t)}&=\frac{t^2 (x_0^q)'(t)}{t(x_0^q)'(t) +x_0^q(t)\big(1-x_0^q(t)\big)}\hfill\\
q^{Y^{q-NILP}(t)}&= \frac{t(x_0^q)'(t) +x_0^q(t)\big(1-x_0^q(t)\big)}{t (x_0^q)'(t)+1-x_0^q(t)}\hfill \end{matrix}\right.
,\quad t\in \R\setminus Im(q^\beta),$$
where $Im(\beta)$ is the range of the NILP endpoint distribution function $\beta(u)=\al(u)+u$.
We observe the same "translation/shear phenomenon" as in the non-q-deformed case, namely that
\begin{align*}X_{SE}^{q-gap}(t) &=X^{q-NILP}(t)+Y^{q-NILP}(t)-\kappa,\quad Y_{SE}^{q-gap}(t)=Y^{q-NILP}(t),\\
X_{SW}^{q-gap}(t) &=X^{q-NILP}(t)-\kappa,\quad Y_{SW}^{q-gap}(t)=Y^{q-NILP}(t).
\end{align*}
This also induces a shear relation between the two touching path curves:
$$X_{SE}^{q-gap}(t) =X_{SW}^{q-gap}(t)+Y_{SW}^{qNILP-gap}(t),\quad Y_{SE}^{q-gap}(t)=Y_{SW}^{q-gap}(t).$$
\end{remark}

As in the non q-deformed case, we recover the SW and SE branches of arctic curve in the general case, by considering either $\kappa=0$ or $\kappa=1$
and subsequently sending $\mu\to\infty$.
This results in the following theorem.

\begin{thm}\label{qnogapthm}
The SW and SE branches of arctic curve for q-weighted touching paths with a distribution $\al(x)$ of endpoints on the N boundary
reads
\begin{align}
&\left\{\begin{matrix}q^{X_{SW}^q(t)}&=\frac{t^2(x_0^q)'(t)}{t(x_0^q)'(t) +1-x_0^q(t)}\hfill \\
q^{Y_{SW}^{q}(t)}&= \frac{t(x_0^q)'(t) +x_0^q(t)\big(1-x_0^q(t)\big)}{t (x_0^q)'(t)+1-x_0^q(t)} \hfill \end{matrix}\right. ,\quad\quad t\in   (-\infty,1]\\
&\left\{\begin{matrix}q^{X_{SE}^q(t)}&=q^{-1} \frac{t^2 (x_0^q)'(t)}{t(x_0^q)'(t) +x_0^q(t)\big(1-x_0^q(t)\big)} \hfill \\
q^{Y_{SE}^q(t)}&= \frac{t(x_0^q)'(t) +x_0^q(t)\big(1-x_0^q(t)\big)}{t (x_0^q)'(t)+1-x_0^q(t)}\hfill \end{matrix}\right. ,\quad t\in [q^{\al(1)+1},\infty).
\end{align}
\end{thm}

So far we have only considered the case $q\geq 1$. To reach the inverse domain $q\leq 1$, one simply has to apply the transformation $q\to q^{-1}$,
under which  segments are mapped according to $[q^a,q^b]\to [q^{-b},q^{-a}]$, and we must take $x_0^q\to \bar x_0^q$, where
$$\bar x_0^q(t)= q^{t \int_0^1\frac{du}{t-q^{-\al(u)-u} }}=q^{\bar t\int_0^1 \frac{du}{\bar t - q^{\delta(u)+u}}},\qquad \bar t =q^{\al(1)+1}\, t.$$

\subsection{Examples}
\subsubsection{Uniform distribution, no gap}\label{sec:unif}

\begin{figure}
\begin{center}
\includegraphics[width=7.cm]{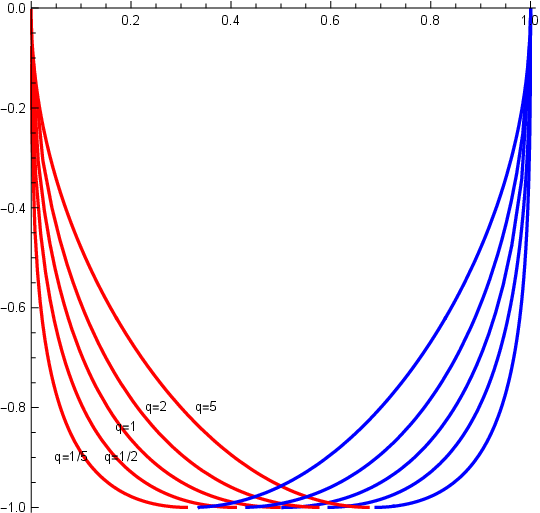}
\end{center}
\caption{\small  Arctic curve for q-weighted touching paths in the case of uniform distribution $\al(u)=u$, for $q=1/5,1/2,1,2,5$.}
\label{fig:qnogap}
\end{figure}

We display in Figure \ref{fig:qnogap} the arctic curve of a typical gapless case, as given by Theorem \ref{qnogapthm}, with
\[
\al(u)=u  \ \ \Rightarrow \ \ x_0^q(t) = \left( \frac{1-tq^{-2}}{1-t}\right)^{1/2}
\]
as a function of $q=\frac15,\frac12,1,2,5$.

\subsubsection{Uniform distribution, one gap}\label{sec:onegap}

We consider the one-gap case with
$$\al(u)=\left\{ \begin{matrix} 
u & u\in [0,\kappa] \\
u+\mu & u\in [\kappa,1]
\end{matrix}\right. \ \ \Rightarrow \ \ x_0^q(t)=q^{-t \int_0^1 \frac{du}{t-q^{\al(u)+u}}}=\left(\frac{1-t q^{-2\kappa}}{1-t}\frac{1-t q^{-2-\mu}}{1-t q^{-2\kappa-\mu}}\right)^{1/2} . $$

\begin{figure}
\begin{center}
\includegraphics[width=10cm]{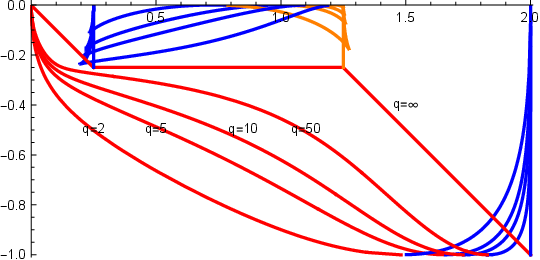}
\end{center}
\caption{\small  Arctic curve for q-weighted touching paths in the case of a gapped otherwise uniform distribution, for $q=2,5,10,50$ and $\infty$.}
\label{fig:qonegap}
\end{figure}

We display in Figure \ref{fig:qonegap} the arctic curve for the case $\kappa=\frac14$ and $\mu=1$ as a function of $q=2,5,10,50$  and $\infty$.
The latter corresponds to the tropical limit of paths with maximal area, leading to a piecewise linear  arctic curve.

\subsubsection{Uniform distribution, one multiple outlet}\label{sec:oneclump}

\begin{figure}
\begin{center}
\includegraphics[width=7.cm]{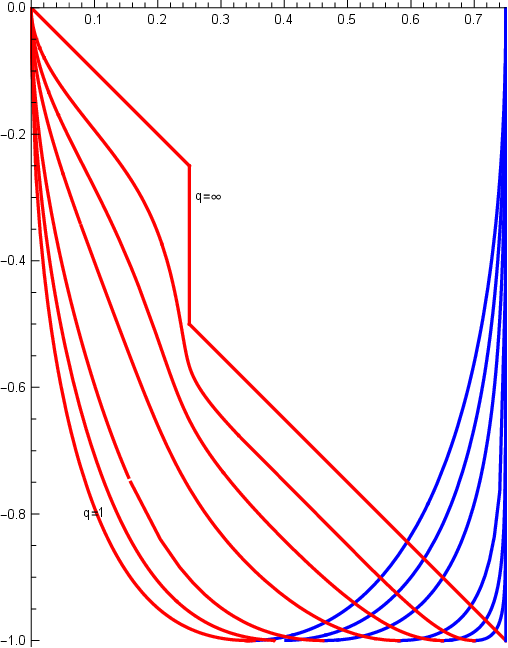}
\end{center}
\caption{\small  Arctic curve for q-weighted touching paths in the case of a uniform distribution with a multiple exit vertex with $\kappa=\lambda=1/4$, for $q$ between 1 and $\infty$.}
\label{fig:qbunchmiddle}
\end{figure}

We consider the situation of a freezing outlet, multiply  occupied by $\ell=\lambda n$ paths, at position $k=\kappa n$, with distribution
$$\al(u)=\left\{ \begin{matrix}u & u\in [0,\kappa]\\
\kappa & u\in [\kappa,\kappa+\lambda]\\
u-\lambda & u\in [\kappa+\lambda,1]\end{matrix} \right. \ \ \Rightarrow \ \ x_0^q(t)= \left(\frac{1-q^{-2\kappa-\lambda}t}{1-t}\frac{1-q^{-2+\lambda}t}{1-q^{-2\kappa}t}\right)^{1/2}.$$
We display in Figure \ref{fig:qbunchmiddle} the arctic curve for the case $\kappa=\lambda=\frac14$, for a few values of $q$ between $1$ and $\infty$.
The limit $q\to \infty$ is the tropical limit of largest area paths, whose arctic curve is piecewise linear.

\subsubsection{Uniform distribution, one multiple outlet next to a gap}\label{sec:clumpgap}

\begin{figure}
\begin{center}
\includegraphics[width=8.cm]{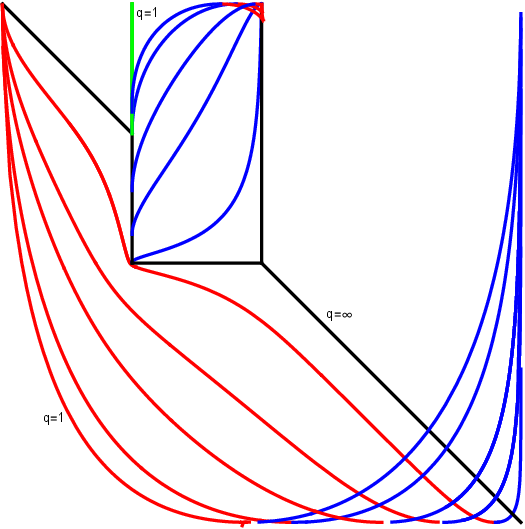}
\end{center}
\caption{\small  Arctic curve for q-weighted touching paths in the case of a uniform distribution with a multiple exit vertex with $\kappa=\lambda=1/4$, followed by a gap of size $\mu=1/4$, for $q$ between 1 and $\infty$. The green segment corresponds to the closing of the frozen region under the inverse sliding map.}
\label{fig:qbunchfreeze}
\end{figure}

We consider the situation of a freezing outlet, multiply  occupied by $\ell=\lambda n$ paths, at position $k=\kappa n$, immediately followed by a gap of size $m=\mu n$, with distribution
$$\al(u)=\left\{ \begin{matrix}u & u\in [0,\kappa]\\
\kappa & u\in [\kappa,\kappa+\lambda]\\
u-\lambda+\mu & u\in [\kappa+\lambda,1]\end{matrix} \right. \ \ \Rightarrow \ \ 
x_0^q(t)= \frac{1-q^{-2\kappa-\lambda}t}{1-q^{-2\kappa}t}\left(\frac{1-q^{-2\kappa}t}{1-t}\frac{1-q^{-2+\lambda-\mu}t}{1-q^{-2\kappa-\lambda-\mu}t}\right)^{1/2}.$$
We display in Figure \ref{fig:qbunchfreeze} the arctic curve for the case $\kappa=\lambda=\frac14$, for a few values of $q$ between $1$ and $\infty$.
The limit $q\to \infty$ is the tropical limit of largest area paths, whose arctic curve is piecewise linear.

\section{Numerical Simulations}

\subsection{Random Sampling}\label{sec:simulations}

In this section we describe the algorithm used to generate random path configurations approximately sampled from the Gibbs measure of the vertex models, that is, with probability proportional to the weight of the configuration. Random sampling algorithms are well-utilized in the study of vertex models, see for example \cite{allison2005numerical,keating2018random,lyberg2017density,propp1998coupling} for a non-exhaustive list. Below we describe briefly describe the well-know Metropolis-Hastings algorithm \cite{hastings1970,metropolis1953} and then highlight the details specific to our model. We allow general $t$ and do not restrict ourselves to $t=0$. After describing the algorithm we provide several examples of generated random configurations.

As mentioned above, we utilize Metropolis-Hastings algorithm in our simulations. Briefly, given a finite state space $X$ and target distribution $\pi$ on this state space, the Metropolis-Hastings algorithm is a way to construct a Markov chain on $X$ with $\pi$ as its stationary distribution. To begin one needs, as a base Markov chain, any irreducible Markov chain on $X$. Let $q(x,y)$ be the probability to transition from $x$ to $y$, $x,y\in X$, for the base chain. Using this base chain, we construct our desired MArkov chain as follows: Suppose we are at state $x$, we propose a jump to state $y$ with probability $q(x,y)$. This proposal is then accepted with probability
\[
\min \left\{1,\frac{\pi(y)q(y,x)}{\pi(x)q(x,y)}\right\}
\]
and rejected otherwise. Note that if the base chain is symmetric, then the $q$ dependence cancels. This process defines a new irreducible Markov chain on $X$ with transition probabilities
\[
p(x,y) = 
\begin{cases}
    q(x,y) \min \left\{1,\frac{\pi(y)q(y,x)}{\pi(x)q(x,y)}\right\}, & x\ne y \\
    1-\sum_{z:z\ne x}p(x,z), & x=y
\end{cases}.
\]
One can check that this Markov chain has $\pi$ as its stationary distribution by showing it satisfies the detailed-balance equations. By running this Markov chain for a large number of steps, one can approximately sample for the desired target distribution.

Recall that our vertex model consists of $n$ paths colored $1$ to $n$ with exactly one path of each color. The paths enter a square grid of $n$ rows with the $i$-th color path entering in the $i$-th row from the top. The path exit the top boundary of the grid at prescribed columns. Label the columns of the vertex model $1$ to $m$ from left to right with $m$ denoting the rightmost column at which a path exits the top boundary. We fix the columns at which paths exit and the multiplicity at each column, but allow the color of the paths at each exit location to vary.  

For the vertex model our target distribution is clear. We would like to sample from the Gibbs measure
\begin{equation}\label{eq:Gibbs}
    \pi(\mathcal{C}) = \frac{\tilde w(\mathcal{C})q^{Area(\mathcal{C})}}{\tilde Z_n(x_1,\ldots,x_n|t,q)}
\end{equation}
where for any path configuration $\mathcal{C}$ we have $\tilde w(\mathcal{C})$ is the weight of the colored path configuration using the weights in \eqref{eq:coloredweights}, $Area(\mathcal{C})$ counts the sum of the areas of each path as described in Section \ref{sec:qweights}, and $\tilde Z_n$ is the partition function of the model given in \eqref{eq:Zfree} now including the area weights.

To construct a base chain, one looks for local moves under which the space of configurations is connected. For models of up-right paths on the rectangular domains we consider, if the  boundary conditions are fixed, it is known that the space of configurations is connected under corner flips of the type shown in the left image of Figure \ref{fig:cornerflip}. In order to allow faster traversal through the state space, we allow larger flips which we now describe.

Fix a color $c$ and a column $k$. If the path of color $c$ does not exit a vertex to the east in column $k$ then a flip is not possible. Otherwise, let $h$ be the height at which the path exits the column to the east (counting the rows from top to bottom), $h_1$ the height the path exits column $k-1$ to the east, and $h_0$ the height the path exits column $k+1$ to the east. By convention we set $h_1=c$ if $k-1=0$ and $h_0=1$ if the path exits column $k+1$ at the top boundary. To execute a flip we choose a height $h'$ uniformly from $\{h_0,h_0+1,\ldots,h_1\}$ and set the height at which the path exits column $k$ to $h'$. See the right image in Figure \ref{fig:cornerflip} for an example. In terms of the weight, the new configuration loses a factor of $x_h$ but gains a factor of $x_{h'}$. The area weight changes by a factor of $q^{h'-h}$. As for the $t$-weight, one must check at each modified vertex in columns $k$ and $k+1$ how the factors of $t$ are affected as it depends on the rest of the paths in the configuration.

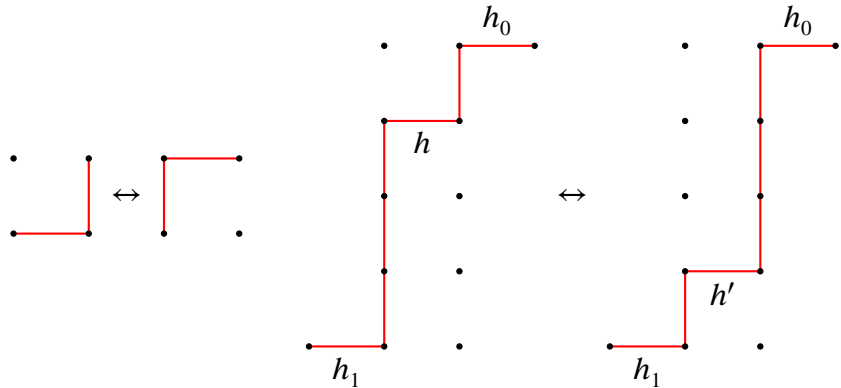
\begin{figure}
    \centering
    $
    \begin{tikzpicture}[baseline=(current bounding box).center]
        \draw[red, thick] (0,0)--(1,0)--(1,1);
        \draw[fill=black] (0,0) circle (1pt);
        \draw[fill=black] (1,0) circle (1pt);
        \draw[fill=black] (0,1) circle (1pt);
        \draw[fill=black] (1,1) circle (1pt);

        \node at (1.5,0.5) {$\leftrightarrow$};
        \begin{scope}[shift={(2,0)}]
            \draw[red, thick] (0,0)--(0,1)--(1,1);
            \draw[fill=black] (0,0) circle (1pt);
            \draw[fill=black] (1,0) circle (1pt);
            \draw[fill=black] (0,1) circle (1pt);
            \draw[fill=black] (1,1) circle (1pt);
        \end{scope}
    \end{tikzpicture}
    \qquad
    \begin{tikzpicture}[baseline=(current bounding box).center]
        \draw[red, thick] (0,0)--(1,0)--(1,3)--(2,3)--(2,4)--(3,4);
        \draw[fill=black] (0,0) circle (1pt);
        \draw[fill=black] (1,0) circle (1pt);
        \draw[fill=black] (1,1) circle (1pt);
        \draw[fill=black] (1,2) circle (1pt);
        \draw[fill=black] (1,3) circle (1pt);
        \draw[fill=black] (1,4) circle (1pt);
        \draw[fill=black] (2,4) circle (1pt);
        \draw[fill=black] (2,3) circle (1pt);
        \draw[fill=black] (2,2) circle (1pt);
        \draw[fill=black] (2,1) circle (1pt);
        \draw[fill=black] (2,0) circle (1pt);
        \draw[fill=black] (3,4) circle (1pt);
        \node[below] at (0.5,0) {$h_1$};
        \node[above] at (2.5,4) {$h_0$};
        \node[below] at (1.5,3) {$h$};

        \node at (3.5,2) {$\leftrightarrow$};
        \begin{scope}[shift={(4,0)}]
            \draw[red, thick] (0,0)--(1,0)--(1,1)--(2,1)--(2,4)--(3,4);
            \draw[fill=black] (0,0) circle (1pt);
            \draw[fill=black] (1,0) circle (1pt);
            \draw[fill=black] (1,1) circle (1pt);
            \draw[fill=black] (1,2) circle (1pt);
            \draw[fill=black] (1,3) circle (1pt);
            \draw[fill=black] (1,4) circle (1pt);
            \draw[fill=black] (2,4) circle (1pt);
            \draw[fill=black] (2,3) circle (1pt);
            \draw[fill=black] (2,2) circle (1pt);
            \draw[fill=black] (2,1) circle (1pt);
            \draw[fill=black] (2,0) circle (1pt);
            \draw[fill=black] (3,4) circle (1pt);
            \node[below] at (0.5,0) {$h_1$};
            \node[above] at (2.5,4) {$h_0$};
            \node[below] at (1.5,1) {$h'$};
        \end{scope}
    \end{tikzpicture}
    $
    \caption{Right: A basic corner flip. Left: A large corner flip in which the new height is chosen uniformly from $\{h_0,h_0+1,\ldots,h_1\}$.}
    \label{fig:cornerflip}
\end{figure}

The models we consider, however, do not have fixed boundary conditions as we consider all possible colorings of the top boundary. To handle this we introduce a new non-local move we call a \emph{color swap}. Consider paths of color $c$ and $c+1$. Suppose that in column $k$ the paths are at the same vertex $v$ with one of the paths exiting to the east and the other to the north. The paths will either meet again at vertex $v'$ in some column $\ell>k$ or never meet again before exiting at the top boundary of the domain. In either case, we swap the colors the two paths from where they exit $v$ to where they enter $v'$, or from where they exit $v$ to where they leave the domain. See Figure \ref{fig:colorswap} for an example. Note that in the latter case, we have changed the coloring of the top boundary. This type of boundary color swap was used previously \cite{corteel2022vertex,keating2021equivalences} in a non-algorithmic context to prove certain identities between partitions functions. 

Note that a color swap does not change the the locations at which paths exit vertices to the east and thus the configurations before and after the swap will have the same $x$-weight and the same area weight. In fact, when one restricts to swaps between colors $c$ and $c+1$, the only change in weight comes the $t$-weight at vertex $v$. If the smaller color is the one exiting $v$ to the north, then the swap increases the power of $t$ by one, otherwise the power of $t$ decreases by one.

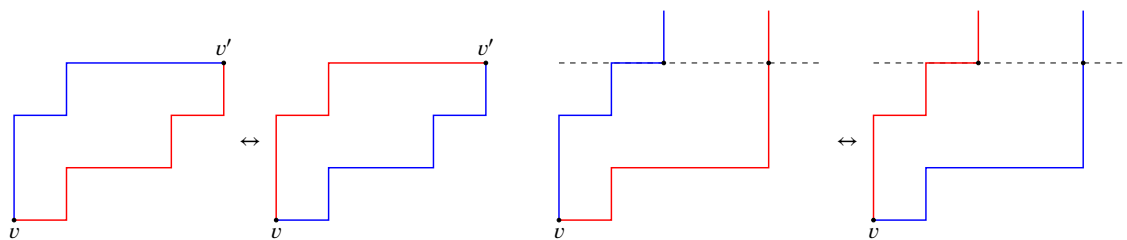
\begin{figure}
    \centering
    \resizebox{0.9\textwidth}{!}{$
    \begin{tikzpicture}
        \draw[blue, thick] (0,0)--(0,2)--(1,2)--(1,3)--(4,3);
        \draw[red, thick] (0,0)--(1,0)--(1,1)--(3,1)--(3,2)--(4,2)--(4,3);
        \draw[fill=black] (0,0) circle (1pt);
        \draw[fill=black] (4,3) circle (1pt);
        \node[below] at (0,0) {$v$};
        \node[above] at (4,3) {$v'$};
        \node at (4.5,1.5) {$\leftrightarrow$};
        \begin{scope}[shift={(5,0)}]
            \draw[red, thick] (0,0)--(0,2)--(1,2)--(1,3)--(4,3);
            \draw[blue, thick] (0,0)--(1,0)--(1,1)--(3,1)--(3,2)--(4,2)--(4,3);
            \draw[fill=black] (0,0) circle (1pt);
            \draw[fill=black] (4,3) circle (1pt);
            \node[below] at (0,0) {$v$};
            \node[above] at (4,3) {$v'$};
        \end{scope}
    \end{tikzpicture}
    \qquad 
    \begin{tikzpicture}
        \draw[dashed] (0,3)--(5,3);
        \draw[blue, thick] (0,0)--(0,2)--(1,2)--(1,3)--(2,3)--(2,4);
        \draw[red, thick] (0,0)--(1,0)--(1,1)--(4,1)--(4,4);
        \draw[fill=black] (0,0) circle (1pt);
        \draw[fill=black] (2,3) circle (1pt);
        \draw[fill=black] (4,3) circle (1pt);
        \node[below] at (0,0) {$v$};
        \node at (5.5,1.5) {$\leftrightarrow$};
        \begin{scope}[shift={(6,0)}]
            \draw[dashed] (0,3)--(5,3);
            \draw[red, thick] (0,0)--(0,2)--(1,2)--(1,3)--(2,3)--(2,4);
            \draw[blue, thick] (0,0)--(1,0)--(1,1)--(4,1)--(4,4);
            \draw[fill=black] (0,0) circle (1pt);
            \draw[fill=black] (2,3) circle (1pt);
            \draw[fill=black] (4,3) circle (1pt);
            \node[below] at (0,0) {$v$};
        \end{scope}
    \end{tikzpicture}
    $}
    \caption{The two types of color swaps. On the left, the case in which the paths separate then later meet. On the right, the case in which the path separate then exit the top boundary of the domain without meeting. The rest of the configuration is left unchanged.}
    \label{fig:colorswap}
\end{figure}

With this our base chain can be given as follows.

\begin{enumerate}
    \item Choose a color $c$ uniformly from $\{1,2,\ldots,n\}$.
    \item Choose a column $k$ uniformly from $\{1,2,\ldots,m\}$.
    \item Propose a move according to the following rules:
    \begin{enumerate}
        \item With probability\footnote{The probability of selecting different moves is somewhat arbitrary. One needs to choose the probability to color swap with $c+1$ to be the same as that to color swap with $c-1$ in order for the transition probabilities to be symmetric. Otherwise, any non-zero probabilities will do.} $\frac{1}{3}$, propose a flip of color $c$ in column $k$ if possible, otherwise leave the configuration unchanged. 
        \item With probability $\frac{1}{3}$, propose a swap starting in column $k$ with color $c+1$  if possible, otherwise leave the configuration unchanged.
        \item With probability $\frac{1}{3}$, propose a swap starting in column $k$ with color $c-1$  if possible, otherwise leave the configuration unchanged.
    \end{enumerate}
\end{enumerate}
One may check that the transition probabilities for the base chain are symmetric. Let $\mathcal{C}'$ be configuration resulting from the proposal. The acceptance probability is given by
\[
\min\left\{1, \frac{\pi(\mathcal{C}')}{\pi(\mathcal{C})}\right\} = \min\left\{1, \frac{\tilde w(\mathcal{C}')}{\tilde w(\mathcal{C})}\right\}
\]
where we note that the partition function cancels from the probability. More explicitly, we see that the probability a flip is then accepted is
\[
\min\left\{1, \frac{x_{h'}}{x_h}q^{h'-h}t^{\delta}\right\}
\]
where $\delta$ counts the change in the power of $t$ between the original and proposed configurations. A color swap is accepted with probability
\[
\min\left\{1,t^{\pm 1}\right\}
\]
where we pick $+$ if in the proposed configuration the smaller color now lies below the larger color, and $-$ otherwise. Note that with $t=0$ the algorithm works without issue, with proposals lowering the power of $t$ always accepted and those increasing the power of $t$ always rejected.

\subsection{Simulations}
Here we present a variety of simulations for various values of $t$, $q$, and several different boundary conditions. In all of these simulations we set the $x$-weights to 1. Recall that the colored vertex model at $t=0$ is equivalent to the colorblind vertex model at $t=0$. In the $t=0$ case, we superimpose the computed arctic curve over the simulation.

\begin{itemize}
    \item Figure \ref{fig:DWBCsim} shows simulations of the colored DWBC for various values of $t$ and $q$.
    \item Figure \ref{fig:Gapsim} shows simulations for a boundary condition with a single gap, as in Section \ref{sec:onegap}, for a various values of $t$ and $q$.
    \item  Figure \ref{fig:Clumpsim} shows simulations for three different boundary conditions with frozen regions, that is, regions in which a macroscopic number of paths exit at the same points on the top boundary, all at $t=0$ and $q=1$. The boundary conditions include those in Sections \ref{sec:oneclump} and \ref{sec:clumpgap}.
    \item Figure \ref{fig:tnear1} shows the colored DWBC for $q=1$ and values of $t>1$ of the form $t=\tau^{1/n}$ where $n$ is the number of paths.
    \item Figure \ref{fig:t>1} shows the colored DWBC for  $t=2$ and $q=1$.
    \item Figure \ref{fig:qinfinity} shows simulation of the tropical limit $q\to\infty$.
\end{itemize}

\begin{figure}
    \centering
    \resizebox{\textwidth}{!}{
    \begin{tabular}{ccccc}
        & $q=0.2$ & $q=0.5$ & $q=1.0$ & $q=5$ \\
        $t=0$ & 
        $\vcenter{\hbox{\includegraphics[width=0.3\textwidth]{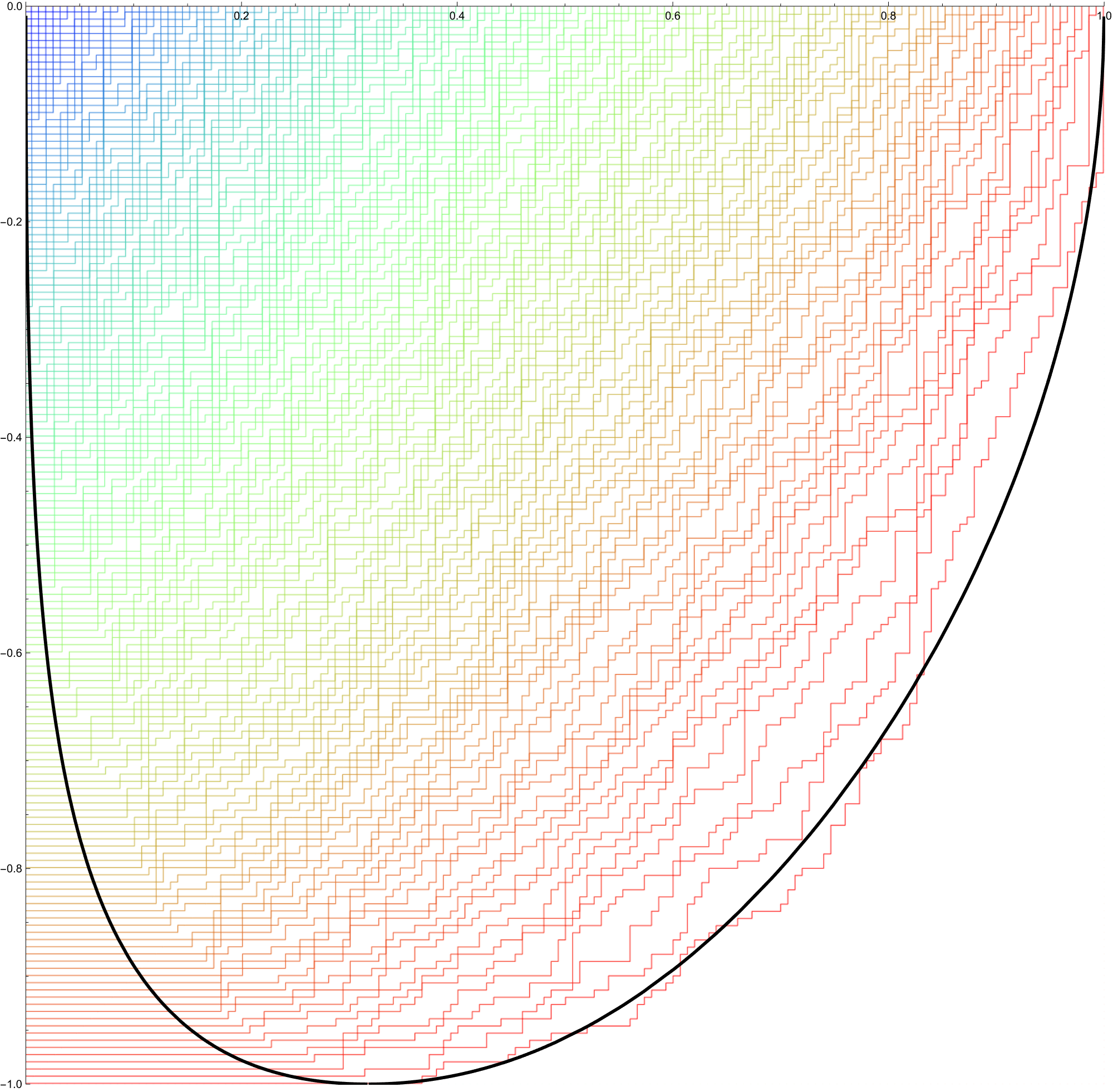}}}$
        & 
        $\vcenter{\hbox{\includegraphics[width=0.3\textwidth]{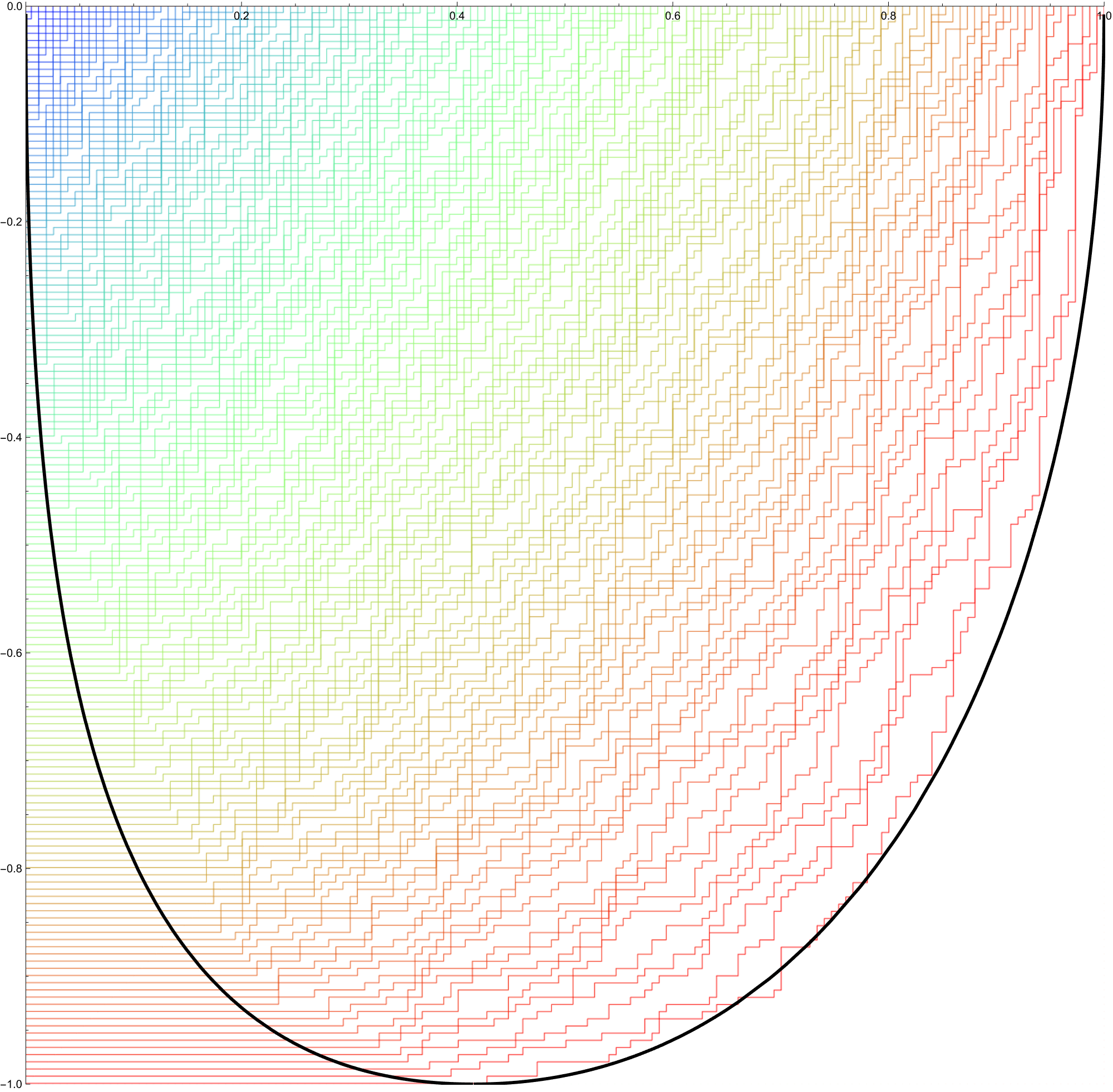}}}$
        & 
        $\vcenter{\hbox{\includegraphics[width=0.3\textwidth]{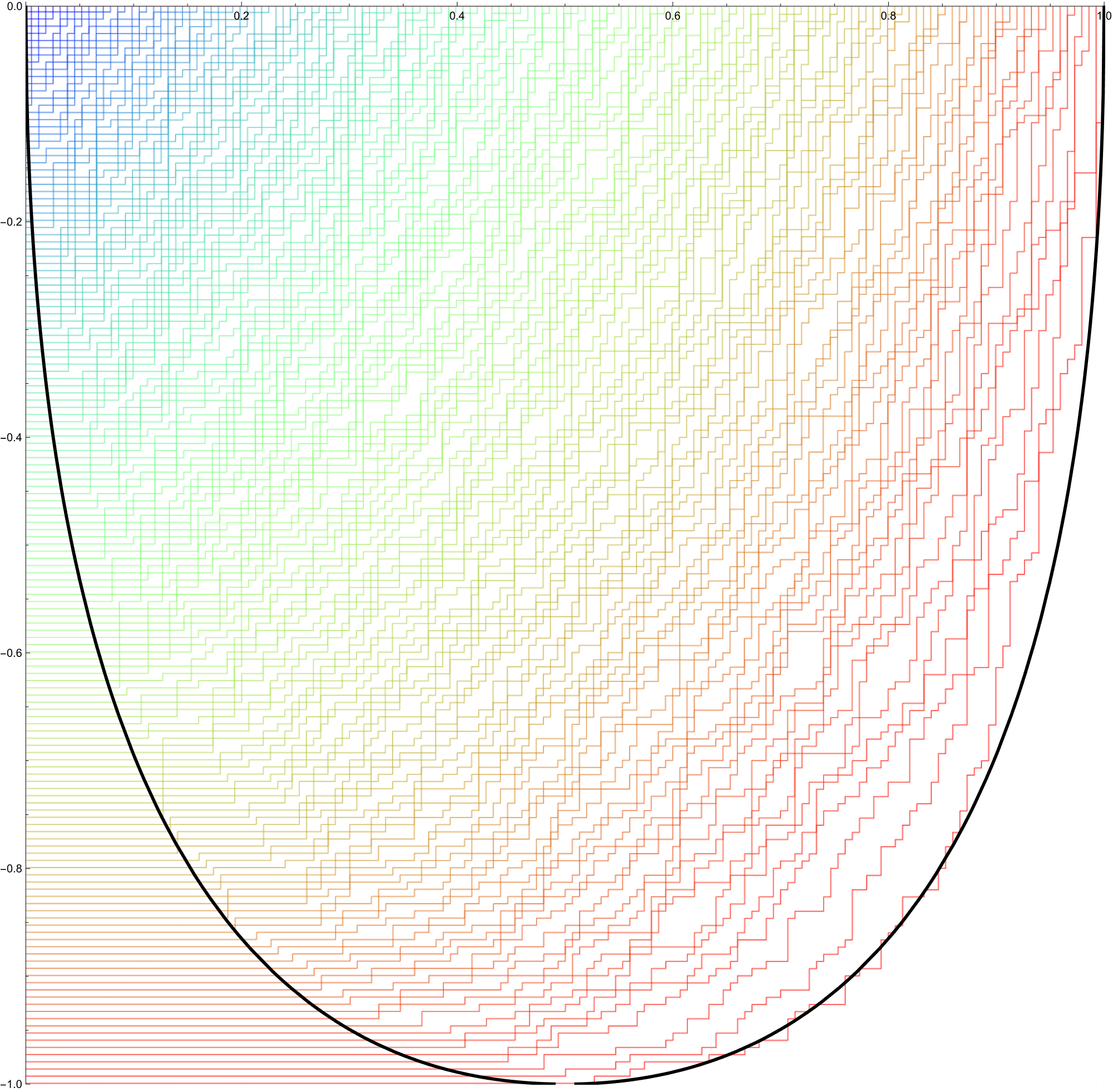}}}$
        & 
        $\vcenter{\hbox{\includegraphics[width=0.3\textwidth]{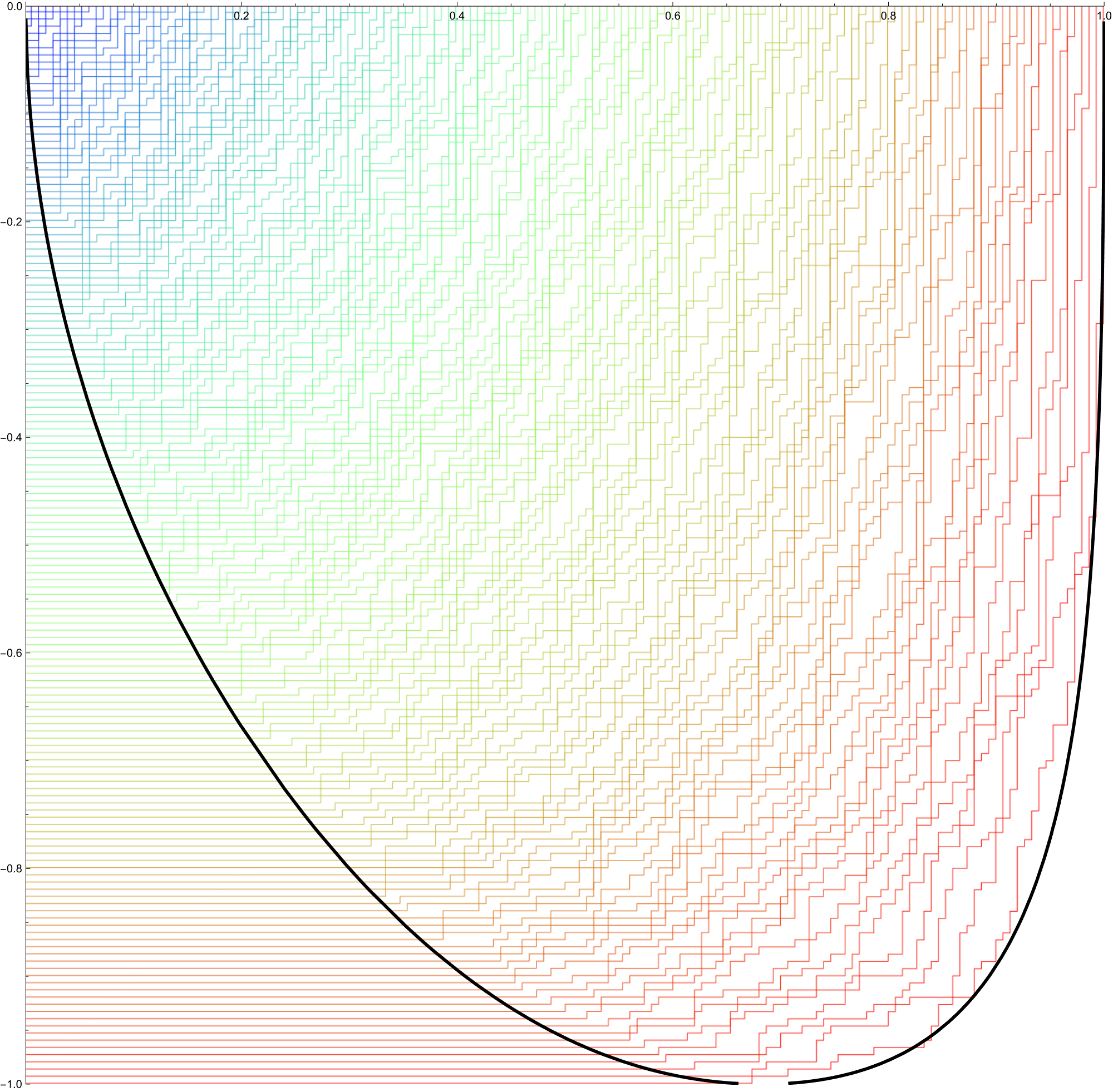}}}$
        \\
        $t=0.3$ & 
        $\vcenter{\hbox{\includegraphics[width=0.3\textwidth]{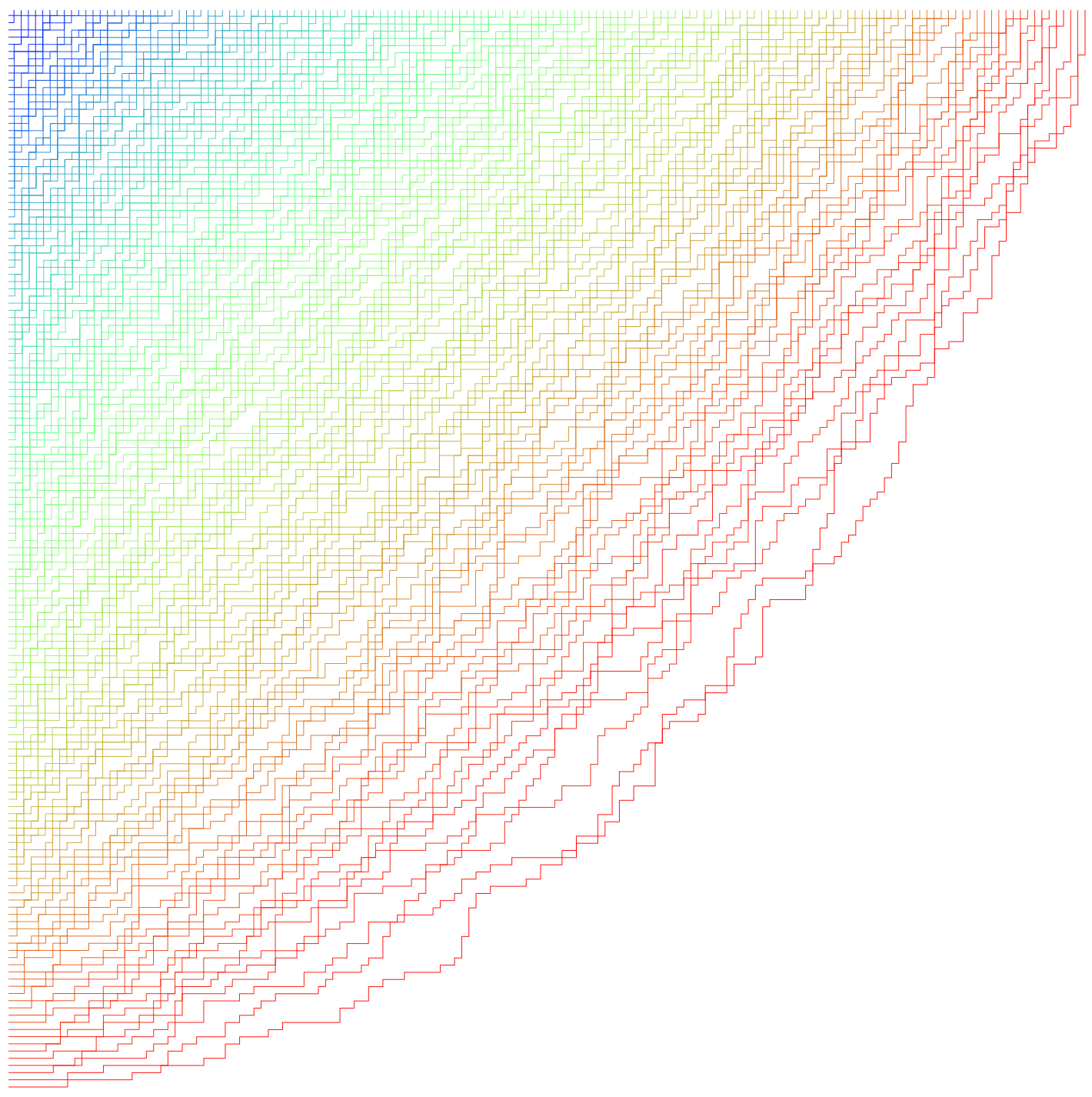}}}$
        & 
        $\vcenter{\hbox{\includegraphics[width=0.3\textwidth]{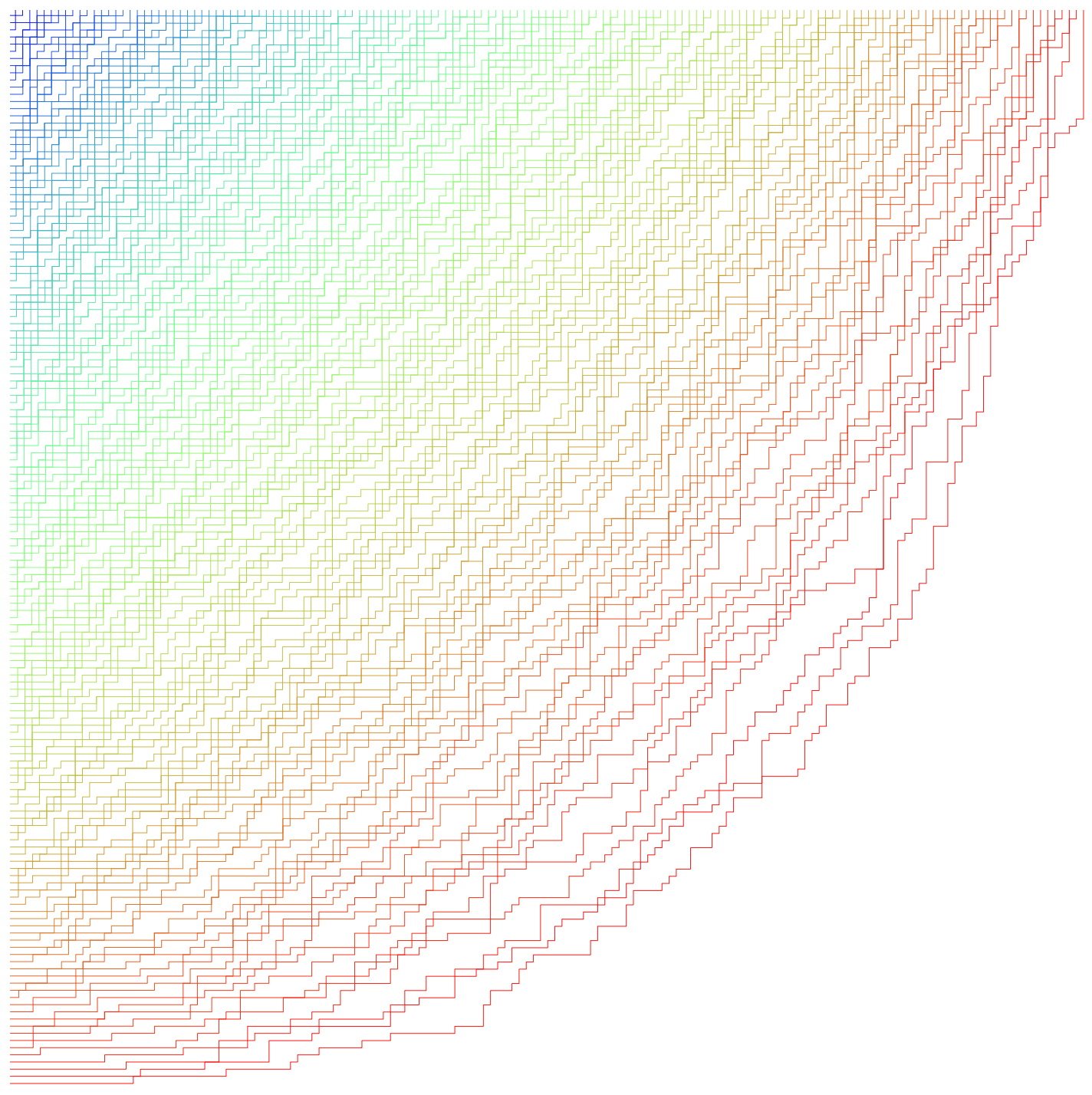}}}$
        & 
        $\vcenter{\hbox{\includegraphics[width=0.3\textwidth]{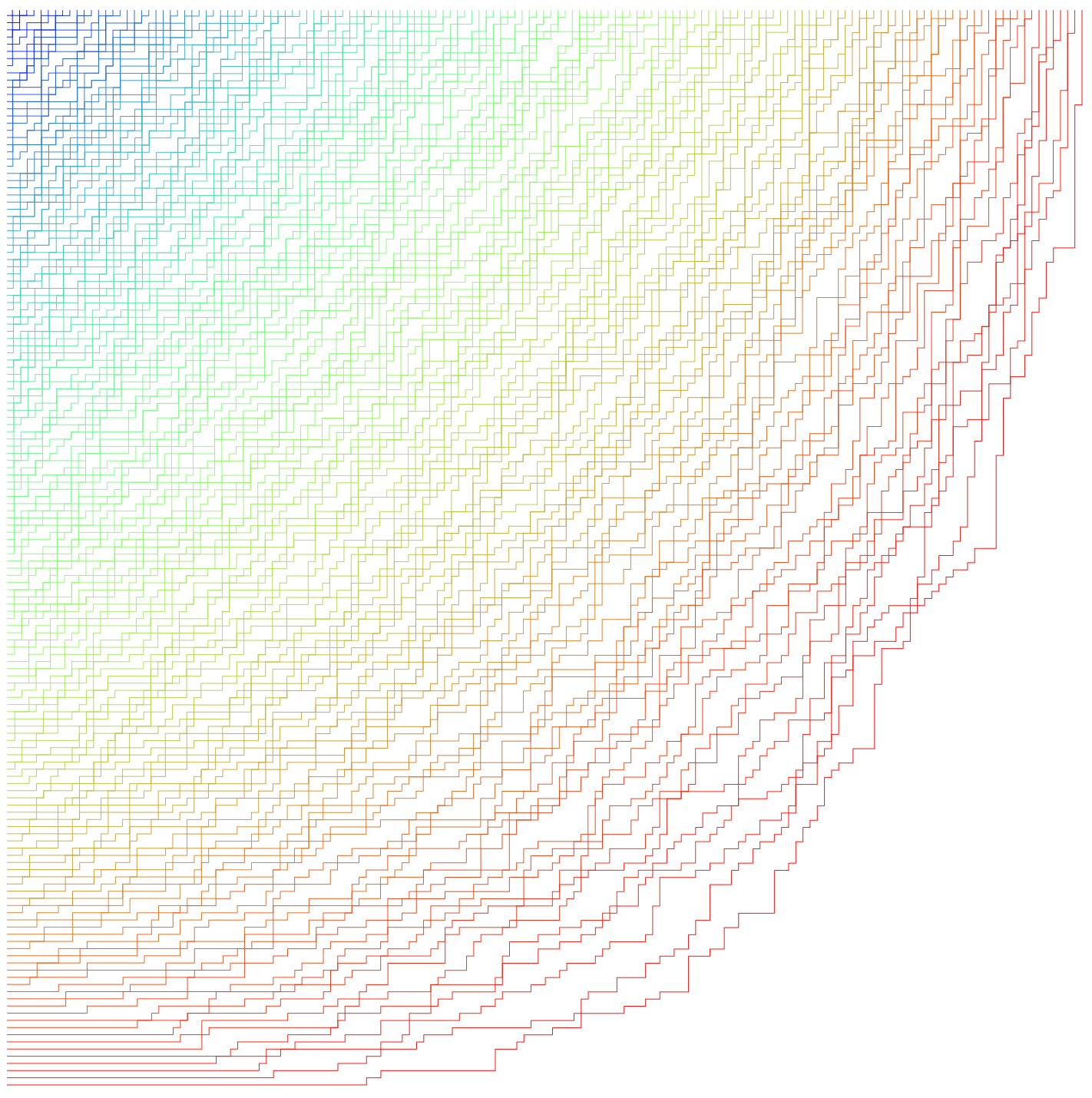}}}$
        & 
        $\vcenter{\hbox{\includegraphics[width=0.3\textwidth]{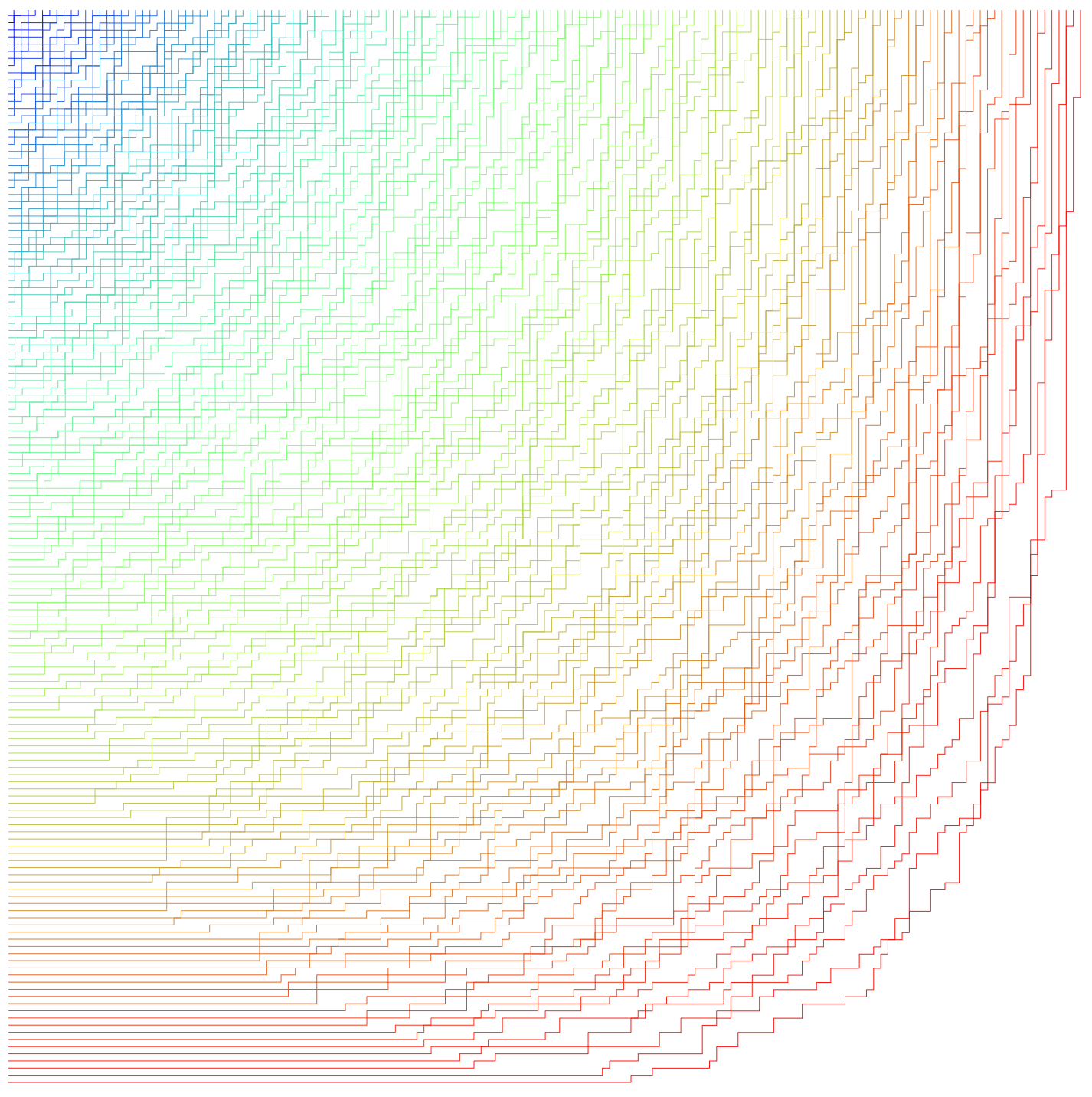}}}$
        \\
        $t=0.6$ & 
        $\vcenter{\hbox{\includegraphics[width=0.3\textwidth]{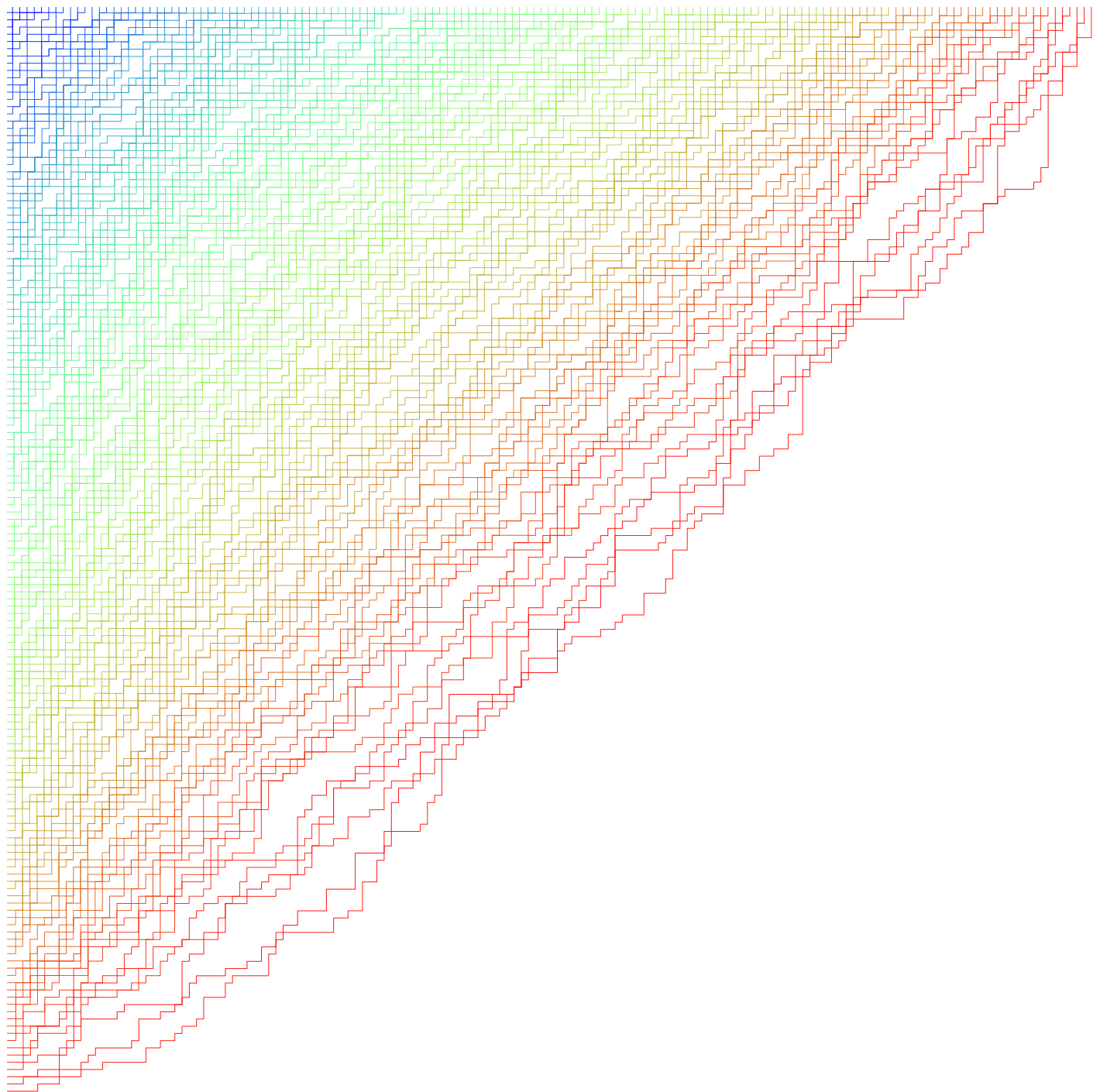}}}$
        & 
        $\vcenter{\hbox{\includegraphics[width=0.3\textwidth]{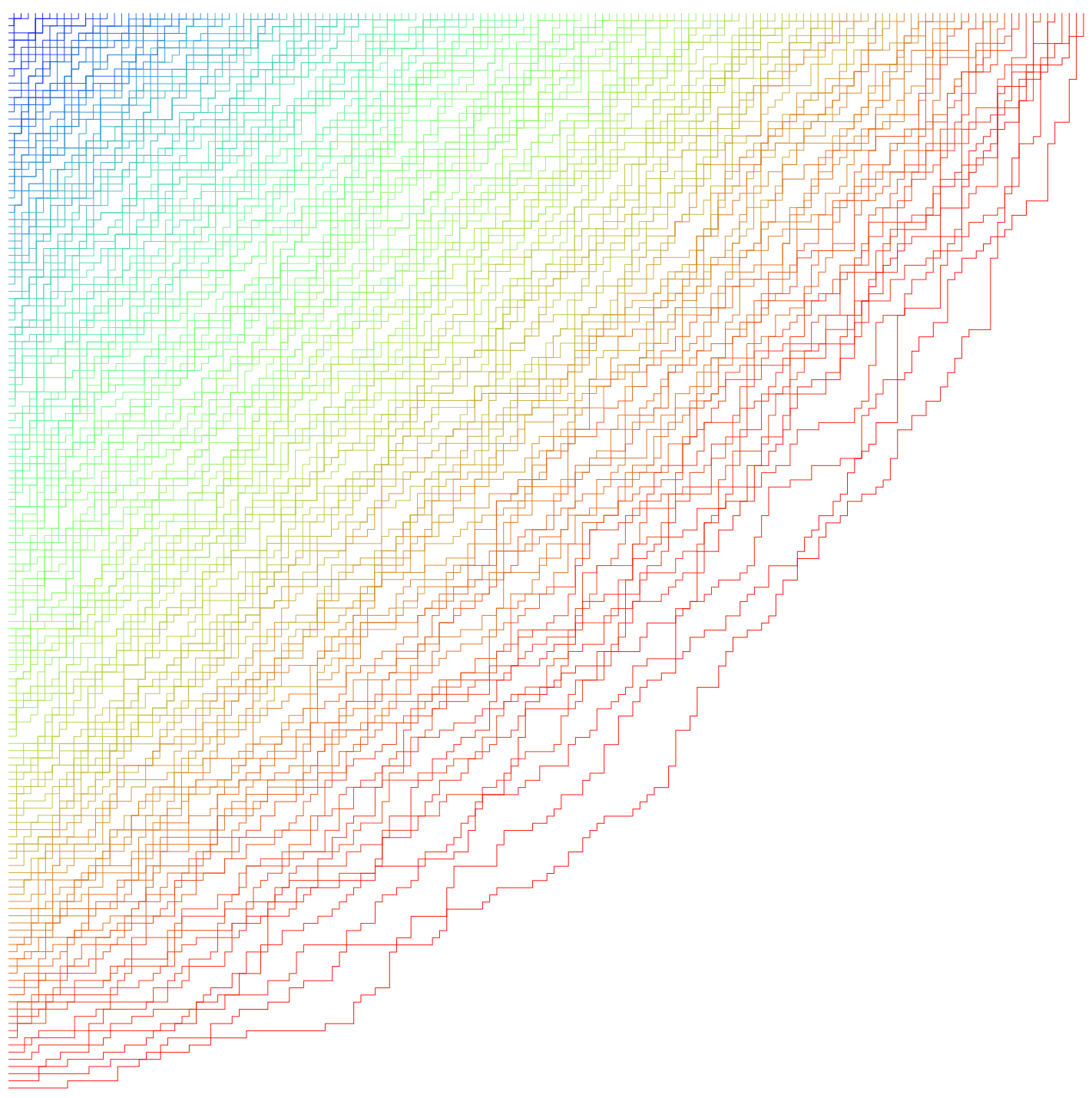}}}$
        & 
        $\vcenter{\hbox{\includegraphics[width=0.3\textwidth]{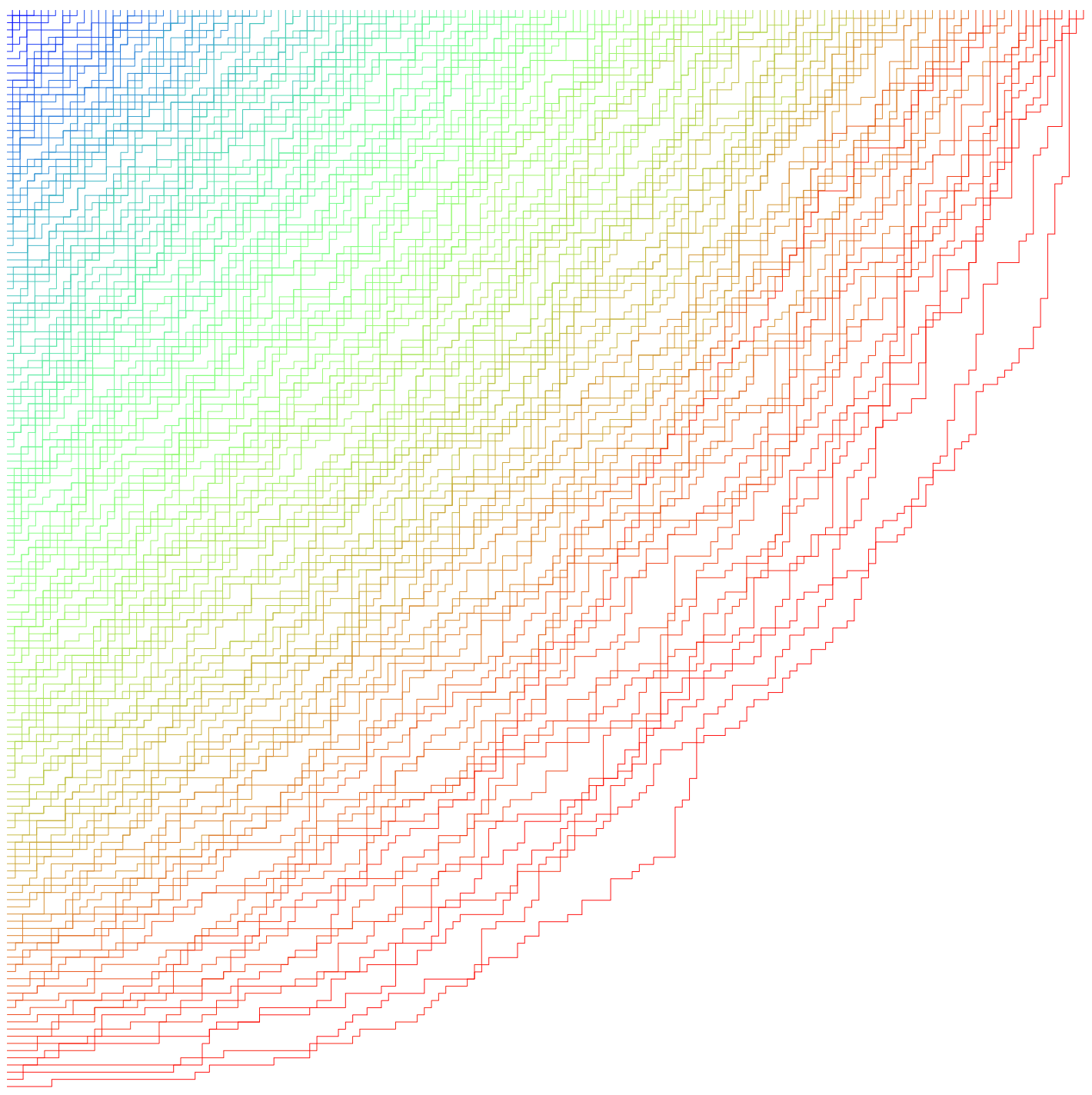}}}$
        & 
        $\vcenter{\hbox{\includegraphics[width=0.3\textwidth]{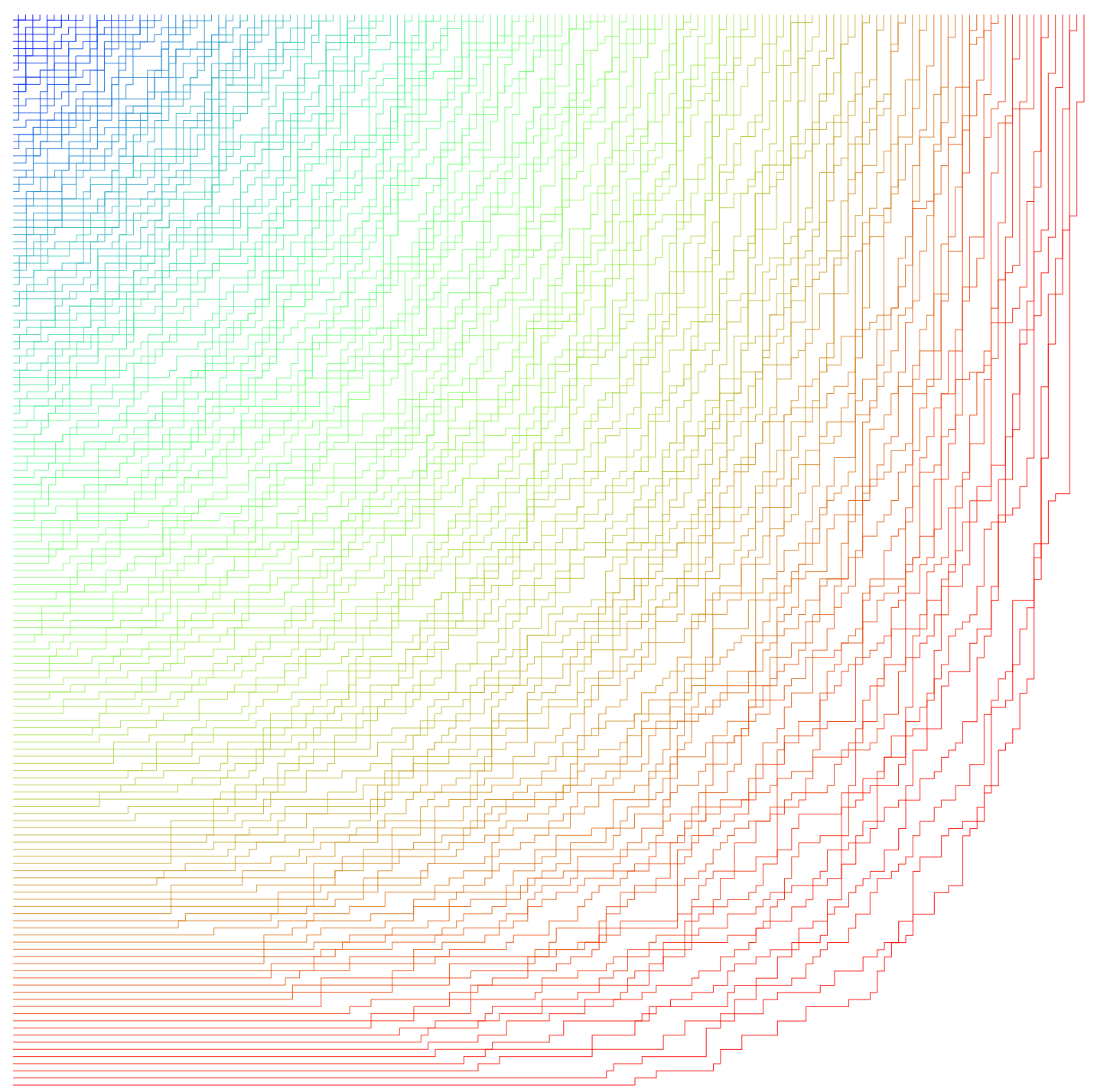}}}$
        \\
        $t=0.9$ & 
        $\vcenter{\hbox{\includegraphics[width=0.3\textwidth]{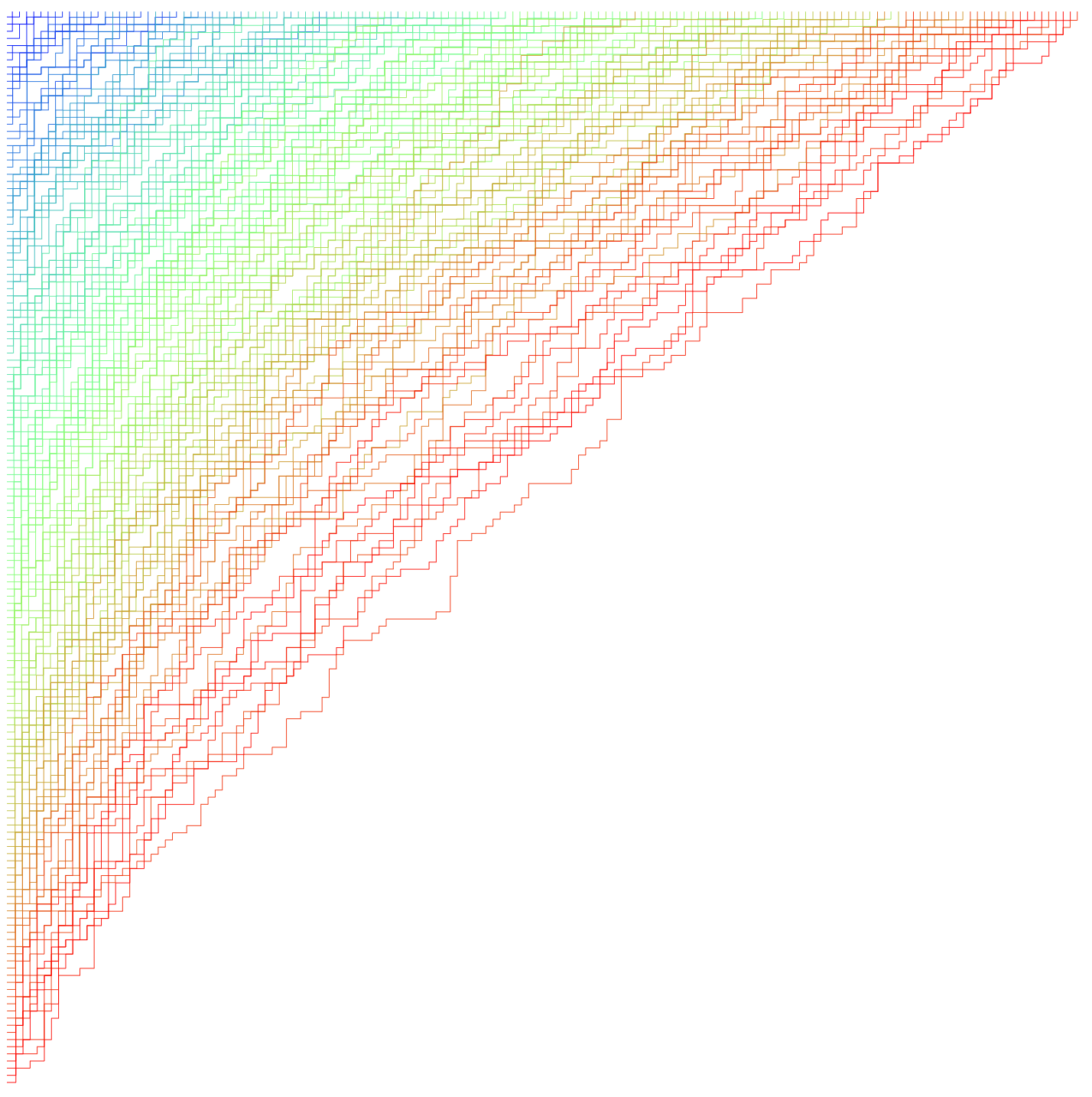}}}$
        & 
        $\vcenter{\hbox{\includegraphics[width=0.3\textwidth]{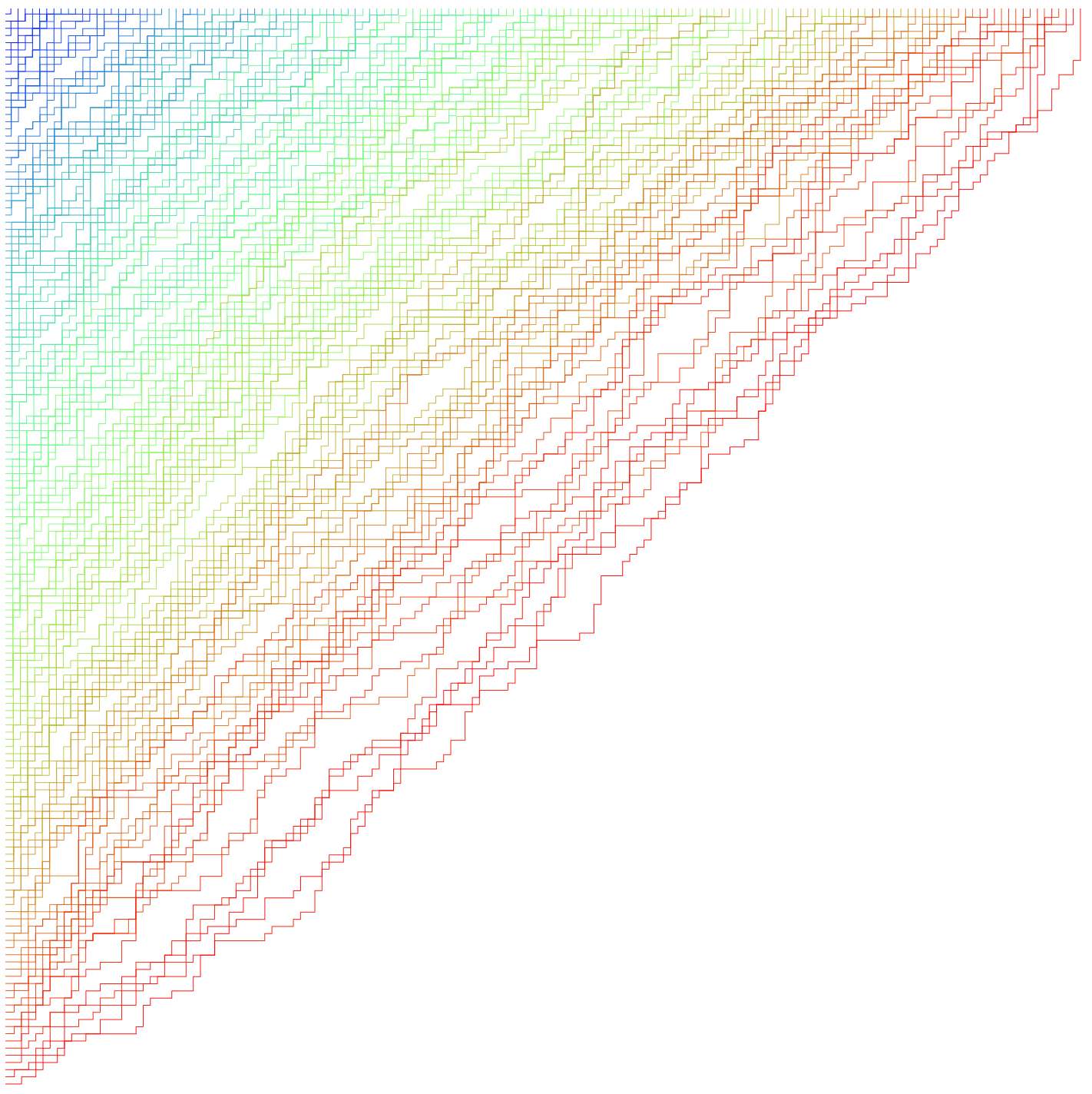}}}$
        & 
        $\vcenter{\hbox{\includegraphics[width=0.3\textwidth]{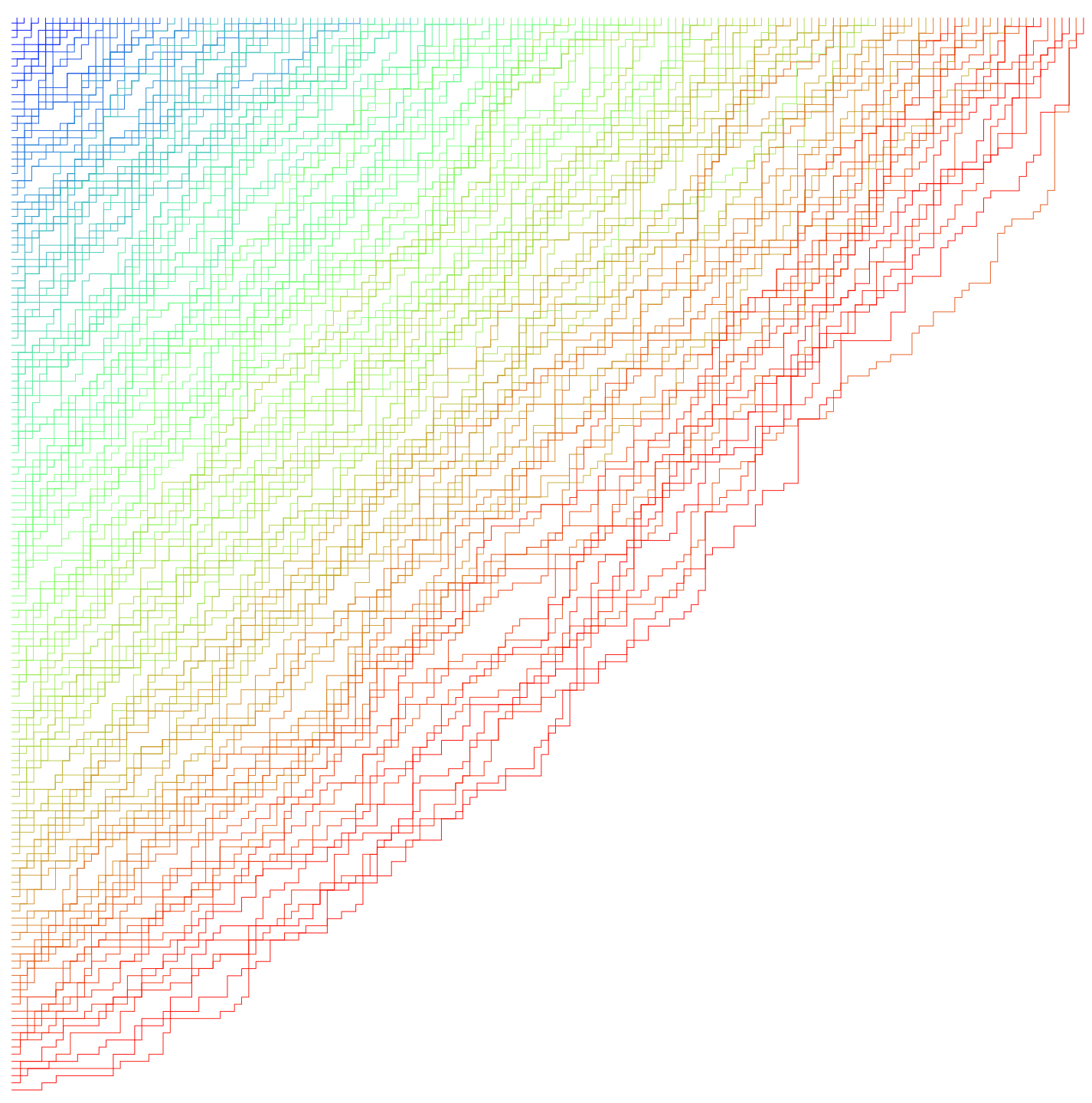}}}$
        & 
        $\vcenter{\hbox{\includegraphics[width=0.3\textwidth]{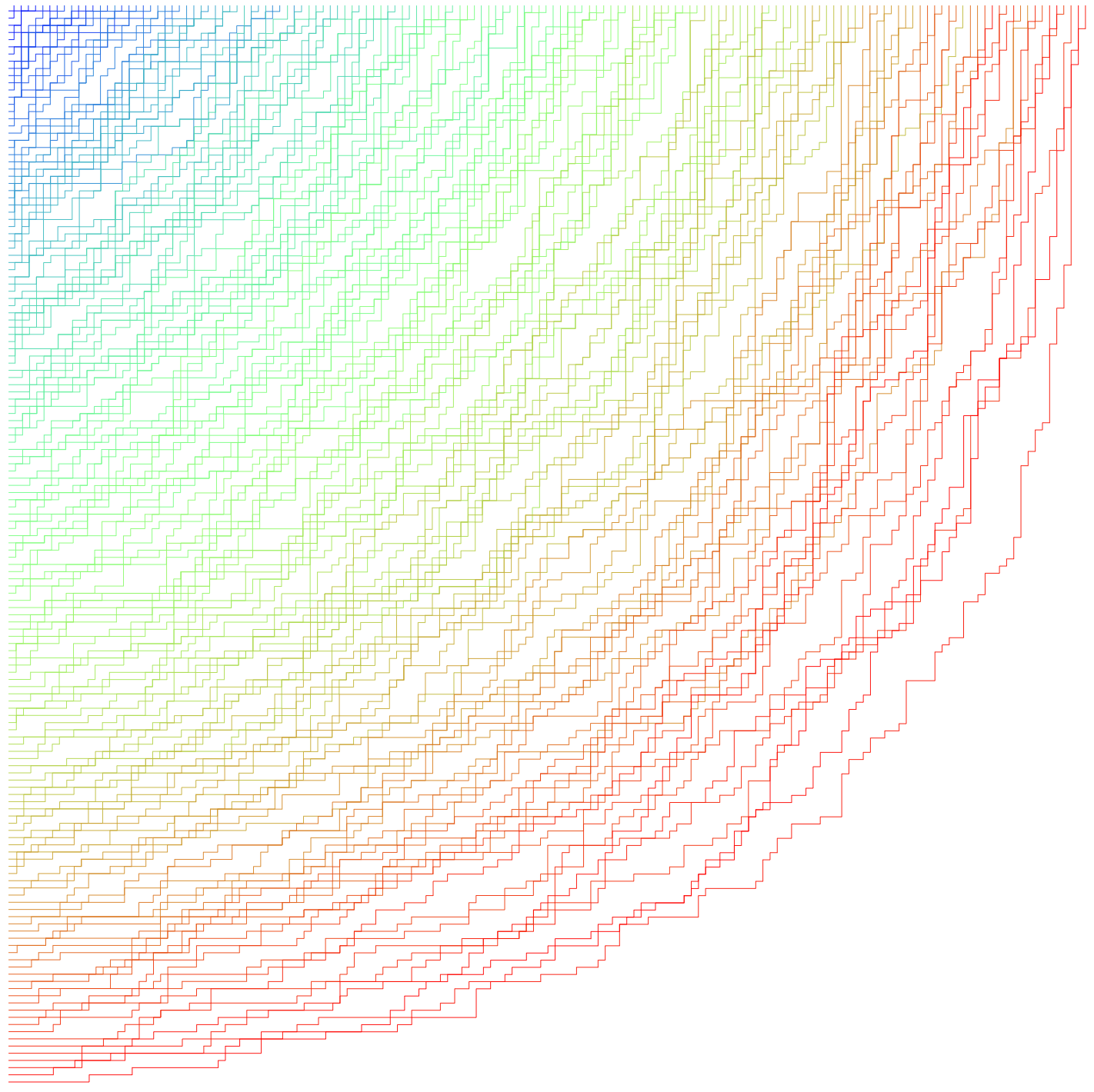}}}$
    \end{tabular}
    }
    \caption{Simulations of colored DWBC with $n=150$ and free-coloring on the top boundary for various values of $t$ and $q$. We set $x_1=\ldots=x_n=1$. When $t=0$ we superimpose the theoretically predicted arctic curve.}
    \label{fig:DWBCsim}
\end{figure}

\begin{figure}
    \centering
    \resizebox{\textwidth}{!}{
    \begin{tabular}{ccccc}
        & $q=0.2$ & $q=0.5$ & $q=1.0$ & $q=5$ \\
        $t=0$ & 
        $\vcenter{\hbox{\includegraphics[width=0.3\textwidth]{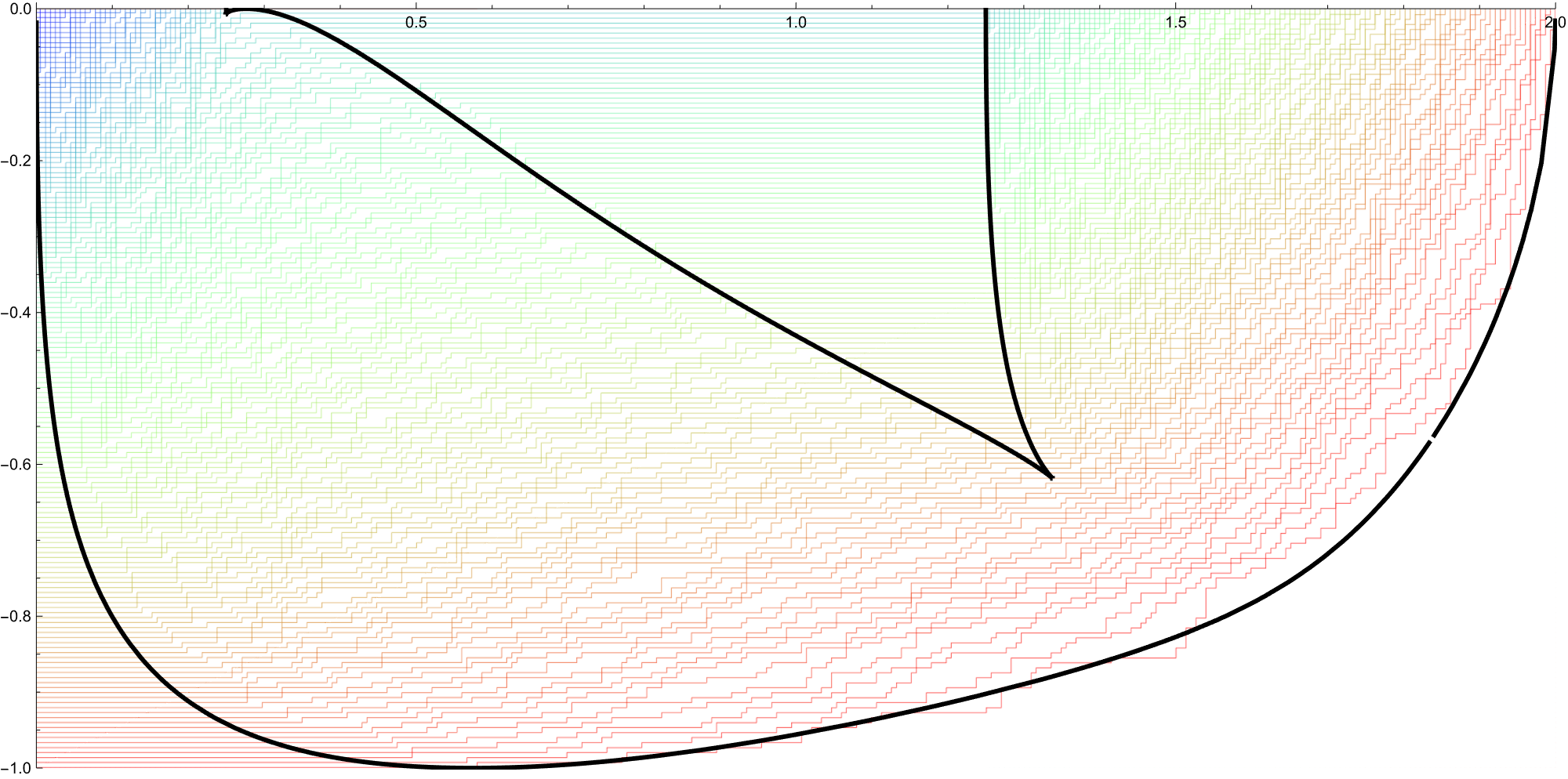}}}$
        & 
        $\vcenter{\hbox{\includegraphics[width=0.3\textwidth]{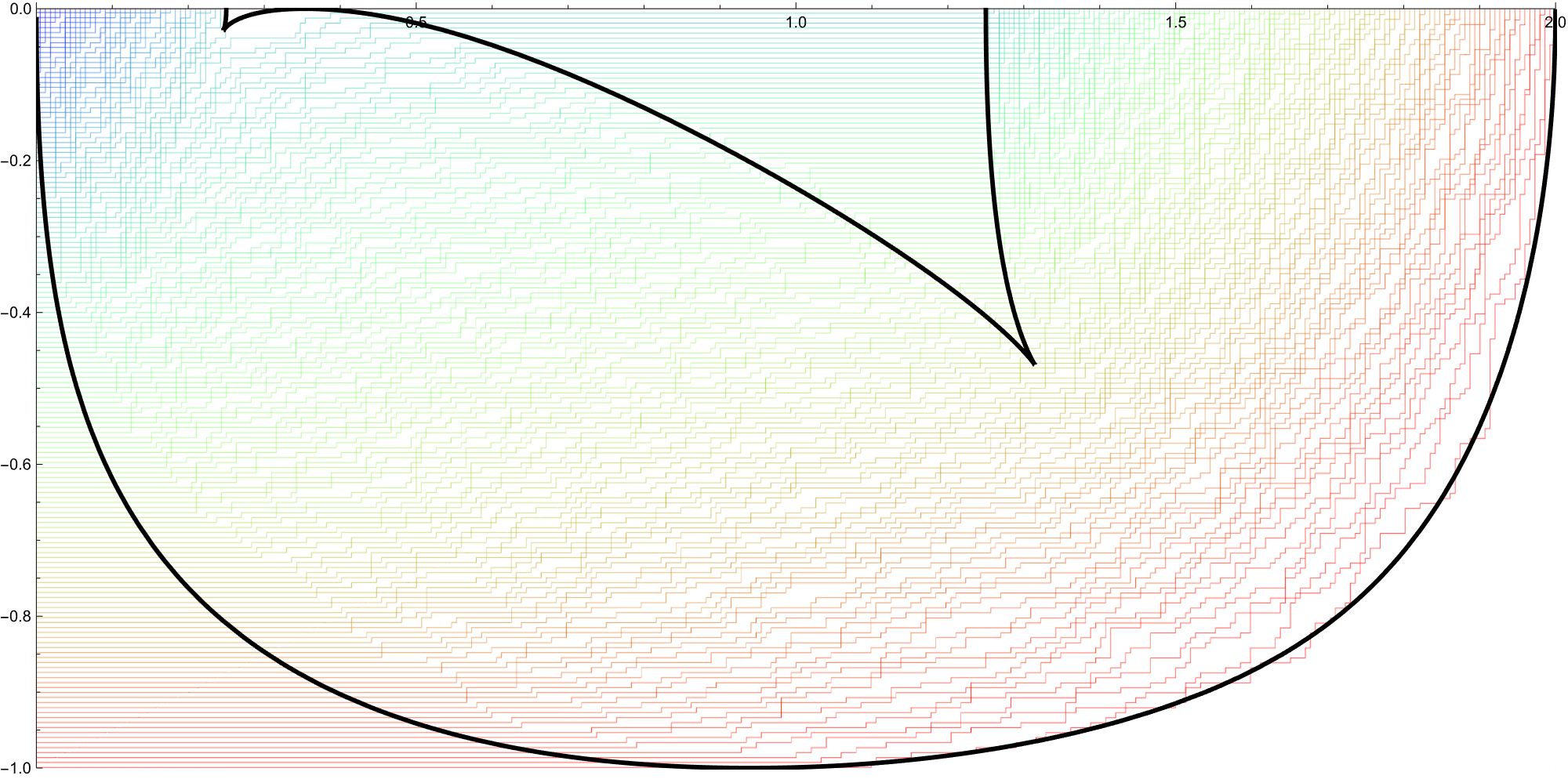}}}$
        & 
        $\vcenter{\hbox{\includegraphics[width=0.3\textwidth]{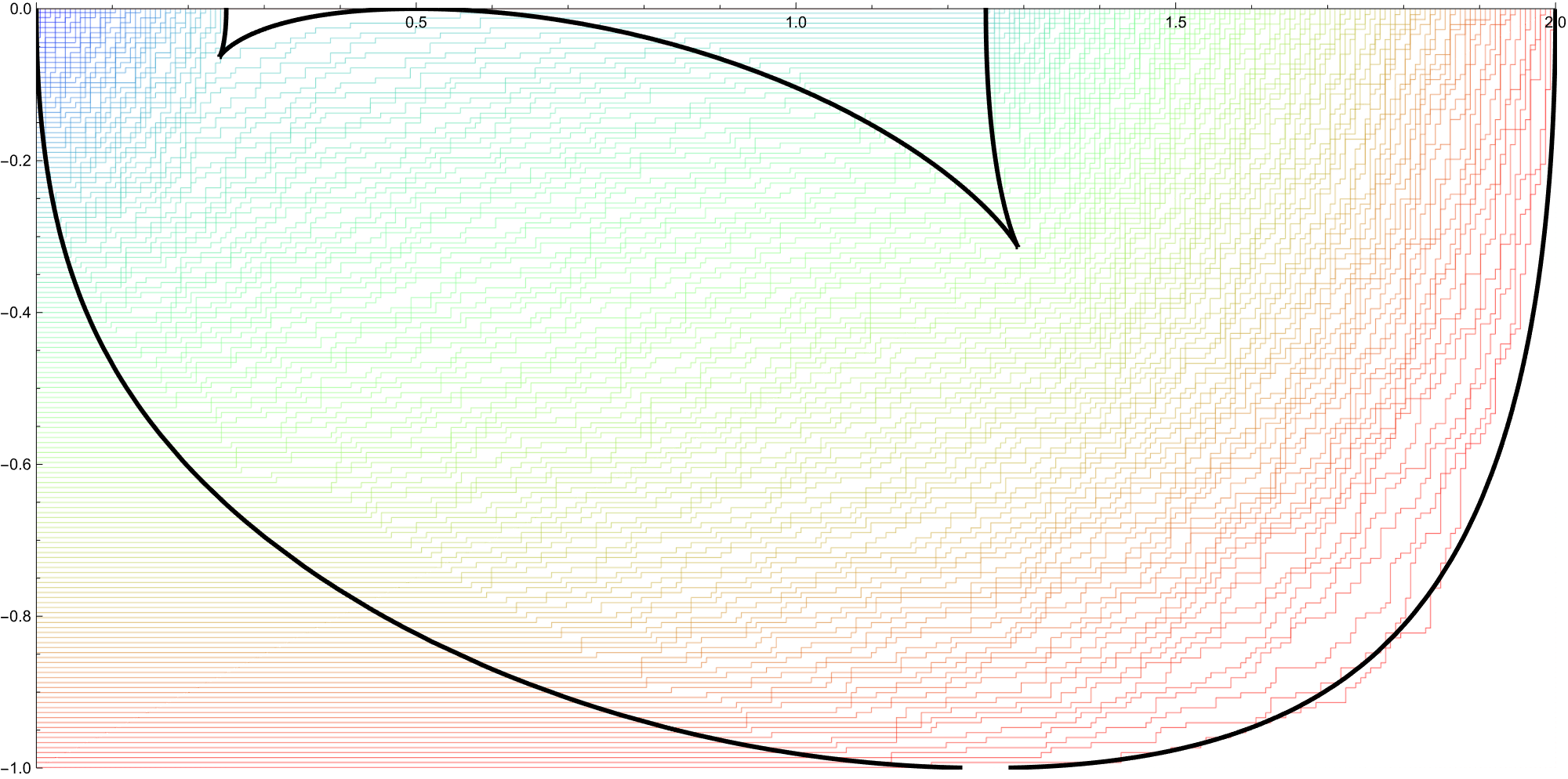}}}$
        & 
        $\vcenter{\hbox{\includegraphics[width=0.3\textwidth]{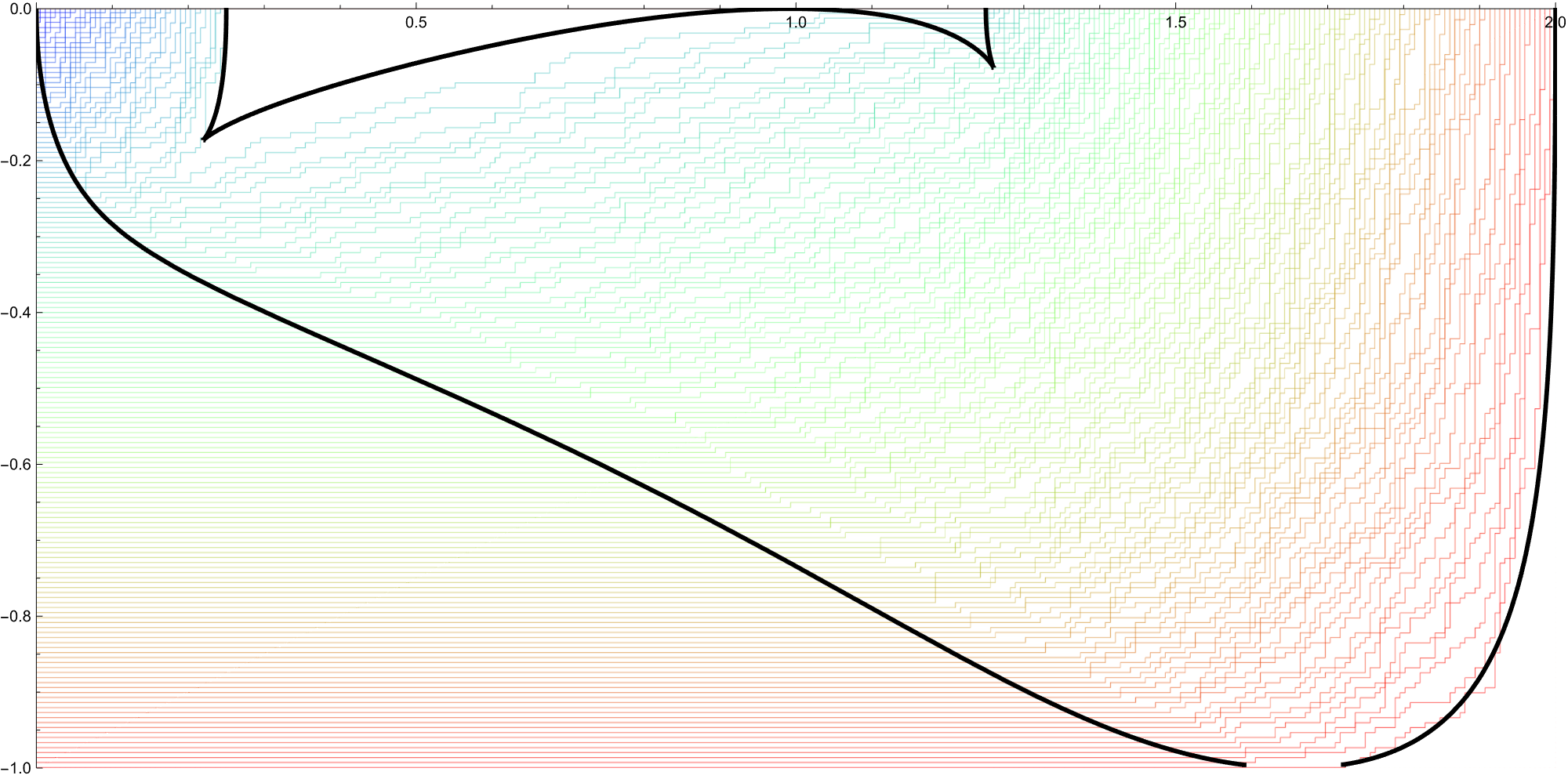}}}$
        \\
        $t=0.3$ & 
        $\vcenter{\hbox{\includegraphics[width=0.3\textwidth]{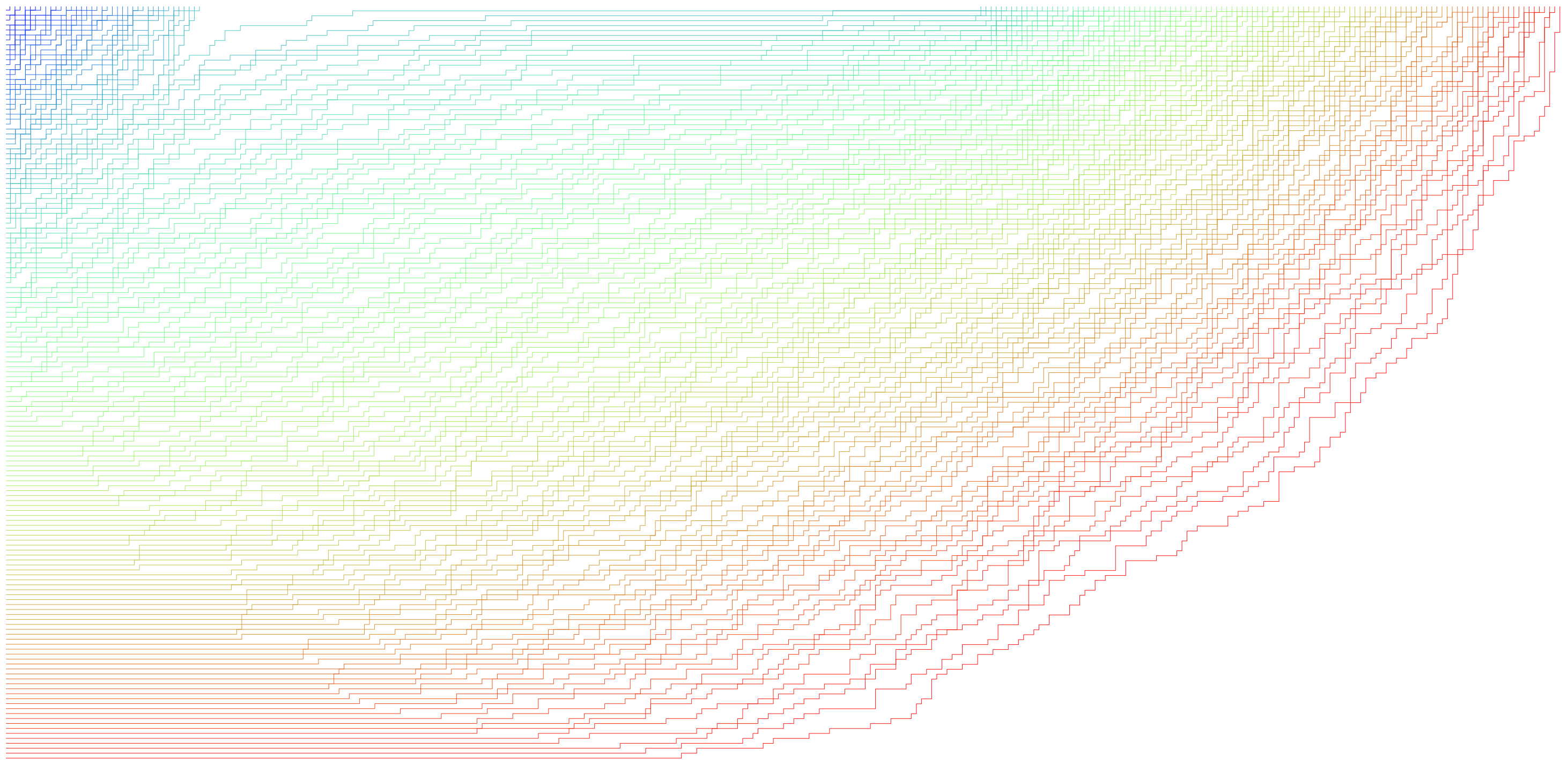}}}$
        & 
        $\vcenter{\hbox{\includegraphics[width=0.3\textwidth]{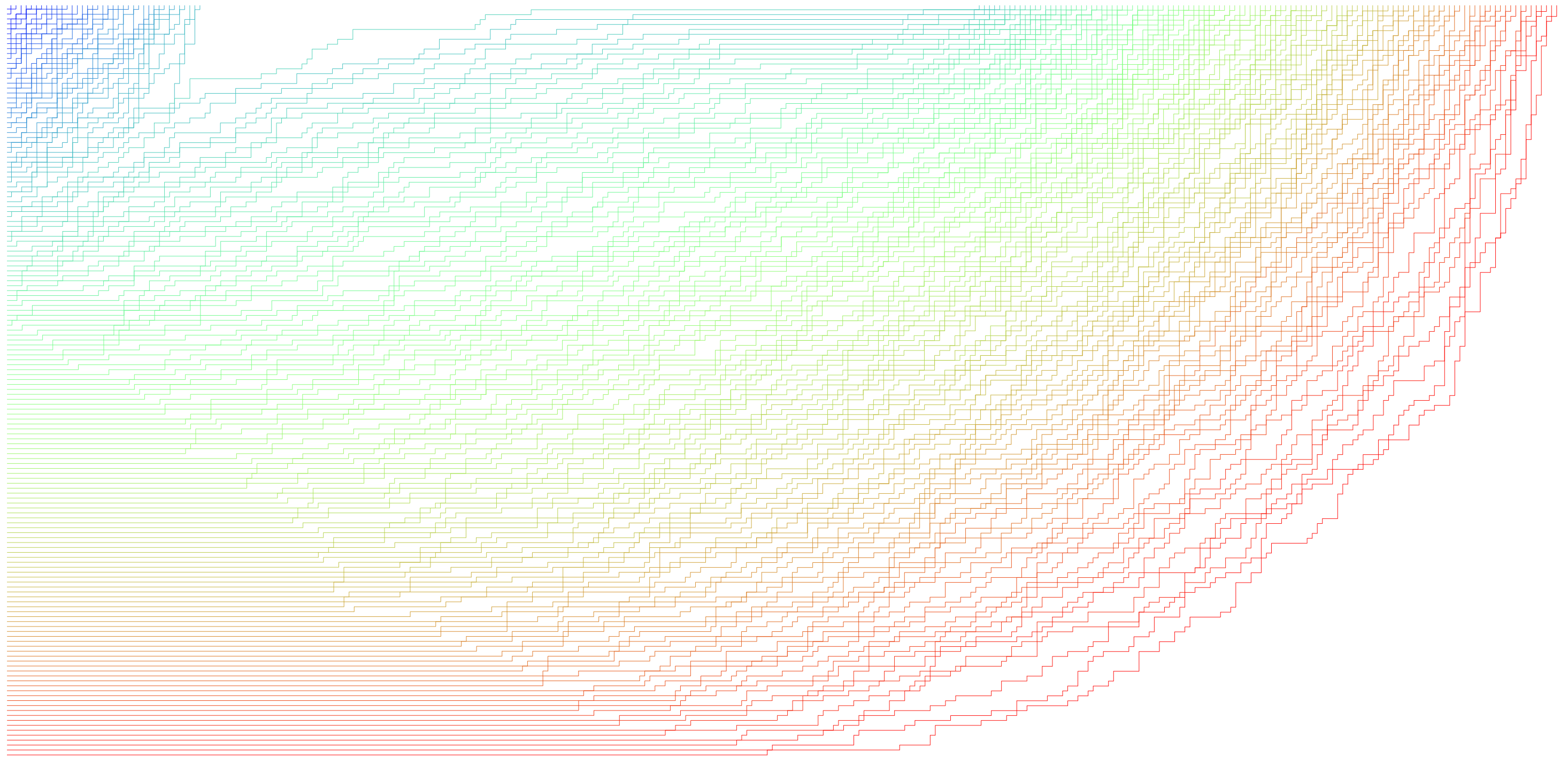}}}$
        & 
        $\vcenter{\hbox{\includegraphics[width=0.3\textwidth]{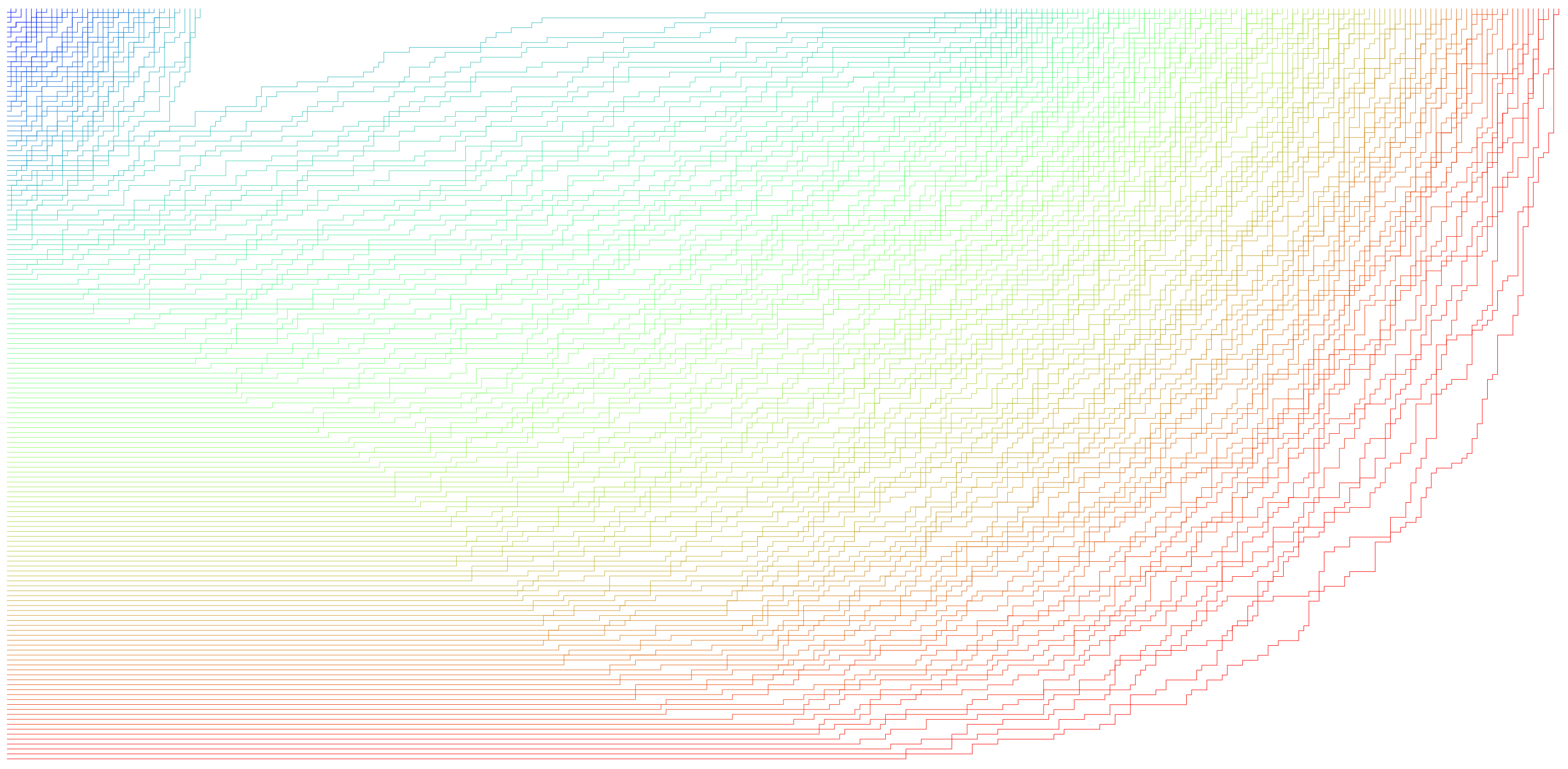}}}$
        & 
        $\vcenter{\hbox{\includegraphics[width=0.3\textwidth]{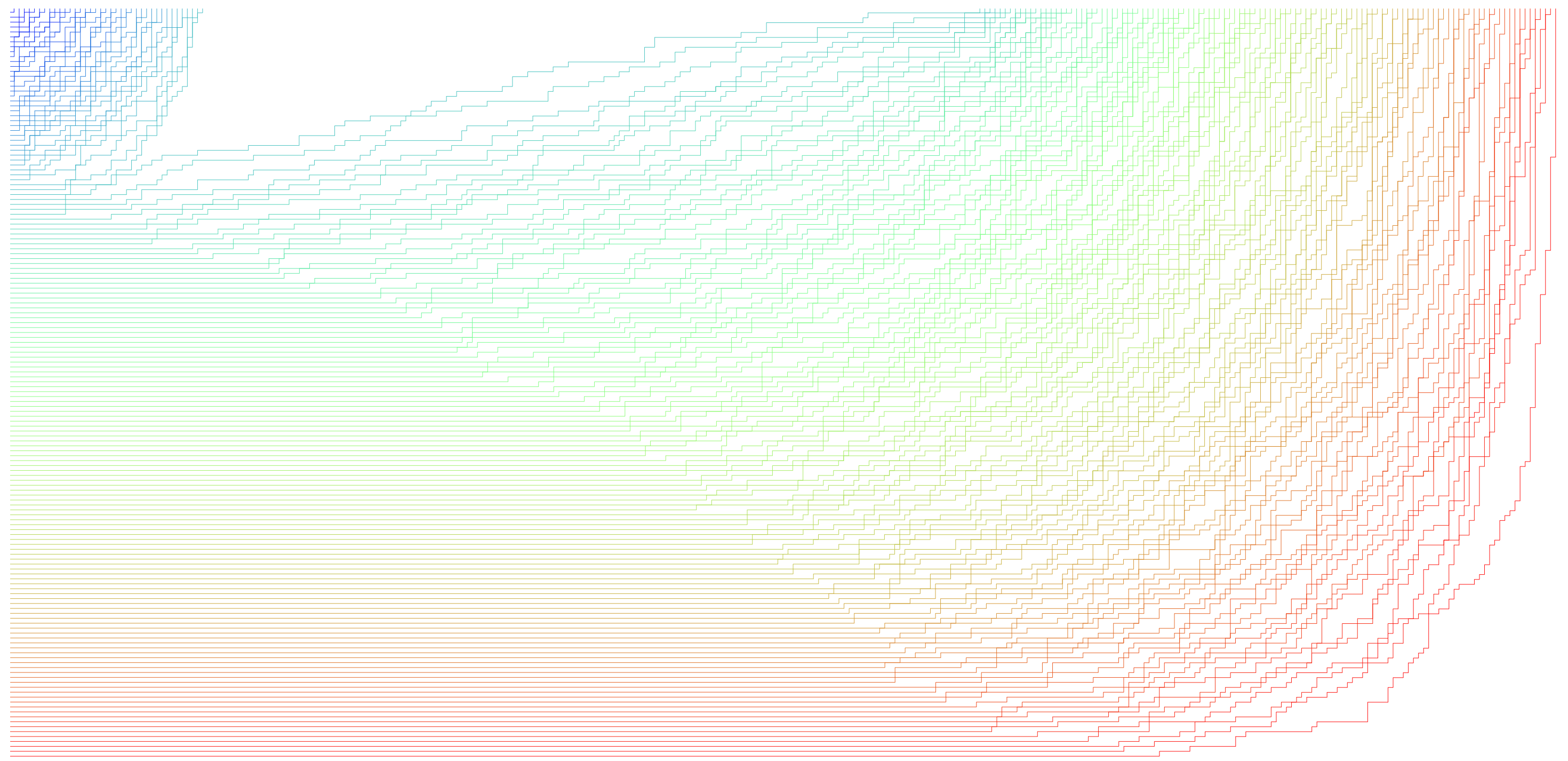}}}$
        \\
        $t=0.6$ & 
        $\vcenter{\hbox{\includegraphics[width=0.3\textwidth]{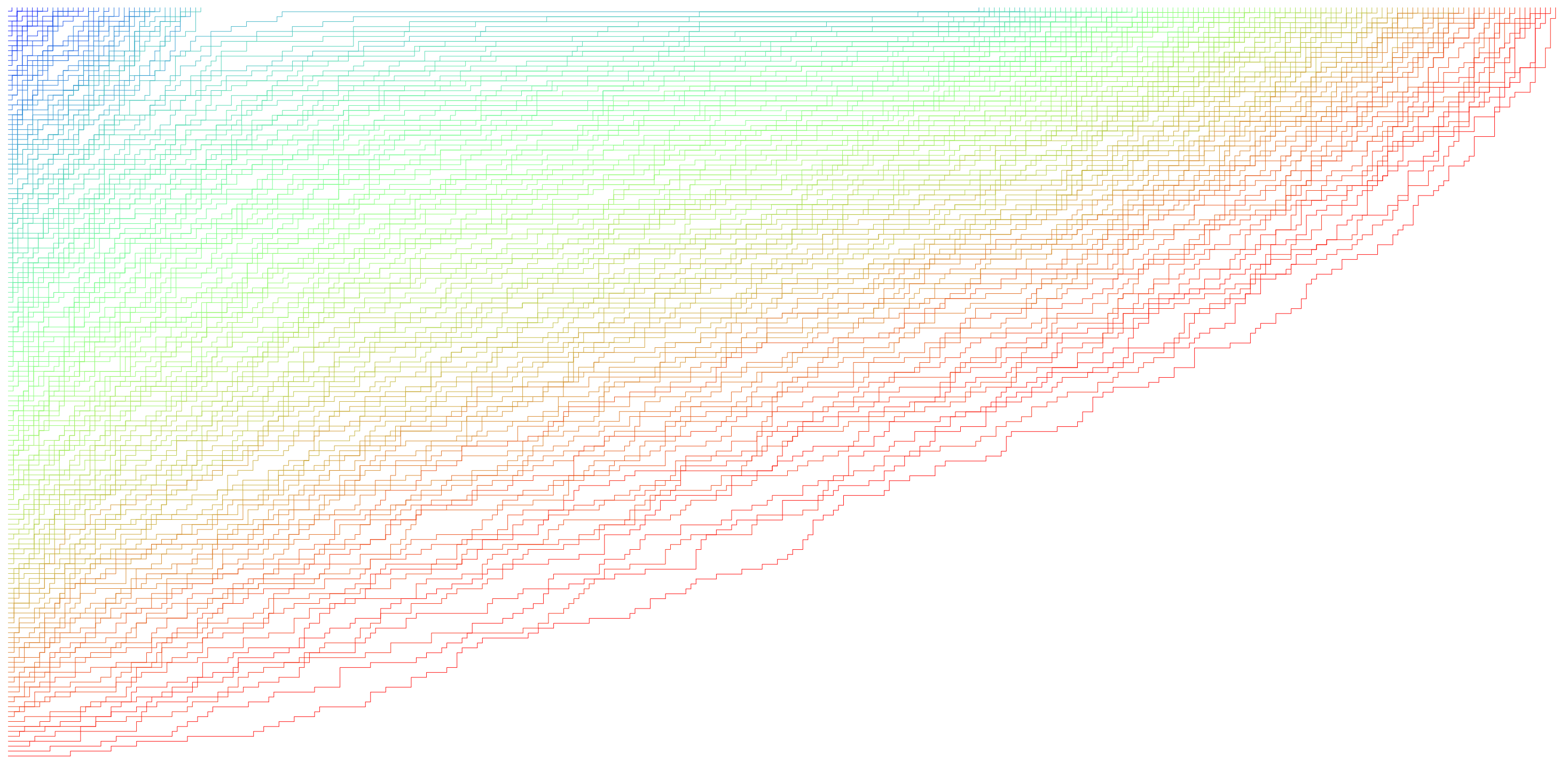}}}$
        & 
        $\vcenter{\hbox{\includegraphics[width=0.3\textwidth]{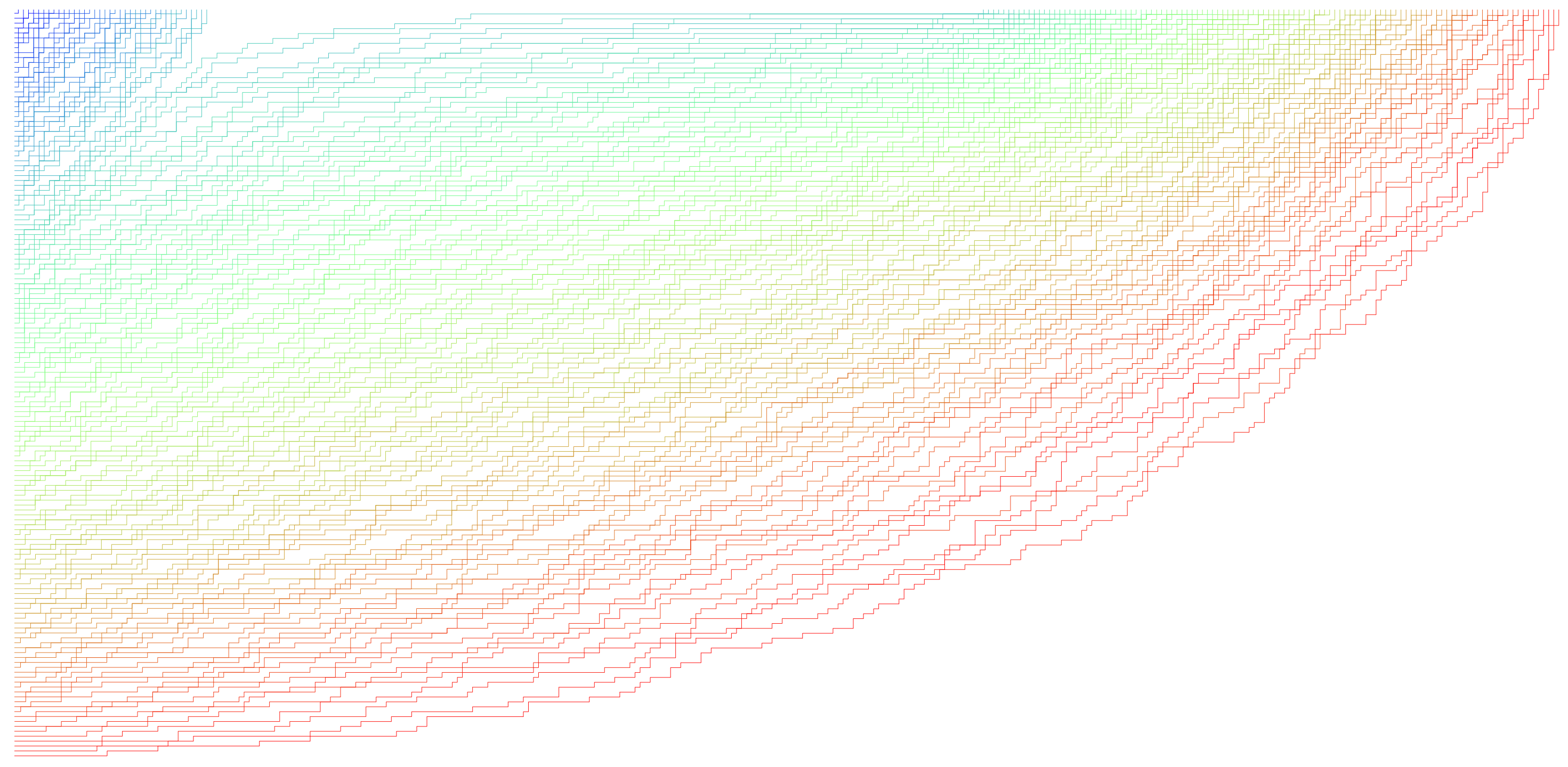}}}$
        & 
        $\vcenter{\hbox{\includegraphics[width=0.3\textwidth]{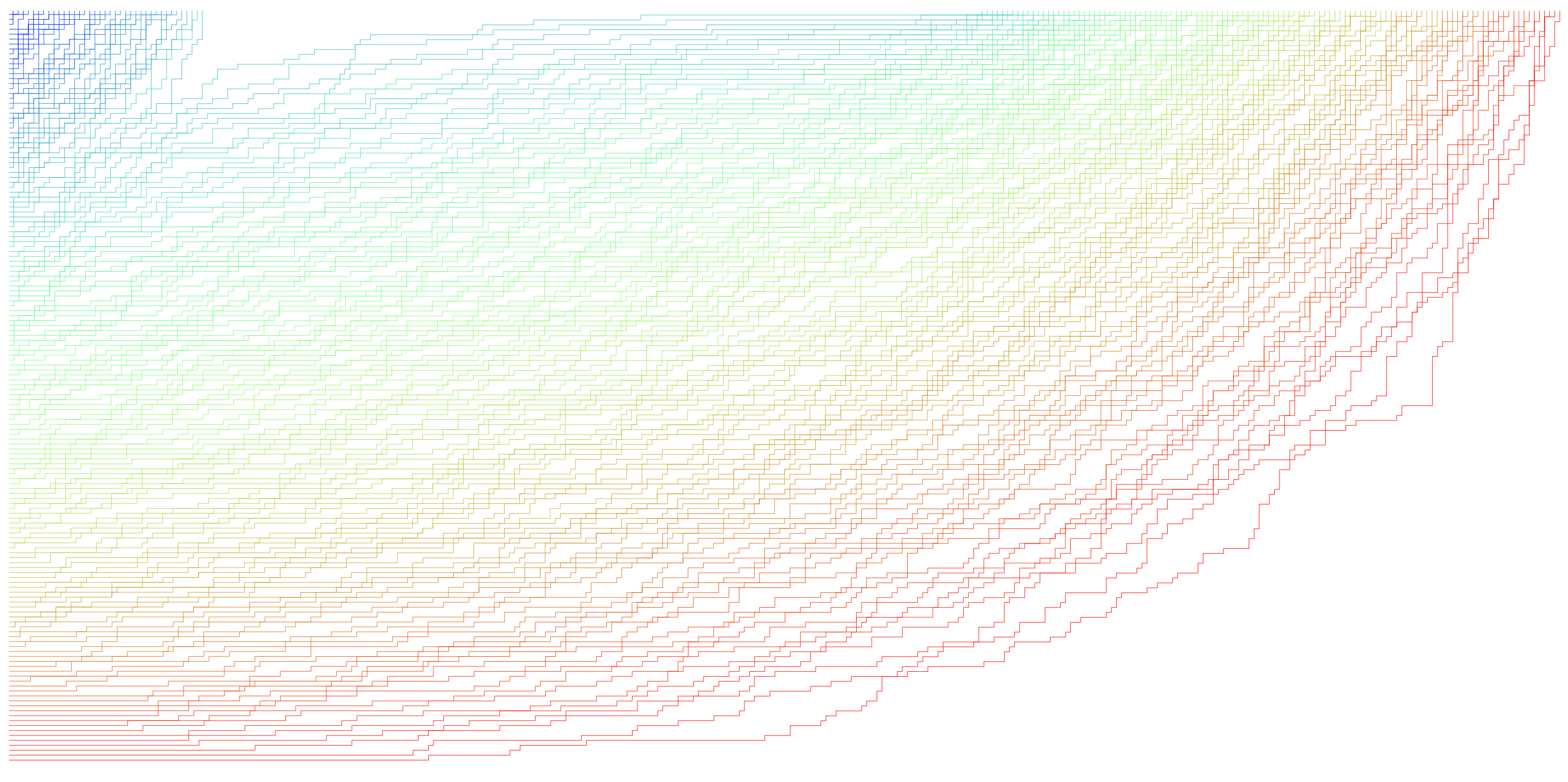}}}$
        & 
        $\vcenter{\hbox{\includegraphics[width=0.3\textwidth]{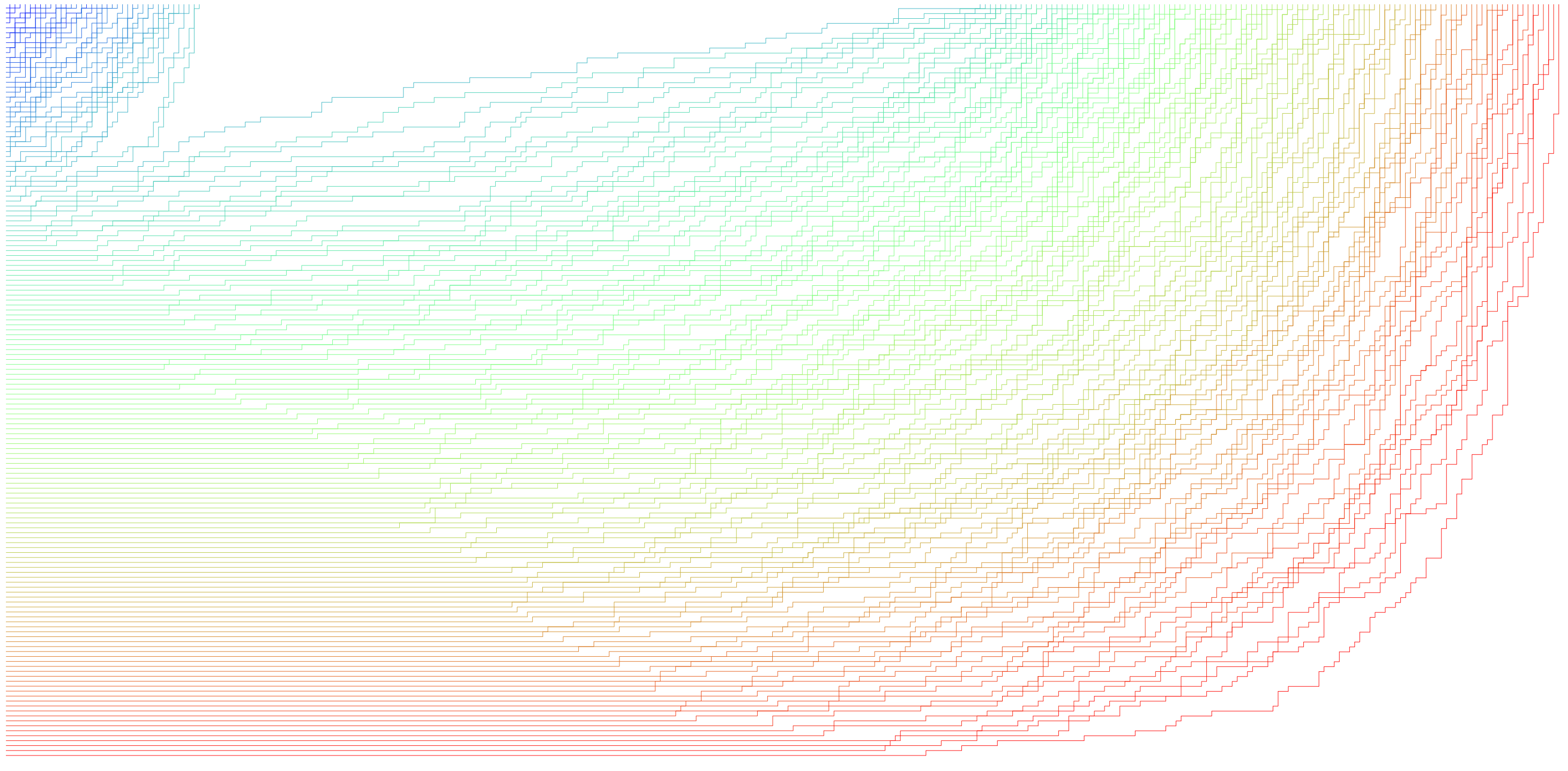}}}$
        \\
        $t=0.9$ & 
        $\vcenter{\hbox{\includegraphics[width=0.3\textwidth]{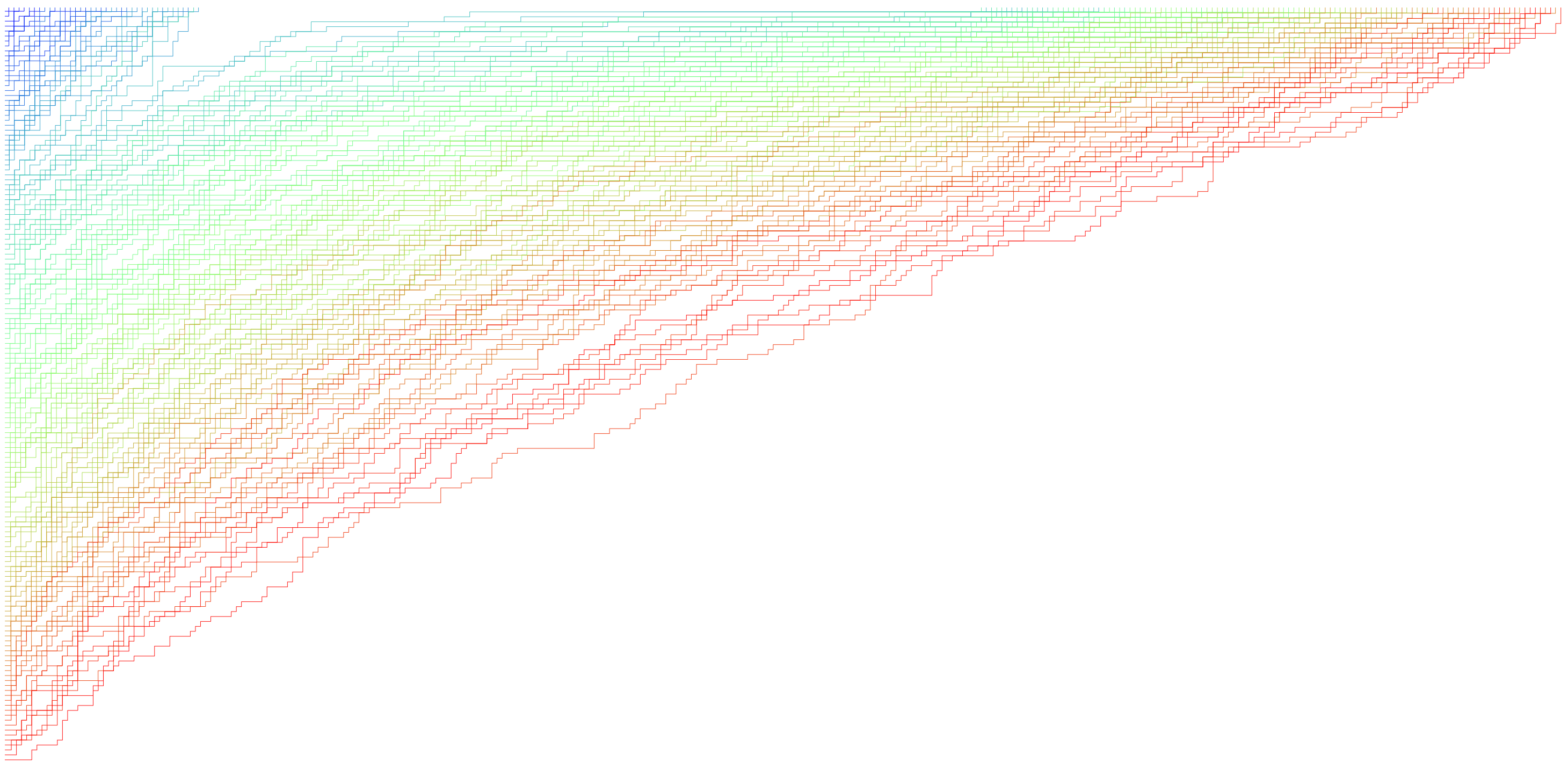}}}$
        & 
        $\vcenter{\hbox{\includegraphics[width=0.3\textwidth]{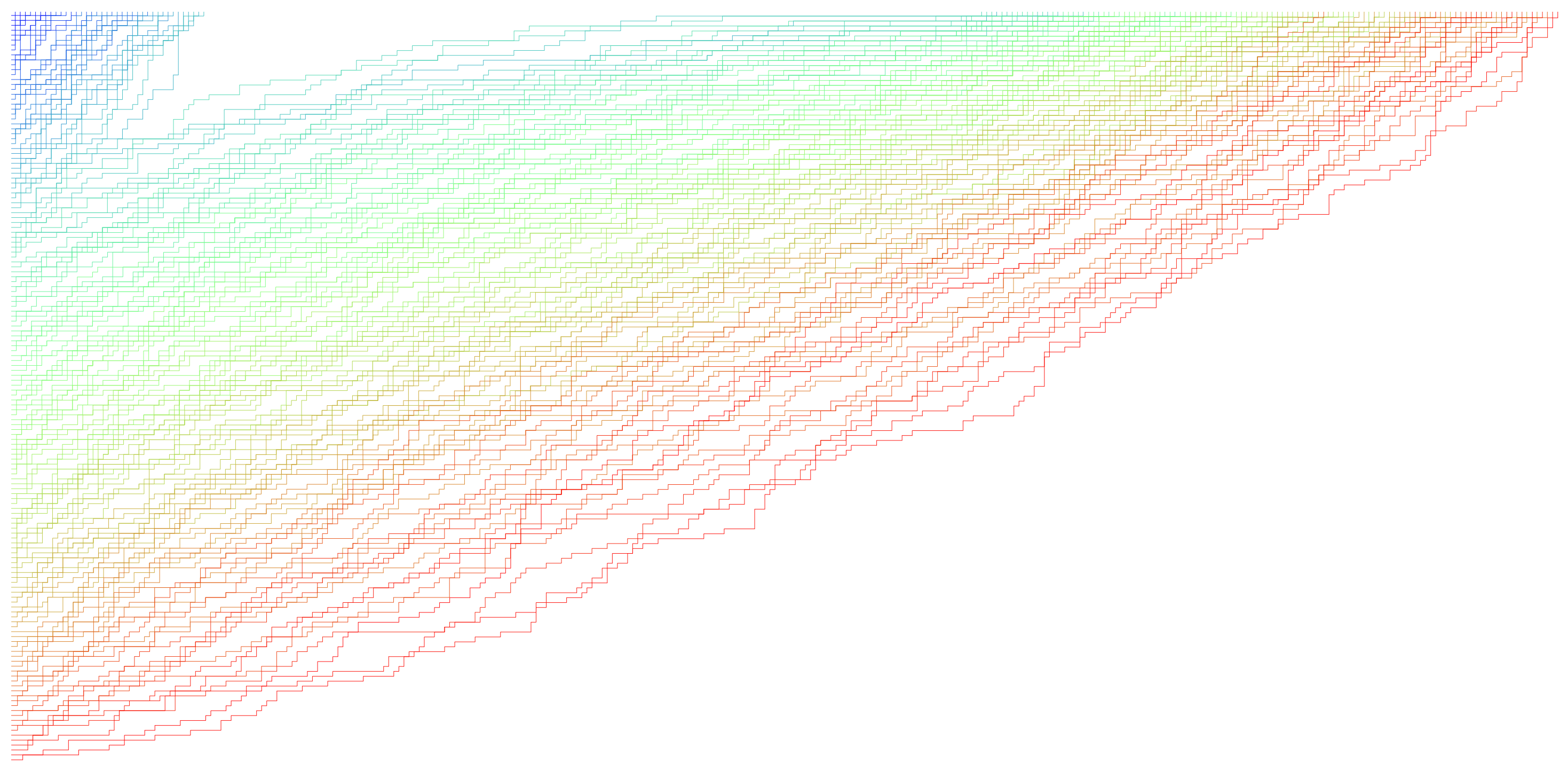}}}$
        & 
        $\vcenter{\hbox{\includegraphics[width=0.3\textwidth]{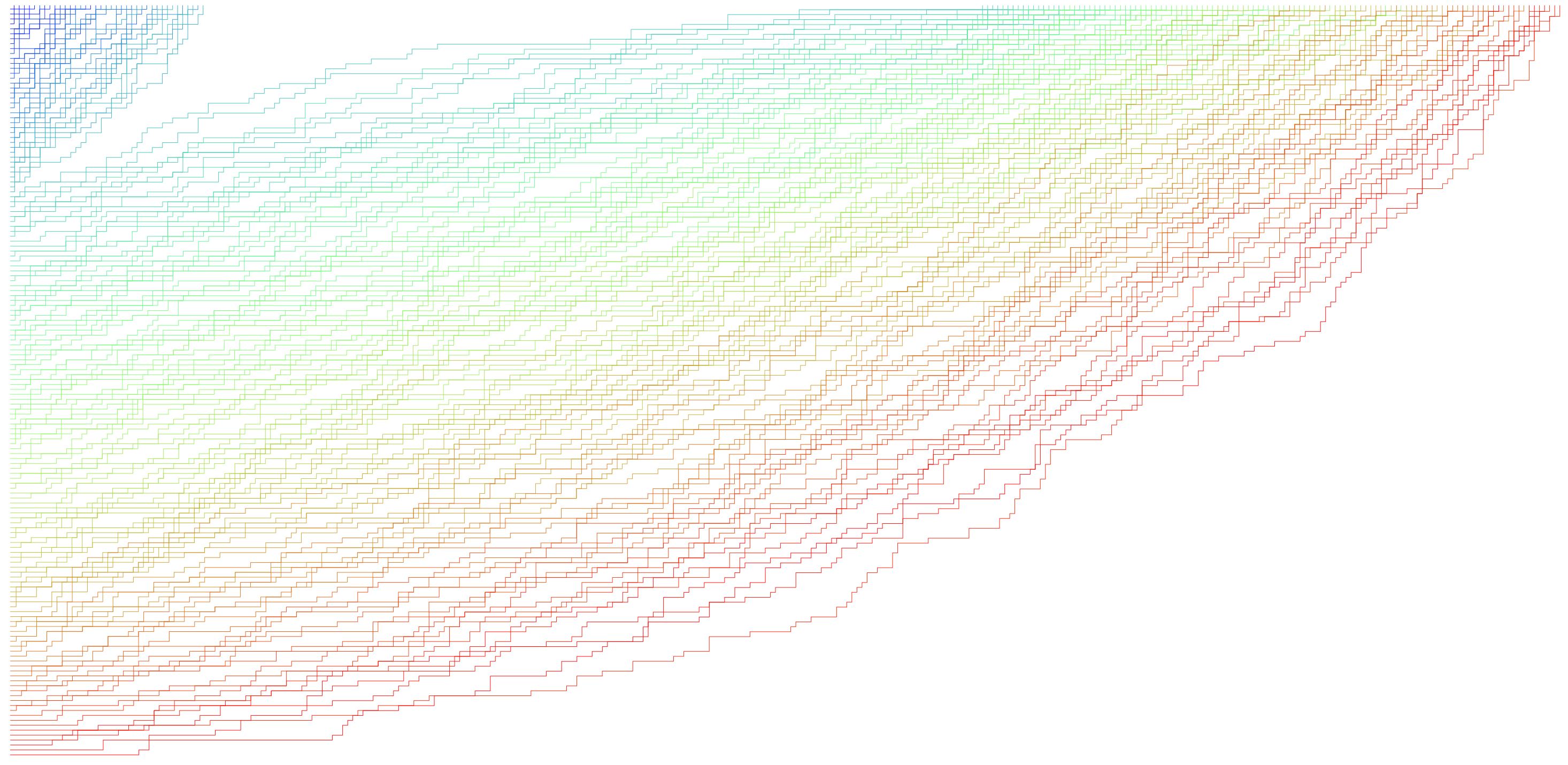}}}$
        & 
        $\vcenter{\hbox{\includegraphics[width=0.3\textwidth]{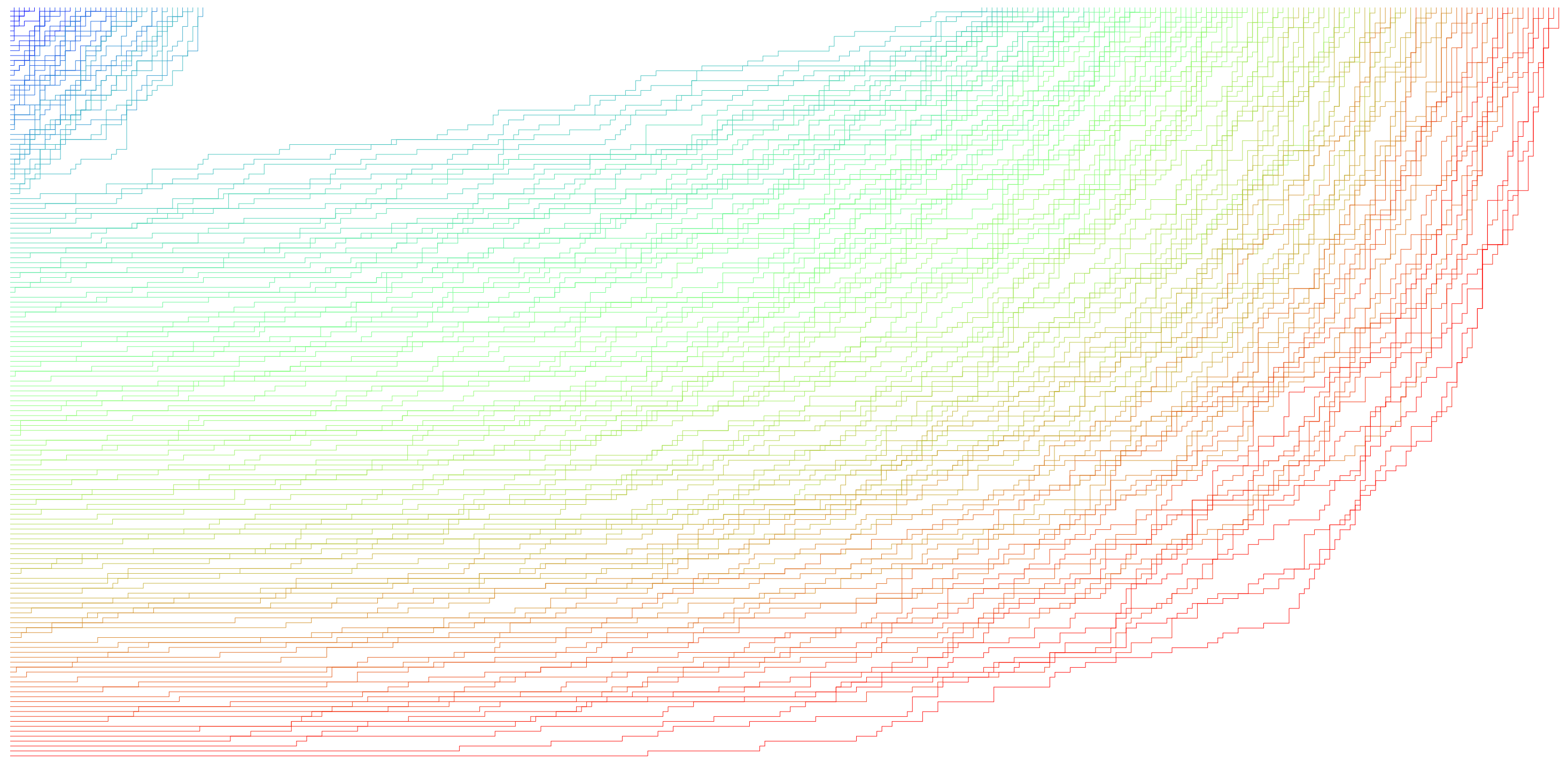}}}$
    \end{tabular}
    }
    \caption{Simulations where the boundary contains a single gap with $n=150$ and free-coloring on the top boundary for various values of $t$ and $q$. We set $x_1=\ldots=x_n=1$. The gap corresponds to $\kappa=\frac{1}{4}$ and $\mu=1$ as in Section \ref{sec:onegap}. When $t=0$ we superimpose the theoretically predicted arctic curve.}
    \label{fig:Gapsim}
\end{figure}

\begin{figure}
    \centering
    \resizebox{\textwidth}{!}{
    \begin{tabular}{ccc}
         $\vcenter{\hbox{\includegraphics[height=0.3\textwidth]{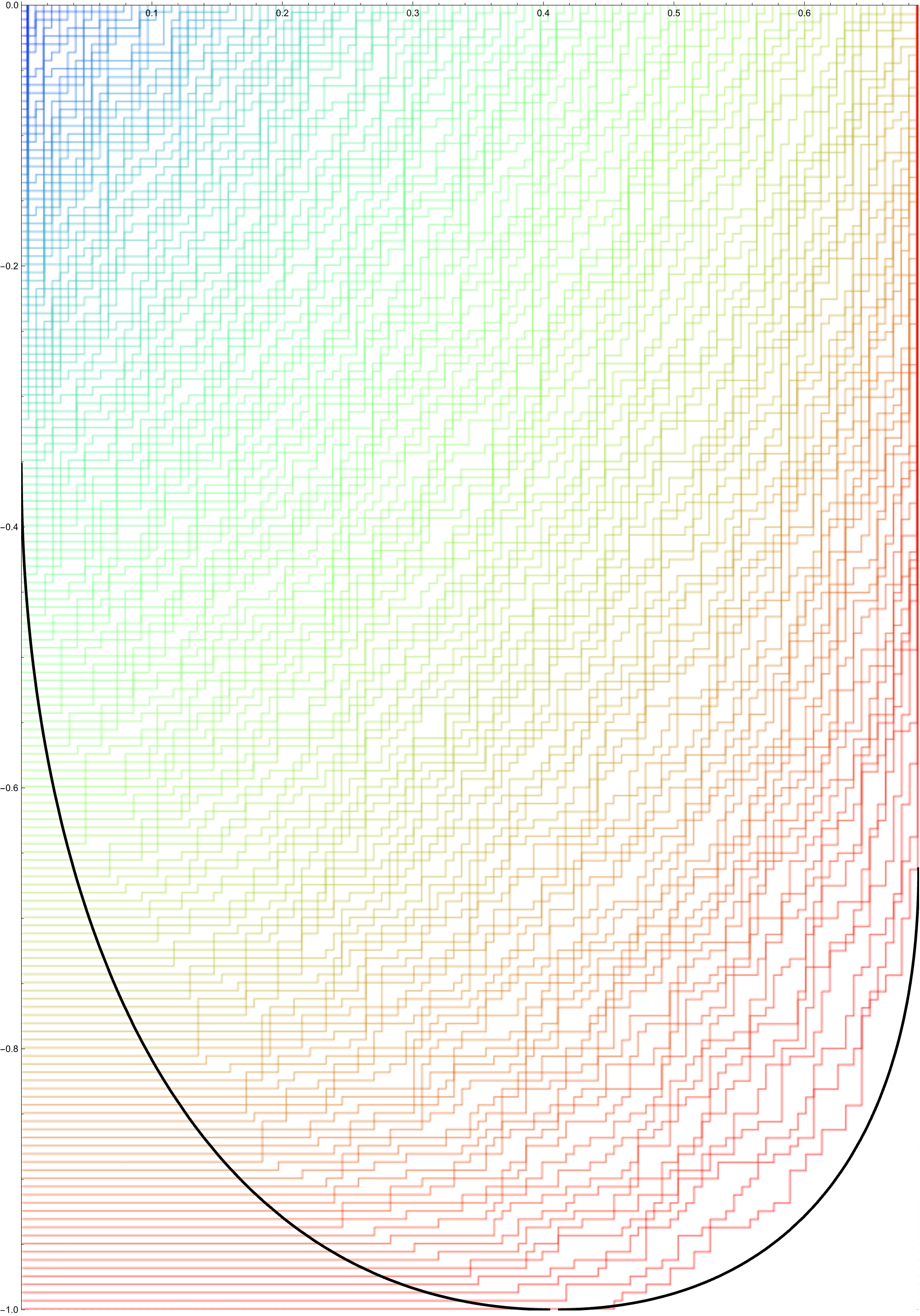}}}$
         &
         $\vcenter{\hbox{\includegraphics[height=0.3\textwidth]{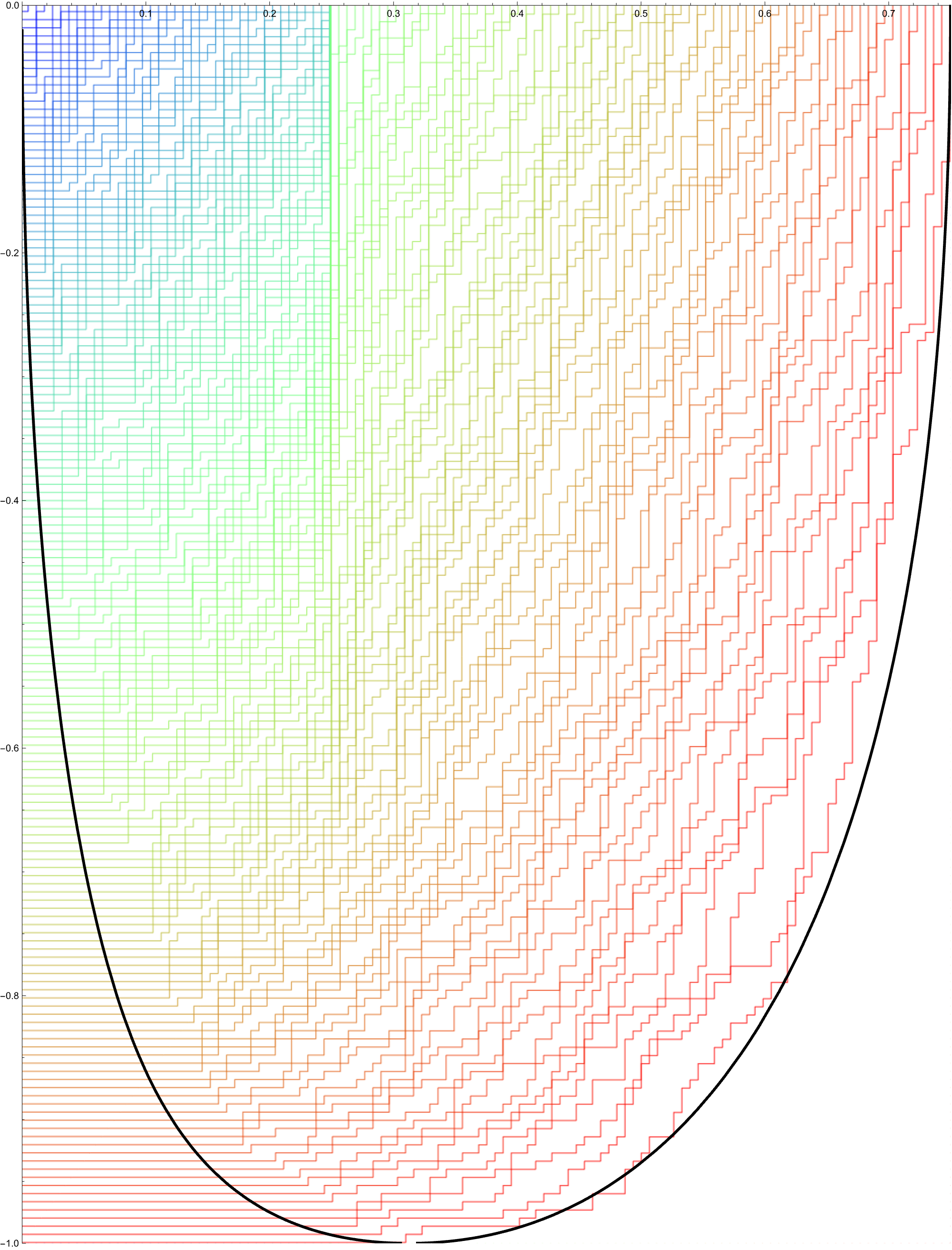}}}$
         &
         $\vcenter{\hbox{\includegraphics[height=0.3\textwidth]{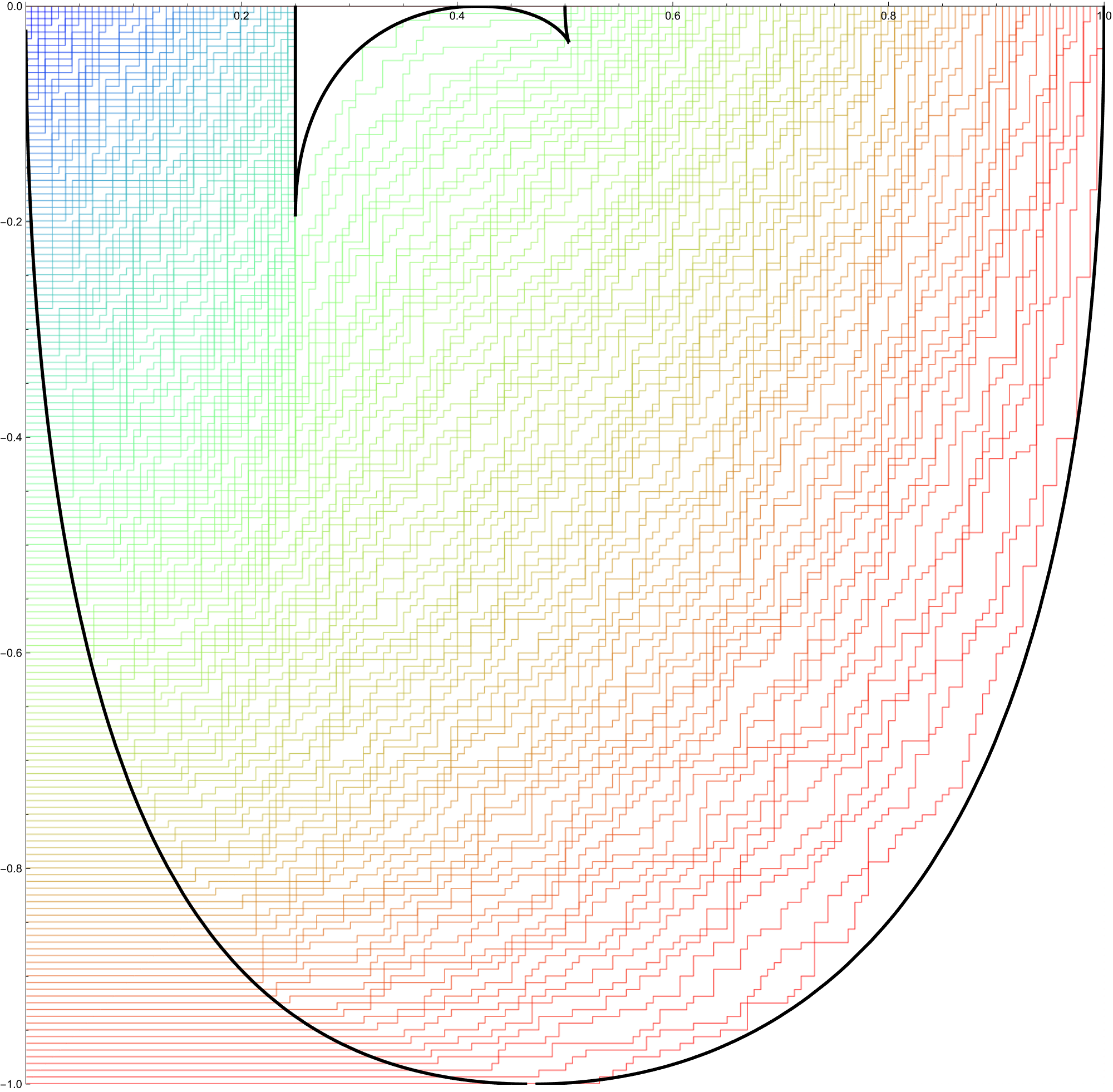}}}$
    \end{tabular}
    
    }
    \caption{Examples of boundary conditions with frozen regions. Left: Boundary conditions with $\frac{1}{16}$ of the paths exiting in the leftmost column, $\frac{1}{4}$ exiting in the rightmost column, and single path exiting at every other column. Center: Boundary conditions as in Section \ref{sec:oneclump}. Right: Boundary conditions as in Section \ref{sec:clumpgap}. In all cases we have $t=q=x_1=\ldots=x_n=1$ and have superimposed the theoretical arctic curve.}
    \label{fig:Clumpsim}
\end{figure}

\begin{figure}
    \centering
    \begin{tabular}{cc}
    $\vcenter{\hbox{\includegraphics[height=0.3\textwidth]{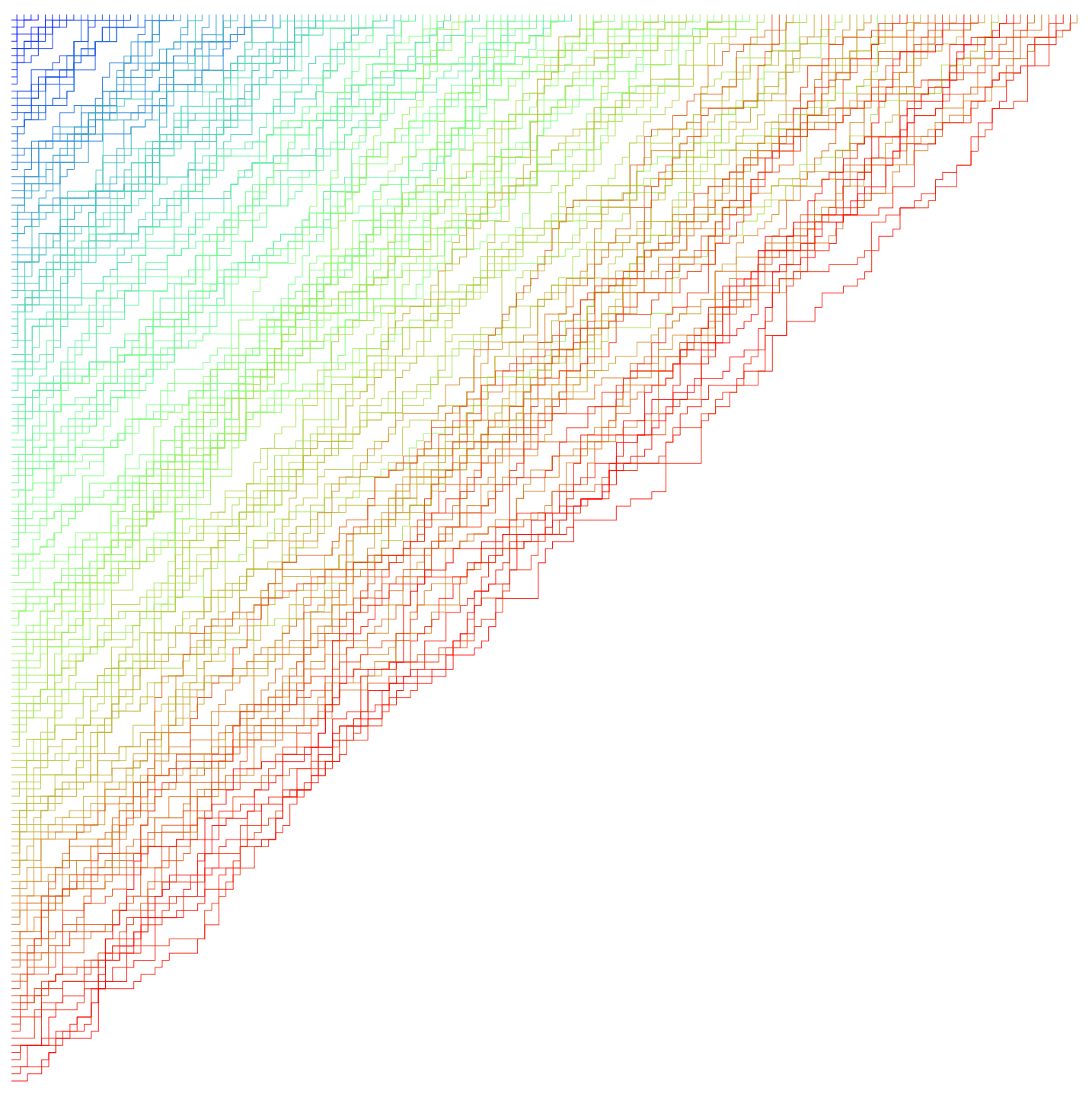}}}$ 
    & 
    $\vcenter{\hbox{\includegraphics[height=0.3\textwidth]{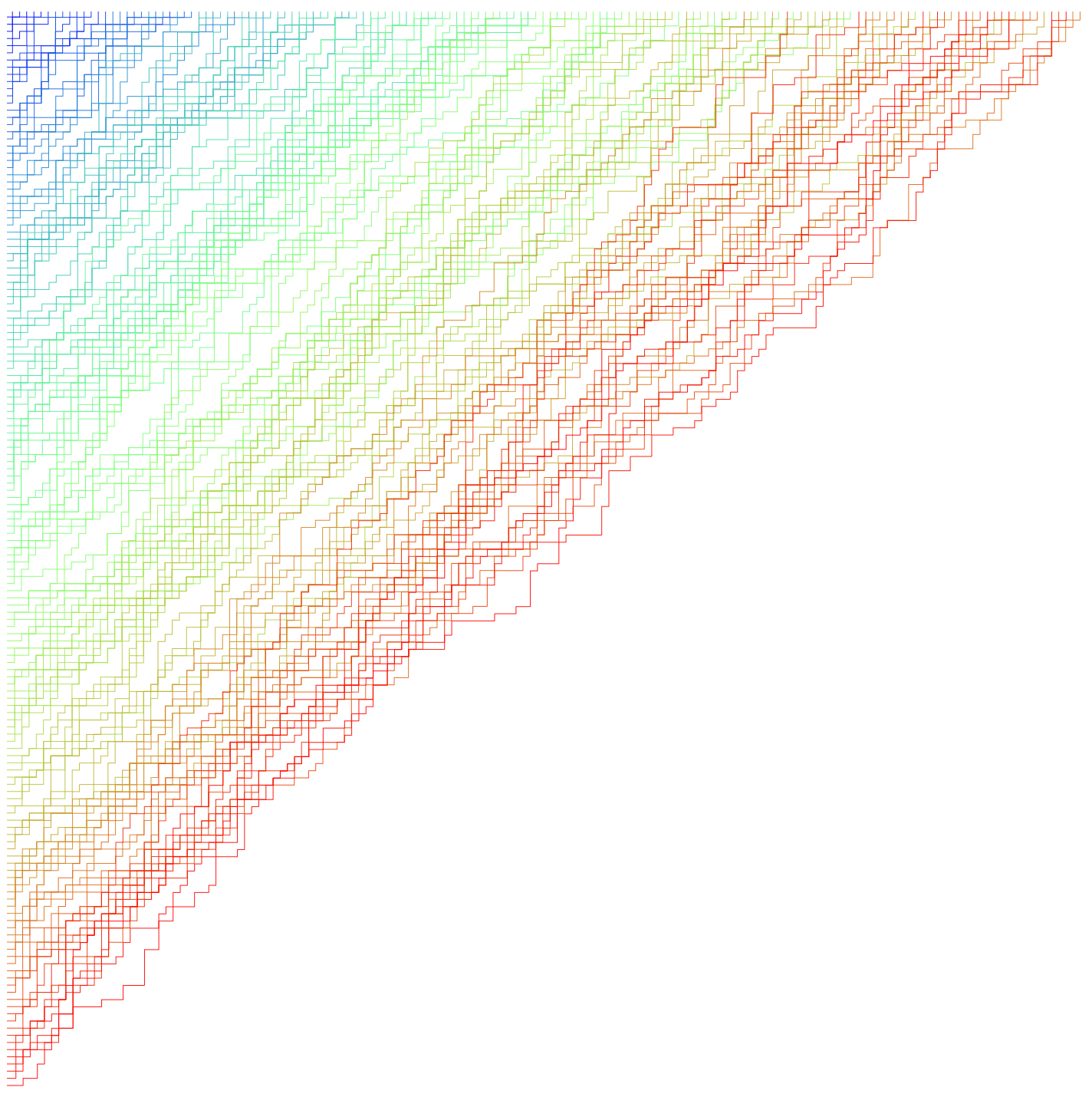}}}$ 
    \\
    $\tau=1$ & $\tau=5$ 
    \\ \\
    $\vcenter{\hbox{\includegraphics[height=0.3\textwidth]{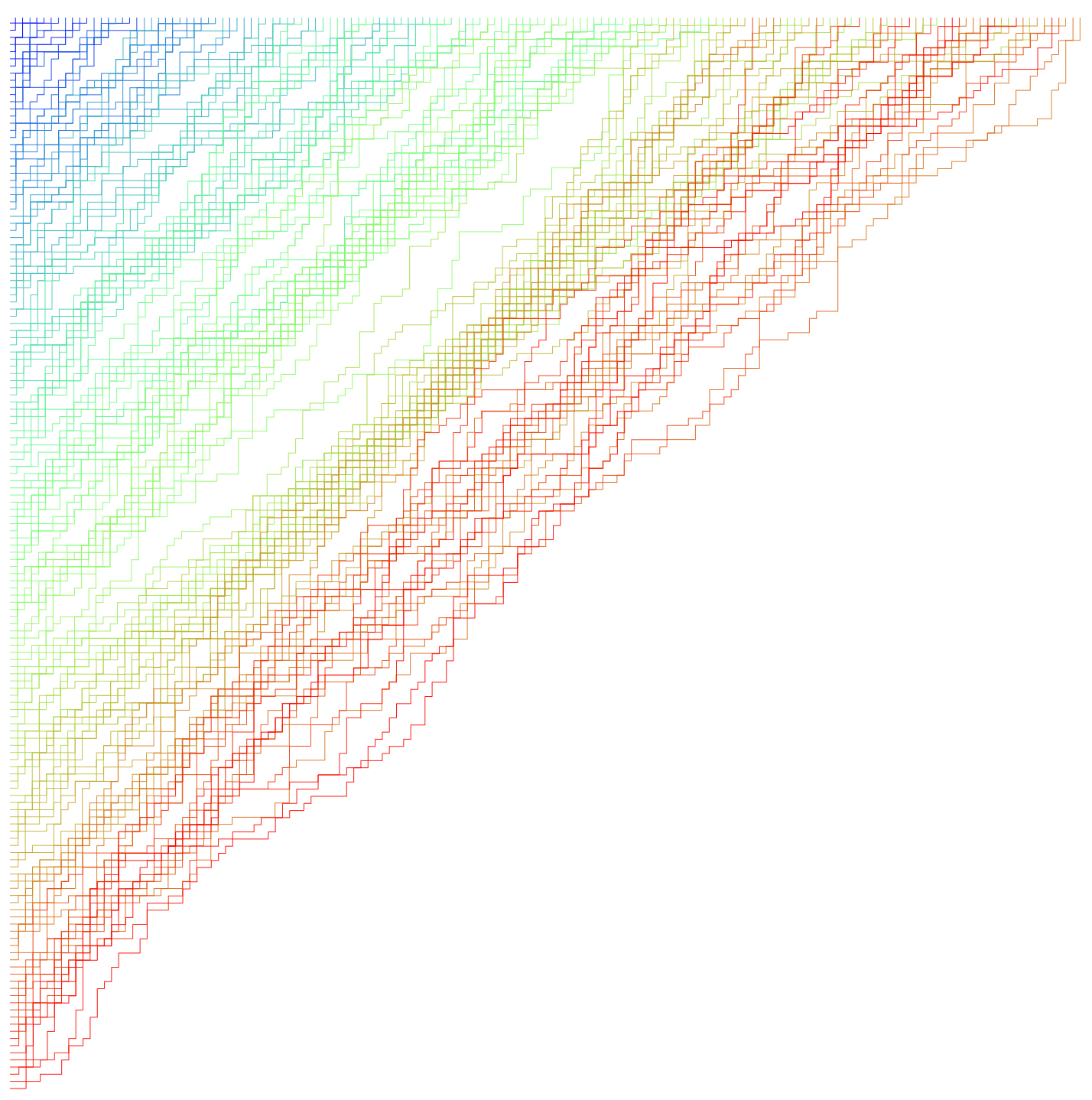}}}$ 
    &
    $\vcenter{\hbox{\includegraphics[height=0.3\textwidth]{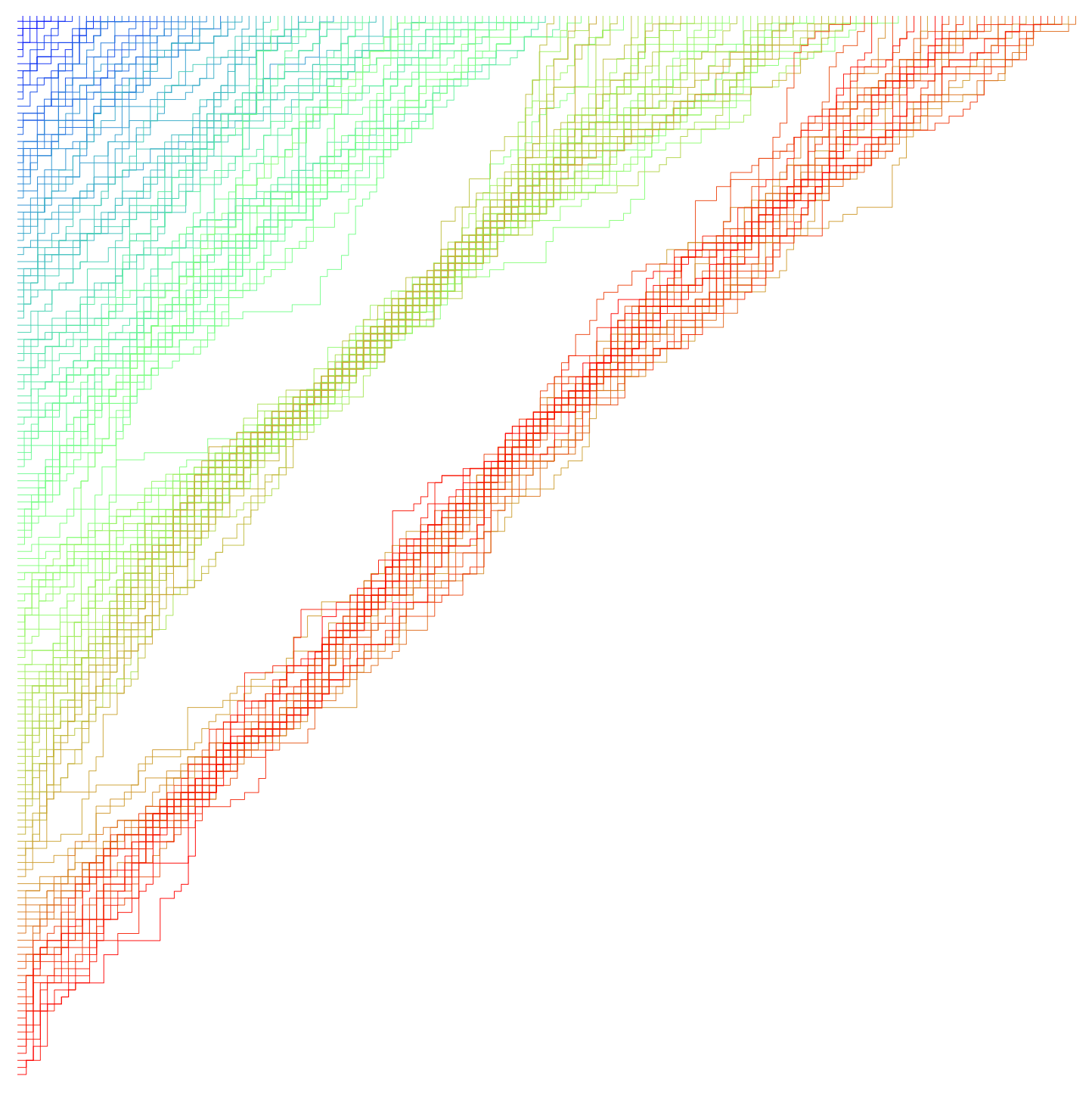}}}$ 
    \\
    $\tau=10$ & $\tau=20$ 
    \end{tabular}
    \caption{Simulations with colored DWBC with $n=150$ and free-coloring on the top boundary. The $t$-weight is given by $t=\tau^{1/n}$. We set $q=x_1=\ldots=x_n=1.$ }
    \label{fig:tnear1}
\end{figure}

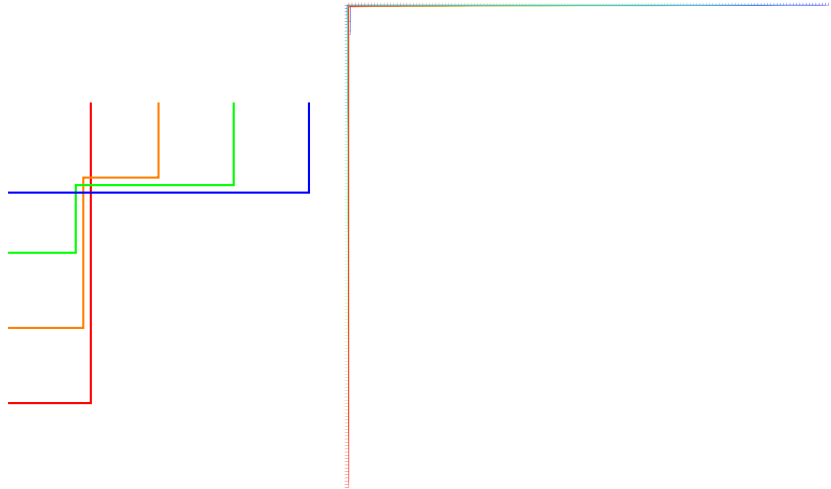
\begin{figure}
    \centering
    \begin{tabular}{cc}
        \begin{tikzpicture}[baseline=(current bounding box).center]
        \draw[red, thick] (0,0)--(1.1,0)--(1.1,4);
        \draw[orange, thick] (0,1)--(1,1)--(1,3)--(2,3)--(2,4);
        \draw[green, thick] (0,2)--(0.9,2)--(0.9,2.9)--(3,2.9)--(3,4);
        \draw[blue, thick] (0,2.8)--(4,2.8)--(4,4);
        \end{tikzpicture}
         &
         $\vcenter{\hbox{\includegraphics[height=0.4\textwidth]{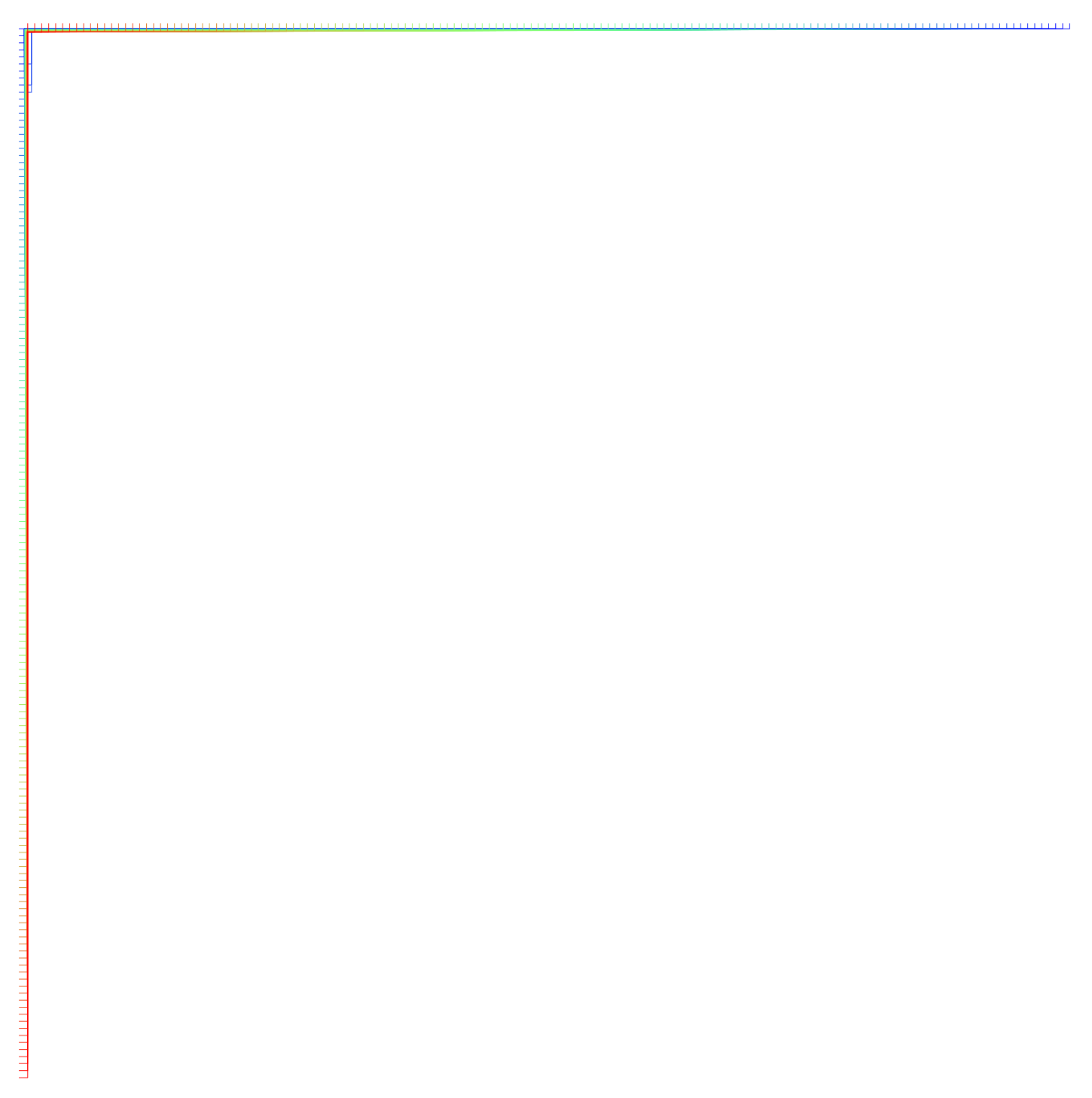}}}$ 
    \end{tabular}
    \caption{Left: The configuration with the maximum power of $t$ in the weight for $n=4$. On can check that the power of $t$ in the weight is $\binom{n+1}{3}=10$. Right: Simulations with colored DWBC with $n=150$ and free-coloring on the top boundary. We set $t=2$ and $q=x_1=\ldots=x_n=1$.}
    \label{fig:t>1}
\end{figure}

\begin{figure}
    \centering
    \begin{tabular}{cc}
        $\vcenter{\hbox{\includegraphics[height=0.4\textwidth]{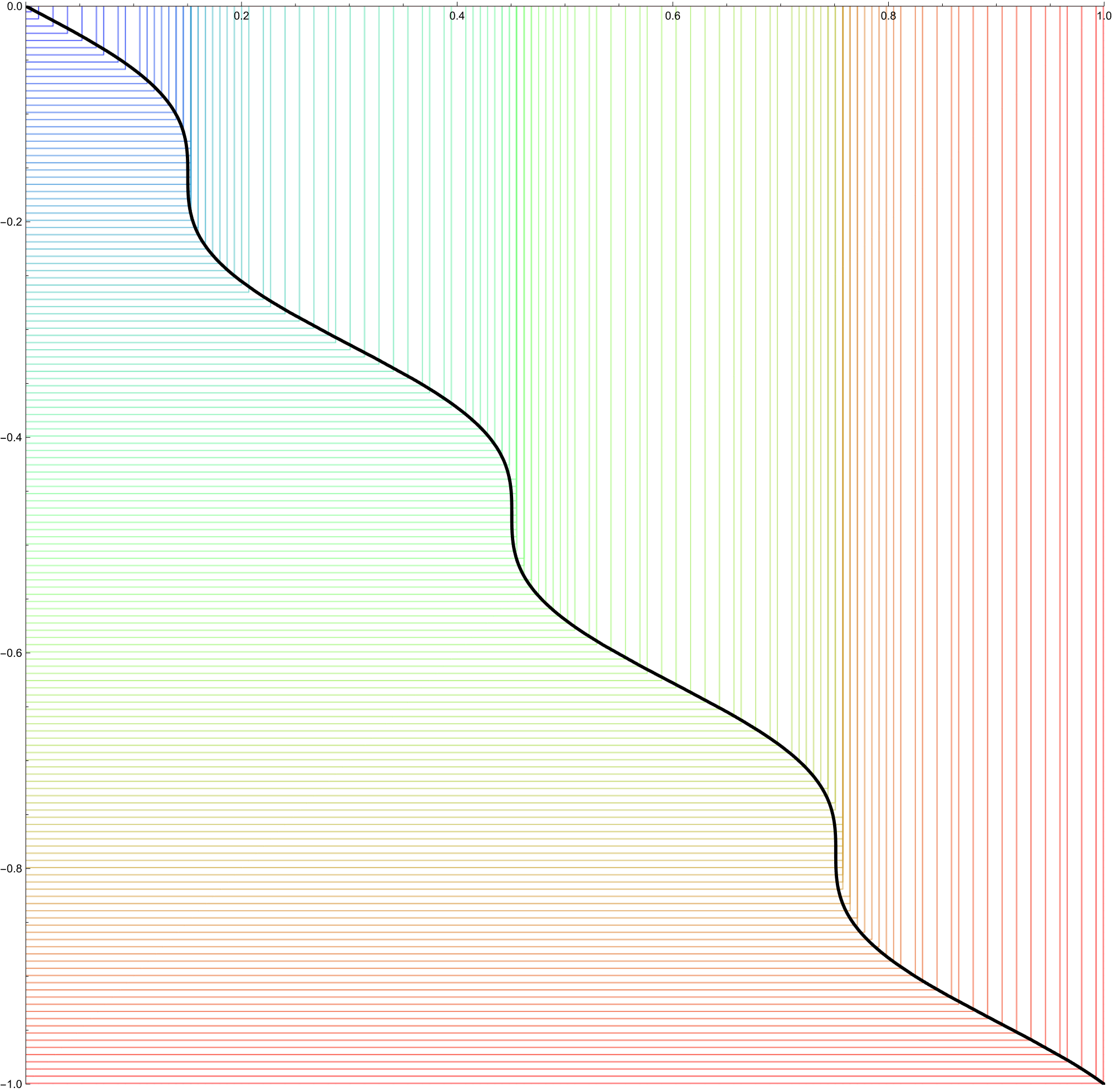}}}$ 
         &
        $\vcenter{\hbox{\includegraphics[height=0.4\textwidth]{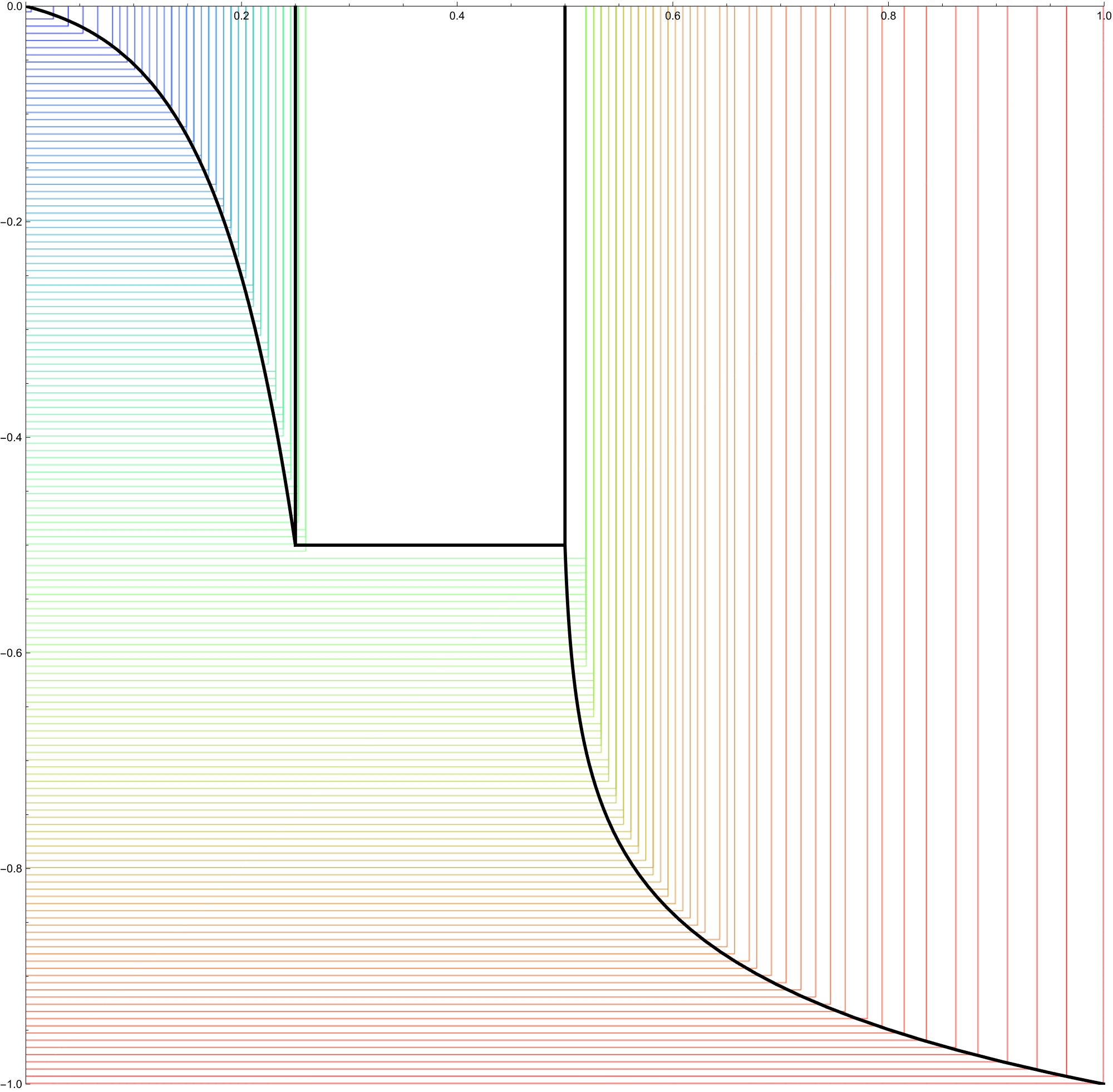}}}$ 
    \end{tabular}
    \caption{Examples of the tropical limit $q\to\infty$.  On the left we have a path configuration whose top boundary approximates the density $\alpha(u) = \frac{400u+20\sin(20u)}{400+\sin(20)}$ with the theoretical arctic curve superimposed.  On the right we have a path configuration whose top boundary approximates the density $\alpha(u) =  \protect\begin{cases} \frac{\log(50u+1)}{\log(26)}, & 0\le u <\frac{1}{2} \\ \frac{e^{10u-5}+e^5-2}{2(e^5-1)}, & \frac{1}{2} \le u \le 1 \protect\end{cases}$ with the theoretical arctic curve superimposed. In both cases, $t=1=x_1=\ldots=x_n=1$. }
    \label{fig:qinfinity}
\end{figure}

\subsection{Discussion}
\medskip
In this section, we present a discussion of the model for values of $t$ away from zero. It is in large part motivated by the simulations in the previous section.

While in the current work we are able to compute the arctic curves at $t=0$, in simulations, for example those in Figures \ref{fig:DWBCsim} and \ref{fig:Gapsim}, we still observe a phase separation between frozen and disordered regions for values of $0<t<1$.  However, for values of $t\ne 0$ the model can no longer be rendered free fermionic. The tangent method has been shown to work in some non-free-fermionic models, such as the the 6- and 20-vertex models, and it would be interesting to apply it to these colored vertex models, especially as there is little known about limit shapes for colored vertex models at the time of this writing.

When $t=1$ paths are independent of each-other. For large $n$, one expects the Gibbs measure to concentrate around the configurations that maximize the entropy. In this case, the coloring corresponding to the identity permutation maximizes the length of the path and each path will be near a straight line between its start and end point.

For values of $t>1$, it is useful to consider the path configuration with maximum power of $t$ in its weights. This configuration has boundary condition determined by the permutation $\sigma=(n,n-1,\ldots,1)$ and has all paths immediately travel north as far as possible then turn right and travel east until they reach their endpoint. One can check that the $t$-weight is given by $t^{\binom{n+1}{3}}\sim t^{n^3/6}$. For fixed $t>1$ and large $n$, the configurations near this maximal configuration dominate as seen in Figure \ref{fig:t>1}.

Here we give a heuristic argument on why an intermediate regime between the behavior with $t<1$ and $t>1$ can be seen when $t>1$ but approaches 1 as we take $n\to \infty$. Note that for a fixed choice of boundary coloring one expects the number of configurations to grow like ${\rm const}^{n^2}$. Consider\ a simultaneous limit in which we take the number of paths $n\to\infty$ while also taking $t=\tau^{1/n}$ for some $\tau>1$. In this case, configurations with $t$-weight $\sim t^{n^3/6}=\tau^{n^2/6}$ no longer dominate. That is, the configurations with larger powers of $t$ in their weight are now balanced with the high entropy configurations. In the simulations shown in Figure \ref{fig:tnear1}, one observes an interesting behavior in which the boundary color permutation is near the identity, the paths are near lines going from their start to end points, but they seem to form ``tendrils" in which nearby paths clump together increasing the power of $t$ in the weight.

Finally let us consider the path model with the area weight $q$, in the tropical limit $q\to \infty$.
We note that if we fix\footnote{We could also choose fixed $t\ne 0$. Since $q\to\infty$, the tendency to maximize the area of the paths overcomes the weight of the interactions between paths.} $t=0$, then in the tropical limit $q\to \infty$ there is a simple relationship between the boundary path density $\alpha(u)$ and the arctic curve. As previously mentioned in Section \ref{sec:qweights}, for boundary conditions in which the path density is piecewise linear (where we consider jumps in the density as segments with infinite slope), the arctic curve is also piecewise linear in the tropical limit. In fact, the slopes of the line segments in the arctic curve are the negative reciprocals of those of the density. In general, one can check that in this limit the arctic curve can be parametrized as
\[
\begin{cases}
    X(t) = \alpha(t) \\
    Y(t) = -t
\end{cases}, \;\; t\in[0,1].
\]
Figure \ref{fig:qinfinity} shows two examples of this in which $\alpha$ is no longer piecewise linear but instead a piecewise smooth function.
 
\bibliographystyle{plain}
\bibliography{coloredDWBC.bib}

\end{document}